\documentclass[%
a4paper,
UKenglish,
cleveref,
autoref,
thm-restate]%
{lipics-v2021}

\pdfoutput=1 
\hideLIPIcs  

\title{On the Runtime of Chemical Reaction Networks Beyond Idealized
Conditions%
\ShortVersion 
\ (Extended Abstract)
\ShortVersionEnd 
} 

\titlerunning{On the Runtime of CRNs Beyond Idealized Conditions%
\ShortVersion 
\ (Extended Abstract)
\ShortVersionEnd 
} 

\author{Anne Condon}{University of British Columbia, Vancouver, British Columbia, Canada}{condon@cs.ubc.ca}{}{}

\author{Yuval Emek}{Technion --- Israel Institute of Technology, Haifa, Israel}{yemek@technion.ac.il}{}{}

\author{Noga Harlev}{Technion --- Israel Institute of Technology, Haifa, Israel}{snogazur@campus.technion.ac.il}{}{}

\authorrunning{A. Condon, Y Emek, and N. Harlev} 

\Copyright{Anne Condon, Yuval Emek, and Noga Harlev} 



\category{} 

\relatedversion{} 

\nolinenumbers 

\EventEditors{John Q. Open and Joan R. Access}
\EventNoEds{2}
\EventLongTitle{42nd Conference on Very Important Topics (CVIT 2016)}
\EventShortTitle{CVIT 2016}
\EventAcronym{CVIT}
\EventYear{2016}
\EventDate{December 24--27, 2016}
\EventLocation{Little Whinging, United Kingdom}
\EventLogo{}
\SeriesVolume{42}
\ArticleNo{23}

\EventEditors{Ho-Lin Chen and Constantine G. Evans}
\EventNoEds{2}
\EventLongTitle{29th International Conference on DNA Computing and Molecular Programming (DNA 29)}
\EventShortTitle{DNA 29}
\EventAcronym{DNA}
\EventYear{29}
\EventDate{September 11-15, 2023}
\EventLocation{Tohoku University, Sendai, Japan}
\EventLogo{}
\SeriesVolume{276}
\ArticleNo{3}

\usepackage{ifthen}

\ifthenelse{\isundefined{\GenerateShortVersion}}
{
\def\LongVersion{}
\def\LongVersionEnd{}
\long\def\ShortVersion#1\ShortVersionEnd{}
}
{
\def\ShortVersion{}
\def\ShortVersionEnd{}
\long\def\LongVersion#1\LongVersionEnd{}
}

\usepackage{bbm}

\usepackage{multirow}

\usepackage{caption}
\usepackage{subcaption}

\usepackage{framed}

\theoremstyle{definition}

\newcommand{\Integers}{\mathbb{Z}}
\newcommand{\Naturals}{\mathbb{N}}
\newcommand{\Species}{\mathcal{S}}
\newcommand{\Reactions}{\mathcal{R}}
\newcommand{\NonVoid}{\operatorname{NV}}
\renewcommand{\c}{\mathbf{c}}
\newcommand{\Applicable}{\operatorname{app}}
\newcommand{\Stop}{\tau}
\newcommand{\TemporalCost}{\operatorname{TC}}
\newcommand{\Ex}{\mathbb{E}}
\renewcommand{\Pr}{\mathbb{P}}
\newcommand{\RunTime}{\operatorname{RT}}
\newcommand{\Stab}{\mathrm{stab}}
\newcommand{\Halt}{\mathrm{halt}}
\newcommand{\Indicator}{\mathbbm{1}}
\newcommand{\Sign}{\operatorname{sign}}
\newcommand{\PermVoters}{\mathcal{P}}
\newcommand{\FluidVoters}{\mathcal{F}}
\newcommand{\OR}{\operatorname{OR}}
\newcommand{\RW}{\mathrm{ra}}
\newcommand{\Amp}{\mathrm{amp}}
\newcommand{\Prp}{\pi}
\newcommand{\Vol}{\varphi}
\newcommand{\Ignt}{\mathrm{ignt}}

\newcommand{\LeadsTo}{\rightharpoonup}
\newcommand{\Reaching}{\stackrel{\ast}{\rightharpoonup}}

\newcommand{\SkipPol}{\sigma}
\newcommand{\RunPol}{\varrho}
\newcommand{\InitStep}{\mathit{t}}
\newcommand{\EffStep}{\InitStep_{\mathrm{e}}}
\newcommand{\EffConf}{\mathbf{e}}
\newcommand{\Identity}{\mathrm{id}}

\newcommand{\Sect}{Sec.}
\newcommand{\Thm}{Thm.}
\newcommand{\Lem}{Lem.}
\newcommand{\Prop}{Prop.}
\newcommand{\Obs}{Obs.}
\newcommand{\Cor}{Cor.}
\newcommand{\Fig}{Fig.}

\begin{document}

\maketitle

\begin{abstract}
This paper studies the (discrete) \emph{chemical reaction network (CRN)}
computational model that emerged in the last two decades as an abstraction for
molecular programming.
The correctness of CRN protocols is typically established under one of two
possible schedulers that determine how the execution advances:
(1)
a \emph{stochastic scheduler} that obeys the (continuous time) Markov process
dictated by the standard model of stochastic chemical kinetics;
or
(2)
an \emph{adversarial scheduler} whose only commitment is to maintain a certain
fairness condition.
The latter scheduler is justified by the fact that the former one crucially
assumes ``idealized conditions'' that more often than not, do not hold in real
wet-lab experiments.
However, when it comes to analyzing the \emph{runtime} of CRN protocols, the
existing literature focuses strictly on the stochastic scheduler, thus raising
the research question that drives this work:
Is there a meaningful way to quantify the runtime of CRNs without the
idealized conditions assumption?

The main conceptual contribution of the current paper is to answer this
question in the affirmative, formulating a new runtime measure for CRN
protocols that does not rely on idealized conditions.
This runtime measure is based on an adapted (weaker) fairness condition as
well as a novel scheme that enables partitioning the execution into short
\emph{rounds} and charging the runtime for each round individually (inspired
by definitions for the runtime of asynchronous distributed algorithms).
Following that, we turn to investigate various fundamental computational tasks
and establish (often tight) bounds on the runtime of the corresponding CRN
protocols operating under the adversarial scheduler.
This includes an almost complete chart of the runtime complexity landscape of
predicate decidability tasks.
\end{abstract}

\section{Introduction}
\label{section:introduction}
Chemical reaction networks (CRNs) are used to describe the evolution of
interacting molecules in a solution \cite{Gillespie1977stochastic} and more
specifically, the behavior of regulatory networks in the cell
\cite{BolouriB2001book}.
In the last two decades, CRNs have also emerged as a computational model for
molecular programming
\cite{SoloveichikCWB2008computation,
CookSWB2009programmability}.
A protocol in this model is specified by a set of species and a set of
reactions, which consume molecules of some species and produce molecules of
others.
In a (discrete) CRN computation, inputs are represented as (integral)
molecular counts of designated species in the initial system configuration;
a sequence of reactions ensues, repeatedly transforming the configuration,
until molecular counts of other designated species represent the output.
The importance of CRNs as a model of computation is underscored by the wide
number of closely related models, including
population protocols
\cite{AngluinADFP2006computation,
AngluinAER2007computational},
Petri nets
\cite{Rozenberg1998},
and
vector addition systems
\cite{Karp-etal-69}.\footnote{%
To simplify the discussions, we subsequently stick to the CRN terminology even
when citing literature that was originally written in terms of these related
models.
}

The standard model of stochastic chemical kinetics
\cite{Gillespie1977stochastic}, referred to hereafter as the \emph{standard
stochastic model}, dictates that the execution of a CRN protocol (operating
under fixed environmental conditions) advances as a continuous time Markov
process, where the rate of each reaction is determined by the molecular counts
of its reactants as well as a reaction specific rate coefficient.
This model crucially assumes that the system is ``well-mixed'', and so any
pair of distinct molecules is equally likely to interact, and that the rate
coefficients remain fixed.\footnote{%
We follow the common assumption that each reaction involves at most two
reactants.}
Under the standard stochastic model, CRNs can simulate Turing machines if a
small error probability is tolerated \cite{CookSWB2009programmability}.
The correctness of some protocols, including Turing machine simulations,
depends sensitively on the ``idealized conditions'' of fixed rate coefficients
and a well-mixed system.

However, correctness of many other CRN protocols, such as those which stably
compute predicates and functions
\cite{AngluinADFP2006computation,
AngluinAER2007computational,
ChenDS2014deterministic,
DotyH2015leaderless,
Brijder2016minimal,
ChenCDS2017speed-faults},
is premised on quite different assumptions:
correct output should be produced on \emph{all} ``fair executions'' of the
protocol, which means that the correctness of these protocols does not depend
on idealized conditions.
These protocols operate under a notion of fairness, adopted originally in
\cite{AngluinADFP2006computation}, requiring that reachable configurations are
not starved;
in the current paper, we refer to this fairness notion as \emph{strong
fairness}.
A celebrated result of Angluin et al.\
\cite{AngluinADFP2006computation,
AngluinAER2007computational}
states that with respect to strong fairness, a predicate can be decided by a
CRN if and only if it is semilinear.

As the
``what can be computed by CRNs?''\
question reaches a conclusion, the focus naturally shifts to its
``how fast?''\
counterpart.
The latter question is important as the analysis of CRN runtime complexity
enables the comparison between different CRN protocols and ultimately guides
the search for better ones.
Even for CRNs designed to operate on all (strongly) fair executions, the
existing runtime analyses assume that reactions are \emph{scheduled
stochastically}, namely, according to the Markov process of the standard
stochastic model, consistent with having the aforementioned idealized
conditions.
However, such conditions may well not hold in real wet-lab experiments, where
additional factors can significantly affect the order at which reactions
proceed \cite{VasicCLKS2022programming}.
For example, temperature can fluctuate, or molecules may be temporarily
unavailable, perhaps sticking to the side of a test tube or reversibly binding
to another reactant.
Consequently, our work is driven by the following research question:
\emph{Is there a meaningful way to quantify the runtime of CRNs when idealized
conditions do not necessarily hold?}

\subparagraph{The Quest for an Adversarial Runtime Measure.}
We search for a runtime measure suitable for \emph{adversarially scheduled}
executions, namely, executions that are not subject to the constraints of the
aforementioned idealized conditions.
This is tricky since the adversarial scheduler may generate (arbitrarily) long
execution intervals during which no progress can be made, even if those are
not likely to be scheduled stochastically.
Therefore, the ``adversarial runtime measure'' should neutralize the devious
behavior of the scheduler by ensuring that the protocol is not unduly
penalized from such bad execution intervals.
To guide our search, we look for inspiration from another domain of
decentralized computation that faced a similar challenge:
distributed network algorithms.

While it is straightforward to measure the runtime of (idealized) synchronous
distributed protocols, early on, researchers identified the need to define
runtime measures also for (adversarially scheduled) asynchronous distributed
protocols
\cite{Awerbuch85,
Dolev-etal-97}.
The adversarial runtime measures that were formulated in this regard share the
following two principles:
(P1)
partition the execution into \emph{rounds}, so that in each round, the
protocol has an opportunity to make progress;
and
(P2)
calibrate the runtime charged to the individual rounds so that if the
adversarial scheduler opts to generate the execution according to the
idealized benchmark, then the adversarial runtime measure coincides with the
idealized one.

Specifically, in the context of asynchronous message passing protocols,
Awerbuch \cite{Awerbuch85} translates principle (P1) to the requirement that
no message is delayed for more than a single round, whereas in the context of
\LongVersion 
self-stabilizing protocols controlled by
\LongVersionEnd 
the distributed daemon, Dolev et al.~\cite{Dolev-etal-97} translate this
principle to the requirement that each node is activated at least once in
every round.
For principle (P2), both
\LongVersion 
Awerbuch and Dolev et al.\
\LongVersionEnd 
\ShortVersion 
\cite{Awerbuch85} and \cite{Dolev-etal-97}
\ShortVersionEnd 
take the ``idealized benchmark'' to be a synchronous execution in which every
round costs one time unit.

When it comes to formulating an adversarial runtime measure for CRN protocols,
principle (P2) is rather straightforward:
we should make sure that on stochastically generated executions (playing the
``idealized benchmark'' role), the adversarial runtime measure agrees (in
expectation) with that of the corresponding continuous time Markov process.
Interpreting principle (P1), however, seems more difficult as it is not clear
how to partition the execution into rounds so that in each round, the protocol
``has an opportunity to make progress''.

The first step towards resolving this difficulty is to introduce an
alternative notion of fairness, referred to hereafter as \emph{weak fairness}:
An execution is weakly fair if a continuously applicable reaction (i.e., one
for which the needed reactants are available) is not starved;
such a reaction is either eventually scheduled or the system reaches a
configuration where the reaction is inapplicable.
Using a graph theoretic characterization, we show that any CRN protocol whose
correctness is guaranteed on weakly fair executions is correct also
on strongly fair executions (see
\Cor{}~\ref{corollary:weak-correctness-implies-strong-correctness}), thus
justifying the weak vs.\ strong terminology choice.
It turns out that for predicate decidability, strong fairness is actually not
strictly stronger:
protocols operating under the weak fairness assumption can decide all (and
only) semilinear predicates (see \Thm{}~\ref{theorem:crd:semilinear}).

It remains to come up with a scheme that partitions an execution of CRN
protocols into rounds in which the weakly fair adversarial scheduler can steer
the execution in a nefarious direction, but also the protocol has an
opportunity to make progress.
A naive attempt at ensuring progress would be to end the current round once
every applicable reaction is either scheduled or becomes inapplicable;
the resulting partition is too coarse though since in general, a CRN
\LongVersion 
protocol
\LongVersionEnd 
does not have to ``wait'' for \emph{all} its applicable reactions
in order to make progress.
Another naive attempt is to end the current round once \emph{any} reaction is
scheduled;
this yields a partition which is too fine, allowing the
\LongVersion 
adversarial
\LongVersionEnd 
scheduler to charge the protocol's run-time for (arbitrarily many)
``progress-less rounds''.

So, which reaction is necessary for the CRN protocol to make progress?
We do not have a good answer for this question, but we know who does\dots

\subparagraph{Runtime and Skipping Policies.}
Our adversarial runtime measure is formulated so that it is the protocol
designer who decides which reaction is necessary for the CRN protocol to make
progress. This is done by means of a \emph{runtime policy} $\RunPol$, used
solely for the runtime analysis, that maps each configuration $\c$ to a
\emph{target} reaction $\RunPol(\c)$.
(Our actual definition of runtime policies is more general, mapping each
configuration to a set of target reactions;
see \Sect{}~\ref{section:runtime}.)
Symmetrically to the protocol designer's runtime policy, we also introduce a
\emph{skipping policy} $\SkipPol$, chosen by the adversarial scheduler,
that maps each step
$t \geq 0$
to a step
$\SkipPol(t) \geq t$.

These two policies partition a given execution $\eta$ into successive rounds
based on the following inductive scheme:
Round $0$ starts at step
$\InitStep(0) = 0$.
Assuming that round
$i \geq 0$
starts at step $\InitStep(i)$, the prefix of round $i$ is determined by the
adversarial skipping policy $\SkipPol$ so that it lasts until step
$\SkipPol(\InitStep(i))$;
let $\EffConf^{i}$ denote the configuration in step $\SkipPol(\InitStep(i))$,
referred to as the round's \emph{effective configuration}.
Following that, the suffix of round $i$ is determined by the protocol
designer's runtime policy $\RunPol$ so that it lasts until the earliest step
in which the target reaction $\RunPol(\EffConf^{i})$ of the round's effective
configuration $\EffConf^{i}$ is either scheduled or becomes inapplicable.
That is, in each round, the adversarial scheduler determines (by means
of the skipping policy) the round's effective configuration, striving to
ensure that progress from this configuration is slow, whereas the runtime
policy determines when progress has been made from the effective
configuration.
This scheme is well defined by the choice of weak fairness;
we emphasize that this would not be the case with strong fairness.

The partition of execution $\eta$ into rounds allows us to ascribe a runtime
to $\eta$ by charging each round with a \emph{temporal cost} and then
accumulating the temporal costs of all rounds until $\eta$
terminates.\footnote{%
The exact meaning of termination in this regard is made clear in
\Sect{}~\ref{section:model}.}
The temporal cost of round $i$ is defined to be the expected (continuous) time
until the target reaction $\RunPol(\EffConf^{i})$ of its effective
configuration $\EffConf^{i}$ is either scheduled or becomes inapplicable in an
imaginary execution that starts at $\EffConf^{i}$ and proceeds according to
the stochastic scheduler.\footnote{%
Here, it is assumed that the stochastic scheduler operates with no rate
coefficients and with a linear volume (a.k.a.\ ``parallel time''), see
\Sect{}~\ref{section:model}.}
In other words, the protocol's runtime is \emph{not} charged for the prefix of
round $i$ that lasts until the (adversarially chosen) effective configuration
is reached;
the temporal cost charged for the round's suffix, emerging from the effective
configuration, is the expected time that this suffix would have lasted in a
stochastically scheduled execution (i.e., the idealized benchmark).

The asymptotic runtime of the CRN protocol is defined by minimizing over all
runtime policies $\RunPol$ and then maximizing over all weakly fair executions
$\eta$ and skipping policies $\SkipPol$.
Put differently, the protocol designer first commits to $\RunPol$ and only
then, the (weakly fair) adversarial scheduler determines $\eta$ and
$\SkipPol$.

Intuitively, the challenge in constructing a good runtime policy $\RunPol$ (the
challenge one faces when attempting to up-bound a protocol's runtime) is
composed of two, often competing, objectives (see, e.g.,
\Fig{}~\ref{figure:crn-kill-b}):
On the one hand, $\RunPol(\c)$ should be selected so that every execution
$\eta$ is guaranteed to gain ``significant progress'' by the time a round
whose effective configuration is $\c$ ends, thus minimizing the number of
rounds until $\eta$ terminates.
On the other hand, $\RunPol(\c)$ should be selected so that the temporal cost
of such a round is small, thus minimizing the contribution of the individual
rounds to $\eta$'s runtime.
In the typical scenarios, a good runtime policy $\RunPol$ results in
partitioning $\eta$ into
$n^{\Theta (1)}$
rounds, each contributing a temporal cost between
$\Theta (1 / n)$
and
$\Theta (n)$,
where $n$ is
\LongVersion 
the molecular count of $\eta$'s initial configuration (these scenarios include
``textbook examples'' such as the classic leader election and rumor spreading
protocols as well as all protocols presented in
\Sect{}\ \ref{section:runtime:tools}--\ref{section:crd:detection})%
\LongVersionEnd 
\ShortVersion 
$\eta$'s initial molecular count%
\ShortVersionEnd 
.

To verify that our adversarial runtime measure is indeed compatible with the
aforementioned principle (P2), we show that if the (adversarial) scheduler
opts to generate the execution $\eta$ stochastically, then our runtime measure
coincides (in expectation) with that of the corresponding continuous time
Markov process (see \Lem{}~\ref{lemma:stochastic-benchmark}).
The adversarial scheduler however can be more malicious than that:
simple examples show that in general, the runtime of a CRN protocol on
adversarially scheduled executions may be significantly larger than on
stochastically scheduled executions (see \Fig{}
\ref{figure:crn-multiple-pitfalls} and
\ref{figure:crn-kill-alternating-leaders}).

While runtime analyses of CRNs in the presence of common defect modes can be
insightful, a strength of our adversarial model is that it is not tied to
specific defects in actual CRNs or their biomolecular implementations.
In particular, if the adversarial runtime of a CRN matches its stochastic
runtime, then we would expect the CRN to perform according to its stochastic
runtime even in the presence of defect modes that we may not anticipate.
Moreover, in cases where stochastic runtime analysis is complex (involving
reasoning about many different executions of a protocol and their
likelihoods), it may in fact be easier to determine the adversarial runtime
since it only requires stochastic analysis from rounds' effective
configurations.
For similar reasons, notions of adversarial runtime have proven to be valuable
in design of algorithms in both centralized and decentralized domains more
broadly, even when they do not capture realistic physical scenarios.
Finally, while the analysis task of finding a good runtime policy for a given
CRN may seem formidable at first, our experience in analyzing the protocols in
this paper is that such a runtime policy is quite easy to deduce, mirroring
intuition about the protocol's strengths and weaknesses.

\subparagraph{The Runtime of Predicate Decidability.}
After formulating the new adversarial runtime measure, we turn our attention
to CRN protocols whose goal is to decide whether the initial configuration
satisfies a given predicate, indicated by the presence of designated Boolean (`yes' and `no')
\emph{voter} species in the output configuration.
As mentioned earlier, the predicates that can be decided in that way are
exactly the semilinear predicates, which raises the following two questions:
What is the optimal adversarial runtime of protocols that decide semilinear
predicates in general?
Are there semilinear predicates that can be decided faster?

A notion that plays an important role in answering these questions is that of
CRN \emph{speed faults}, introduced in the impressive work of Chen et
al.~\cite{ChenCDS2017speed-faults}.
This notion captures a (reachable) configuration from which any path to an
output configuration includes a (bimolecular) reaction both of whose reactants
appear in
$O (1)$
molecular counts.
The significance of speed faults stems from the fact that any execution that
reaches such a ``pitfall configuration'' requires
$\Omega (n)$
time (in expectation) to terminate under the standard stochastic
model.\footnote{%
The definition of runtime in \cite{ChenCDS2017speed-faults} is based on a
slightly different convention which results in scaling the runtime expressions
by a
$1 / n$
factor.}
The main result of \cite{ChenCDS2017speed-faults} states that a predicate can
be decided by a speed fault free CRN protocol (operating under the strongly
fair adversarial scheduler) if and only if it belongs to the class of
detection predicates (a subclass of semilinear predicates).

The runtime measure introduced in the current paper can be viewed as a
quantitative generalization of the fundamentally qualitative notion of speed
faults (the quest for such a generalization was, in fact, the main motivation
for this work).
As discussed in \Sect{}~\ref{section:runtime:speed-faults}, in our adversarial
setting, a speed fault translates to an
$\Omega (n)$
runtime lower bound, leading to an
$\Omega (n)$
runtime lower bound for the task of deciding any non-detection semilinear
predicate.
On the positive side, we prove that this bound is tight:
any semilinear predicate (in particular, the non-detection ones) can be
decided by a CRN protocol operating under the weakly fair adversarial
scheduler whose runtime is
$O (n)$
(see \Thm{}~\ref{theorem:crd:semilinear}).
For detection predicates, we establish a better upper bound (which is also
tight):
any detection predicate can be decided by a CRN protocol operating under the
weakly fair adversarial scheduler whose runtime is
$O (\log n)$
(see \Thm{}~\ref{theorem:crd:detection}).
Refer to
\LongVersion 
\Sect{}~\ref{section:related-work} for an additional discussion and to
\LongVersionEnd 
Table~\ref{table:predicate-decidability-landscape-adversarial} for a summary
of the adversarial runtime complexity bounds established for predicate
decidability tasks;
for comparison, Table~\ref{table:predicate-decidability-landscape-stochastic}
presents a similar summary of the known stochastic runtime complexity bounds.

\LongVersion 
\subparagraph{Amplifying the Voter Signal.}
By definition, a predicate deciding CRN protocol accepts (resp., rejects) a
given initial configuration by including $1$-voter (resp., $0$-voter) species
in the output configuration.
This definition merely requires that the ``right'' voter species are present
in the output configuration in a positive molecular count, even if this
molecular count is small.
In practice, the signal obtained from species with a small molecular count may
be too weak, hence we aim towards \emph{vote amplified} protocols, namely,
protocols with the additional guarantee that the fraction of non-voter
molecules in the output configuration is arbitrarily small.

To this end, we introduce a generic compiler that takes any predicate
decidability protocol and turns it into a vote amplified protocol.
The core of this compiler is a (standalone) computational task, referred to as
\emph{vote amplification}, which is defined over four species classes:
\emph{permanent} $0$- and $1$-voters and \emph{fluid} $0$- and $1$-voters.
A vote amplification protocol is correct if for
$v \in \{ 0, 1 \}$,
starting from any initial
configuration $\c^{0}$ with a positive molecular count of permanent $v$-voters
and no permanent
$(1 - v)$-voters,
the execution is guaranteed to terminate in a configuration that includes only
(permanent and fluid) $v$-voters;
this guarantee holds regardless of the molecular counts of the fluid voters in
$\c^{0}$.
As it turns out, the runtime of the vote amplification protocol is the
dominant component in the runtime overhead of the aforementioned compiler.

A vote amplification protocol whose runtime is
$O (n)$
is presented in \cite{AngluinAE2008fast} (using the term ``random-walk
broadcast''), however this protocol is designed to operate under the
stochastic scheduler and, as shown in Appendix~\ref{appendix:random-walk}, its
correctness breaks once we switch to the weakly fair adversarial scheduler.
One of the main technical contributions of the current paper is a vote
amplification protocol whose (adversarial) runtime is also
$O (n)$,
albeit under the weakly fair adversarial scheduler (see
\Thm{}~\ref{theorem:amplification:protocol}).
\LongVersionEnd 

\subparagraph{Paper's Outline.}
The rest of the paper is organized as follows.
The CRN model used in this paper is presented in \Sect{}~\ref{section:model}.
In \Sect{}~\ref{section:configuration-digraph}, we link the correctness of a
CRN protocol to certain topological properties of its configuration digraph.
Our new runtime notion for adversarially scheduled executions is introduced
in \Sect{}~\ref{section:runtime}, where we also establish the soundness of
this notion%
\LongVersion 
,
\LongVersionEnd 
\ShortVersion 
\ and
\ShortVersionEnd 
formalize its connection to speed faults%
\LongVersion 
, and provide a toolbox of useful techniques for protocol runtime analysis%
\LongVersionEnd 
.
\Sect{}~\ref{section:crd} presents our results for predicate deciding CRNs%
\LongVersion 
\ including
the protocols that decide semilinear and detection predicates%
\LongVersionEnd 
.
\LongVersion 
The generic vote amplification compiler is
\LongVersionEnd 
\ShortVersion 
These are accompanied by a generic technique for amplifying the molecular
count of the voter species in the outcome,
\ShortVersionEnd 
introduced in \Sect{}~\ref{section:amplification}.
\LongVersion 
In \Sect{}~\ref{section:justify-runtime-policy}, we consider four ``natural
restrictions'' for the definition of the runtime policy and show that they
actually lead to (asymptotic) inefficiency in terms of the resulting runtime
bounds.
\Sect{}~\ref{section:large-adversarial-small-stochastic} demonstrates that the
adversarial runtime of a CRN protocol may be significantly larger than its
expected runtime under the standard stochastic model.
We conclude in \Sect{}~\ref{section:related-work} with additional related work
and some open questions.
\LongVersionEnd 

\section{Chemical Reaction Networks}
\label{section:model}
In this section, we present the \emph{chemical reaction network (CRN)}
computational model.
For the most part, we adhere to the conventions of the existing CRN literature
(e.g.,
\cite{CookSWB2009programmability,
Condon2018design,
Brijder2019tutorial}),
but we occasionally deviate from them for the sake of simplifying the
subsequent discussions.
(Refer to \Fig{}\
\ref{figure:crn-kill-b:protocol}--\ref{figure:crn-skipping-policy:protocol}
for illustrations of the notions presented in this section.)

A \emph{CRN} is a protocol $\Pi$ specified by the pair
$\Pi = (\Species, \Reactions)$,
where
$\Species$ is a fixed set of \emph{species}
and
$\Reactions \subset \Naturals^{\Species} \times \Naturals^{\Species}$
is a fixed set of \emph{reactions}.\footnote{%
Throughout this paper, we denote
$\Naturals = \{ z \in \Integers \mid z \geq 0 \}$.}
For a reaction
$\alpha = (\mathbf{r}, \mathbf{p}) \in \Reactions$,
the vectors
$\mathbf{r} \in \Naturals^{\Species}$
and
$\mathbf{p} \in \Naturals^{\Species}$
specify the stoichiometry of $\alpha$'s \emph{reactants} and \emph{products},
respectively.\footnote{%
We stick to the convention of identifying vectors in $\Naturals^{\Species}$
with multisets over $\Species$ expressed as a ``molecule summation''.}
Specifically, the entry $\mathbf{r}(A)$ (resp., $\mathbf{p}(A)$) indexed by a
species
$A \in \Species$
in the vector
$\mathbf{r}$ (resp., $\mathbf{p}$) encodes the number of molecules of $A$ that
are consumed (resp., produced) when
$\alpha$ is applied.
Species $A$ is a \emph{catalyst} for the reaction
$\alpha = (\mathbf{r}, \mathbf{p})$
if
$\mathbf{r}(A) = \mathbf{p}(A) > 0$.

We adhere to the convention (see, e.g.,
\cite{ChenDS2014deterministic,
Doty2014timing,
CummingsDS2016probability-1,
ChenCDS2017speed-faults})
that each reaction
$(\mathbf{r}, \mathbf{p}) \in \Reactions$
is either \emph{unimolecular} with
$\| \mathbf{r} \| = 1$
or \emph{bimolecular} with
$\| \mathbf{r} \| = 2$%
\LongVersion 
;%
\LongVersionEnd 
\ShortVersion 
.%
\ShortVersionEnd 
\footnote{%
The notation
$\| \cdot \|$
denotes the $1$-norm $\ell_{1}$.}
\LongVersion 
forbidding higher order reactions is justified as more than two molecules are
not likely to directly interact.
\LongVersionEnd 
Note that if all reactions
$(\mathbf{r}, \mathbf{p}) \in \Reactions$
are bimolecular and \emph{density preserving}, namely,
$\| \mathbf{r} \| = \| \mathbf{p} \|$,
then the CRN model is equivalent to the extensively studied \emph{population
protocols} model
\cite{AngluinADFP2006computation,
AspnesR2009introduction,
MichailCS2011book}
assuming that the population protocol agents have a constant state space.

For a vector (or multiset)
$\mathbf{r} \in \Naturals^{\Species}$
with
$1 \leq \| \mathbf{r} \| \leq 2$,
let
$\Reactions(\mathbf{r})
=
( \{ \mathbf{r} \} \times \Naturals^{\Species} ) \cap \Reactions$
denote the subset of reactions whose reactants correspond to $\mathbf{r}$.
In the current paper, it is required that none of these reaction subsets is
empty, i.e.,
$|\Reactions(\mathbf{r})| \geq 1$
for every
$\mathbf{r} \in \Naturals^{\Species}$
with
$1 \leq \| \mathbf{r} \| \leq 2$.
Some of the reactions in $\Reactions$ may be \emph{void}, namely,
reactions
$(\mathbf{r}, \mathbf{p})$
satisfying
$\mathbf{r} = \mathbf{p}$;
let
$\NonVoid(\Reactions)
=
\{
(\mathbf{r}, \mathbf{p}) \in \Reactions
\mid
\mathbf{r} \neq \mathbf{p}
\}$
denote the set of non-void reactions in
$\Reactions$.
To simplify the exposition, we assume that if
$\alpha = (\mathbf{r}, \mathbf{r}) \in \Reactions$
is a void reaction, then
$\Reactions(\mathbf{r}) = \{ \alpha \}$;
this allows us to describe protocol $\Pi$ by listing only its non-void
reactions.
\LongVersion 
We
\LongVersionEnd 
\ShortVersion 
For the sake of simplicity, we
\ShortVersionEnd 
further assume that
$\| \mathbf{r} \| \leq \| \mathbf{p} \|$
for all reactions
$(\mathbf{r}, \mathbf{p}) \in \Reactions$.%
\LongVersion 
\footnote{%
The last two assumptions, are not fundamental to our CRN setup and are made
only for the sake of simplicity.}
\LongVersionEnd 

\subparagraph{Configurations.}
A \emph{configuration} of a CRN
$\Pi = (\Species, \Reactions)$
is a vector
$\c \in \Naturals^{\Species}$
that encodes the \emph{molecular count} $\c(A)$ of species $A$ in the solution
for each
$A \in \Species$.\footnote{%
\LongVersion 
Note that we
\LongVersionEnd 
\ShortVersion 
We
\ShortVersionEnd 
consider the discrete version of the CRN model, where the
configuration encodes integral molecular counts%
\LongVersion 
.
This is
\LongVersionEnd 
\ShortVersion 
,
\ShortVersionEnd 
in contrast to the continuous
\LongVersion 
CRN
\LongVersionEnd 
model, where a configuration is given by real species densities.}
The molecular count notation is extended to species (sub)sets
$\Lambda \subseteq \Species$,
denoting
$\c(\Lambda) = \sum_{A \in \Lambda} \c(A)$.
We refer to
$\c(\Species) = \| \c \|$
as the \emph{molecular count}
of the configuration $\c$.
Let
$\c|_{\Lambda} \in \Naturals^{\Lambda}$
denote the restriction of a configuration
$\c \in \Naturals^{\Species}$
to a species subset
$\Lambda \subseteq \Species$.

A reaction
$\alpha = (\mathbf{r}, \mathbf{p}) \in \Reactions$
is said to be \emph{applicable} to a configuration
$\c \in \Naturals^{\Species}$
if
$\mathbf{r}(A) \leq \c(A)$
for every
$A \in \Species$.
Let
$\Applicable(\c) \subseteq \Reactions$
denote the set of reactions which are applicable to $\c$ and let
$\overline{\Applicable}(\c) = \Reactions - \Applicable(\c)$,
referring to the reactions in
$\overline{\Applicable}(\c)$
as being \emph{inapplicable} to $\c$.
We restrict our attention to configurations $\c$ with molecular count
$\| \c \| \geq 1$,
which ensures that $\Applicable(\c)$ is never empty.
For a reaction
$\alpha \in \Applicable(\c)$,
let
$\alpha(\c)
=
\c - \mathbf{r} + \mathbf{p}$
be the configuration obtained by applying $\alpha$ to $\c$.\footnote{%
Unless stated otherwise, all vector arithmetic is done component-wise.}

Given two configurations
$\c, \c' \in \Naturals^{\Species}$,
the binary relation
$\c \LeadsTo \c'$
holds if there exists a reaction
$\alpha \in \Applicable(\c)$
such that
$\alpha(\c) = \c'$.
We denote the reflexive transitive closure of $\LeadsTo$ by
$\Reaching$ and say that $\c'$ is \emph{reachable} from $\c$ if
$\c \Reaching \c'$.
Given a configuration set
$Z \subseteq \Naturals^{\Species}$,
let
\[\textstyle
\Stab(Z)
\, \triangleq \,
\left\{ \c \in Z \mid \c \Reaching \c' \Longrightarrow \c' \in Z \right\}
\quad\text{and}\quad
\Halt(Z)
\, \triangleq \,
\left\{ \c \in Z \mid \c \Reaching \c' \Longrightarrow \c' = \c \right\}
\,
\LongVersion 
;
\LongVersionEnd 
\ShortVersion 
,
\ShortVersionEnd 
\]
\LongVersion 
that is,
$\Stab(Z)$ consists of every configuration
$\c \in Z$
all of whose reachable configurations are also in $Z$
whereas
$\Halt(Z)$ consists of every configuration
$\c \in Z$
which is halting in the sense that the only configuration reachable from $\c$
is $\c$ itself,
\LongVersionEnd 
observing that the latter set is a (not necessarily strict) subset of the
former.

For the sake of simplicity, we restrict this paper's focus to protocols that
\emph{respect finite density} \cite{Doty2014timing}, namely,
$\c \Reaching \c'$
implies that
$\| \c' \| \leq O (\| \c \|)$.%
\LongVersion 
\footnote{%
This restriction is not fundamental to our CRN setup and can be swapped for a
weaker one.}
\LongVersionEnd 
\ShortVersion 
\
\ShortVersionEnd 
We note that density preserving CRNs inherently respect finite density,
however we also allow for reactions that have more products than reactants as
long as the CRN protocol is designed so that the molecular count cannot
increase arbitrarily.
This means, in particular, that although the configuration space
$\Naturals^{\Species}$
is inherently infinite, the set
$\{ \c' \in \Naturals^{\Species} \mid \c \Reaching \c' \}$
is finite (and bounded as a function of
$\| \c \|$)
for any configuration
$\c \in \Naturals^{\Species}$.

\subparagraph{Executions.}
An \emph{execution} $\eta$ of the CRN $\Pi$ is an infinite sequence
$\eta = \langle \c^{t}, \alpha^{t} \rangle_{t \geq 0}$
of
$\langle \text{configuration}, \text{reaction} \rangle$
pairs such that
$\alpha^{t} \in \Applicable(\c^{t})$
and
$\c^{t + 1} = \alpha^{t}(\c^{t})$
for every
$t \geq 0$.
It is convenient to think of $\eta$ as progressing in discrete \emph{steps} so
that configuration $\c^{t}$ and reaction $\alpha^{t}$ are associated with step
$t \geq 0$.
We refer to $\c^{0}$ as the \emph{initial configuration} of $\eta$ and, unless
stated otherwise, denote the molecular count of $\c^{0}$ by
$n = \| \c^{0} \|$.
Given a configuration set
$Z \subseteq \Naturals^{\Species}$,
we say that $\eta$ \emph{stabilizes} (resp., \emph{halts}) into $Z$ if there
exists a step
$t \geq 0$
such that
$\c^{t} \in \Stab(Z)$
(resp.,
$\c^{t} \in \Halt(Z)$)
and refer to the earliest such step $t$ as the execution's \emph{stabilization
step} (resp., \emph{halting step}) with respect to $Z$.

In this paper, we consider an \emph{adversarial scheduler} that knows the
CRN protocol $\Pi$ and the initial configuration $\c^{0}$ and determines the
execution
$\eta = \langle \c^{t}, \alpha^{t} \rangle_{t \geq 0}$
in an arbitrary (malicious) way.
The execution $\eta$ is nonetheless subject to the following \emph{fairness}
condition:
for every
$t \geq 0$
and for every
$\alpha \in \Applicable(\c^{t})$,
there exists
$t' \geq t$
such that either
(I)
$\alpha^{t'} = \alpha$;
or
(II)
$\alpha \notin \Applicable(\c^{t'})$.
In other words, the scheduler is not allowed to (indefinitely) ``starve'' a
continuously applicable reaction.
We emphasize that the mere condition that a reaction
$\alpha \in \Reactions$
is applicable infinitely often does not imply that $\alpha$ is scheduled
infinitely often.

Note that the fairness condition adopted in the current paper differs from
the one used in the existing CRN
\LongVersion 
(and population protocols)
\LongVersionEnd 
literature
\cite{AngluinADFP2006computation,
AngluinAER2007computational,
ChenDS2014deterministic,
Brijder2019tutorial}.
The latter, referred to hereafter as \emph{strong fairness}, requires that if
a configuration $\c$ appears infinitely often in the execution $\eta$ and a
configuration $\c'$ is reachable from $\c$, then $\c'$ also appears infinitely
often in $\eta$.
Strictly speaking, a strongly fair execution $\eta$ is not necessarily fair
according to the current paper's notion of fairness (in particular, $\eta$ may
starve void reactions).
However, as we show in \Sect{}~\ref{section:configuration-digraph}, protocol
correctness under the current paper's notion of fairness implies protocol
correctness under strong fairness (see
\Cor{}~\ref{corollary:weak-correctness-implies-strong-correctness}), where the
exact meaning of correctness is defined soon.
Consequently, we refer hereafter to the notion of fairness adopted in the
current paper as \emph{weak fairness}.

\subparagraph{Interface and Correctness.}
The CRN notions introduced so far are independent of any particular
computational task.
To correlate between a CRN protocol
$\Pi = (\Species, \Reactions)$
and concrete computational tasks, we associate $\Pi$ with a (task specific)
\emph{interface}
$\mathcal{I} = (\mathcal{U}, \mu, \mathcal{C})$
whose semantics is as follows:
$\mathcal{U}$ is a fixed set of \emph{interface values} that typically encode
the input and/or output associated with the species;
$\mu : \Species \rightarrow \mathcal{U}$
is an \emph{interface mapping} that maps each species
$A \in \Species$
to an interface value
$\mu(A) \in \mathcal{U}$;
and
$\mathcal{C} \subseteq \Naturals^{\mathcal{U}} \times \Naturals^{\mathcal{U}}$
is a \emph{correctness relation} that determines the correctness of an
execution as explained soon.\footnote{%
The abstract interface formulation generalizes various families of
computational tasks addressed in the CRN literature, including predicate
decision
\cite{Brijder2016minimal,
ChenCDS2017speed-faults,
Brijder2019tutorial}
(see also \Sect{}~\ref{section:crd}) and function computation
\cite{ChenDS2014deterministic,
DotyH2015leaderless,
Brijder2019tutorial},
as well as the vote amplification task discussed in
\Sect{}~\ref{section:amplification}, without committing to the
specifics of one particular family.
For example, for the CRDs presented in \Sect{}~\ref{section:crd}, we define
$\mathcal{U} = ( \Sigma \cup \{ \bot \} ) \times \{ 0, 1, \bot \}$.
The interface mapping $\mu$ then maps each species
$A \in \Species$
to the interface value
$\mu(A) = (x, y) \in \mathcal{U}$
defined so that
(I)
$x = A$
if
$A \in \Sigma$;
and
$x = \bot$
otherwise;
and
(II)
$y = v$
if
$A \in \Upsilon_{v}$;
and
$y = \bot$
otherwise.}

Hereafter, we refer to the vectors in $\Naturals^{\mathcal{U}}$ as
\emph{interface vectors}.
The interface of a configuration
$\c \in \Naturals^{\Species}$
in terms of the input/output that $\c$ encodes
\LongVersion 
(if any)
\LongVersionEnd 
is captured by the
interface vector
\[\textstyle
\mu(\c)
\, \triangleq \,
\left\langle
\c \left( \left\{ A \in \Species \mid \mu(A) = u \right\} \right)
\right\rangle_{u \in \mathcal{U}}
\, .
\]

The abstract interface
$\mathcal{I} = (\mathcal{U}, \mu, \mathcal{C})$
allows us to define what it means for a protocol
to be correct.
To this end, for each configuration
$\c \in \Naturals^{\Species}$,
let
$Z_{\mathcal{I}}(\c) = \{
\c' \in \Naturals^{\Species} \mid (\mu(\c), \mu(\c')) \in \mathcal{C}
\}$
be the set of configurations which are mapped by $\mu$ to interface vectors
that satisfy the correctness relation with $\mu(\c)$.
A configuration
$\c^{0} \in \Naturals^{\Species}$
is a \emph{valid initial configuration} with respect to
$\mathcal{I}$ if
$Z_{\mathcal{I}}(\c^{0}) \neq \emptyset$;
an execution is \emph{valid} (with respect to $\mathcal{I}$) if it emerges
from a valid initial configuration.
A valid execution $\eta$ is said to be \emph{stably correct} (resp.,
\emph{haltingly correct})
\LongVersion 
with respect to $\mathcal{I}$
\LongVersionEnd 
if $\eta$ stabilizes (resp., halts) into
$Z_{\mathcal{I}}(\c^{0})$.
\LongVersion 
\par
\LongVersionEnd 
The protocol $\Pi$ is said to be \emph{stably correct} (resp.,
\emph{haltingly correct})
\LongVersion 
with respect to $\mathcal{I}$
\LongVersionEnd 
if every weakly fair valid execution is guaranteed to be stably (resp.,
haltingly) correct.\footnote{%
Both notions of correctness have been studied in the CRN literature, see,
e.g., \cite{Brijder2019tutorial}.}
\LongVersion 
When the interface $\mathcal{I}$ is not important or clear from the context,
we may address the stable/halting correctness of executions and protocols
without explicitly mentioning $\mathcal{I}$.
\LongVersionEnd 

\subparagraph{The Stochastic Scheduler.}
While the current paper focuses on the (weakly fair) adversarial scheduler,
another type of scheduler that receives a lot of attention in the literature
is the \emph{stochastic scheduler}.
Here, we present the stochastic scheduler so that it can serve as a
``benchmark'' for the runtime definition introduced in
\Sect{}~\ref{section:runtime}.
To this end, we define the \emph{propensity} of a reaction
$\alpha = (\mathbf{r}, \mathbf{p}) \in \Reactions$
in a configuration
$\c \in \Naturals^{\Species}$,
denoted by $\Prp_{\c}(\alpha)$, as
\[\textstyle
\Prp_{\c}(\alpha)
\, = \,
\begin{cases}
\c(A) \cdot \frac{1}{|\Reactions(\mathbf{r})|} \, ,
&
\mathbf{r} = A
\\
\frac{1}{\Vol} \cdot \binom{\c(A)}{2} \cdot \frac{1}{|\Reactions(\mathbf{r})|}
\, ,
&
\mathbf{r} = 2 A
\\
\frac{1}{\Vol} \cdot \c(A) \cdot \c(B)
\cdot
\frac{1}{|\Reactions(\mathbf{r})|}
\, ,
&
\mathbf{r} = A + B, A \neq B
\end{cases}
\, ,
\]
where
$\Vol > 0$
is a (global) \emph{volume} parameter.\footnote{%
In the standard stochastic model \cite{Gillespie1977stochastic}, the
propensity expression is multiplied by a reaction specific rate coefficient.
In the current paper, that merely uses the stochastic scheduler as a
benchmark, we make the simplifying assumption that all rate coefficients are
set to $1$
(c.f.~\cite{ChenDS2014deterministic,
ChenCDS2017speed-faults}).}
Notice that reaction $\alpha$ is applicable to $\c$ if and only if
$\Prp_{\c}(\alpha) > 0$.
The propensity notation is extended to reaction (sub)sets
$Q \subseteq \Reactions$
by defining
$\Prp_{\c}(Q) = \sum_{\alpha \in Q} \Prp_{\c}(\alpha)$.
Recalling that
$\Reactions(\mathbf{r}) \neq \emptyset$
for each
$\mathbf{r} \in \Naturals^{\Species}$
with
$1 \leq \| \mathbf{r} \| \leq 2$,
we
observe that
\[\textstyle
\Prp_{\c}
\, \triangleq \,
\Prp_{\c}(\Reactions)
\, = \,
\| \c \| + \frac{1}{\Vol} \cdot \binom{\| \c \|}{2}
\, .
\]

The stochastic scheduler determines the execution
$\eta = \langle \c^{t}, \alpha^{t} \rangle_{t \geq 0}$
by scheduling a reaction
$\alpha \in \Applicable(\c^{t})$
in step $t$, setting
$\alpha^{t} = \alpha$,
with probability proportional to $\alpha$'s propensity $\Prp_{\c^{t}}(\alpha)$
in $\c^{t}$.
The assumption that the CRN protocol respects finite density implies that
$\eta$ is (weakly and strongly) fair with probability $1$.
We define the \emph{time span} of step
$t \geq 0$
to be
$1 \big/ \Prp_{\c^{t}}$,
i.e., the normalizing factor of the reaction selection probability.\footnote{%
The time span definition is consistent with the expected time until a reaction
occurs under the continuous time Markov process formulation of the standard
stochastic model \cite{Gillespie1977stochastic} with no rate coefficients.
}
Given a step
$t^{*} \geq 0$,
the \emph{stochastic runtime} of the execution prefix
$\eta^{*} = \langle \c^{t}, \alpha^{t} \rangle_{0 \leq t < t^{*}}$
is defined to be the accumulated time span
$\sum_{t = 0}^{t^{*} - 1} 1 \big/ \Prp_{\c^{t}}$.

We adopt the convention that the volume is proportional to the initial
molecular count
$n = \| \c^{0} \|$
\cite{Doty2014timing}.
The assumption that the CRN protocol respects finite density ensures that
$\Vol = \Theta (\| \c^{t} \|)$
for every
$t \geq 0$
which means that the volume is sufficiently large to contain all molecules
throughout the (stochastic) execution $\eta$.
This also means that the time span of each step
$t \geq 0$
is
\begin{equation} \label{equation:time-span}
\textstyle
1 / \Prp_{\c^{t}}
\, = \,
\frac{\Vol}{\Vol \cdot \| \c^{t} \| + \binom{\| \c^{t} \|}{2}}
\, = \,
\Theta (1 / \| \c^{t} \|)
\, = \,
\Theta (1 / n)
\, ,
\end{equation}
hence the stochastic runtime of an execution prefix that lasts for $t^{*}$
steps is
$\Theta (t^{*} / n)$.

\section{Correctness Characterization via the
Configuration Digraph}
\label{section:configuration-digraph}
It is often convenient to look at CRN protocols through the lens of the
following abstract directed graph (a.k.a.\ digraph):
The \emph{configuration digraph} of a protocol
$\Pi = (\Species, \Reactions)$
is a digraph, denoted by $D^{\Pi}$, whose edges are labeled by reactions in
$\Reactions$.
The nodes of $D^{\Pi}$ are identified with the configurations in
$\Naturals^{\Species}$;
the edge set of $D^{\Pi}$ includes an $\alpha$-labeled edge from $\c$ to
$\alpha(\c)$ for each configuration
$\c \in \Naturals^{\Species}$
and reaction
$\alpha \in \Applicable(\c)$
(thus the outdegree of $\c$ in $D^{\Pi}$ is $|\Applicable(\c)|$).
Observe that the self-loops of $D^{\Pi}$ are exactly the edges labeled by
(applicable) void reactions.
Moreover, a configuration $\c'$ is reachable, in the graph theoretic sense,
from a configuration $\c$ if and only if
$\c \Reaching \c'$.
For a configuration
$\c \in \Naturals^{\Species}$,
let $D^{\Pi}_{\c}$ be the digraph induced by $D^{\Pi}$ on the set of
configurations reachable from $\c$ and observe that $D^{\Pi}_{\c}$ is finite as
$\Pi$ respects finite density.
(Refer to \Fig{}\
\ref{figure:crn-kill-b:digraph}--\ref{figure:crn-skipping-policy:digraph} for
illustrations of the notions presented in this section.)

\LongVersion 
By definition, there is a one-to-one correspondence between the executions
$\eta = \langle \c^{t}, \alpha^{t} \rangle_{t \geq 0}$
of $\Pi$ and the infinite paths
$P = (\c^{0}, \c^{1}, \dots)$
in $D^{\Pi}$, where the edges of $P$ are labeled by the reaction sequence
$(\alpha^{0}, \alpha^{1}, \dots)$.
We say that an infinite path in $D^{\Pi}$ is \emph{weakly fair} (resp.,
\emph{strongly fair}) if its corresponding execution is weakly (resp.,
strongly) fair.
\LongVersionEnd 

The \emph{(strongly connected) components} of the configuration digraph
$D^{\Pi}$ are the equivalence classes of the ``reachable from each other''
relation over the configurations in $\Naturals^{\Species}$.
We say that a reaction
$\alpha \in \Reactions$
\emph{escapes} from a component $S$ of $D^{\Pi}$ if every configuration in $S$
admits an outgoing $\alpha$-labeled edge to a configuration not in $S$;
i.e.,
$\alpha \in \Applicable(\c)$
and
$\alpha(\c) \notin S$
for every
$\c \in S$
(see, e.g., \Fig{}~\ref{figure:crn-kill-b:digraph}).
\ShortVersion 
The notion of escaping reactions allows us to express the stable/halting
correctness of CRNs in terms of their
configuration digraphs.
\ShortVersionEnd 
\LongVersion 
The notion of escaping reactions allows us to state the following key lemma.

\begin{lemma} \label{lemma:non-escaping-scc-and-fair-path}
Consider a component $S$ of $D^{\Pi}$.
The digraph $D^{\Pi}$ admits a weakly fair infinite path all of whose nodes
are in $S$ if and only if none of the reactions in $\Reactions$ escapes from
$S$.
\end{lemma}
\begin{proof}
By definition, if $S$ admits an escaping reaction
$\alpha \in \Reactions$,
then every weakly fair infinite path $P$ in $D^{\Pi}$ that visits $S$ cannot
stay in $S$ indefinitely without starving $\alpha$, hence $P$ must eventually
leave $S$.
In the converse direction, assume that none of the reactions in $\Reactions$
escapes from $S$ and let $D^{\Pi}(S)$ be the digraph induced by $D^{\Pi}$ on $S$.
For each reaction
$\alpha \in \Reactions$,
let
$e_{\alpha} = (\c, \c')$
be an edge in $D^{\Pi}(S)$ that satisfies either
(1)
$e_{\alpha}$ is labeled by $\alpha$;
or
(2)
$\alpha$ is inapplicable to $\c$.
(Such an edge $e_{\alpha}$ is guaranteed to exist as $\alpha$ does not escape
from $S$.)
Since $D^{\Pi}(S)$ is a strongly connected digraph, it follows that there
exists a (not necessarily simple) cycle $C$ in $D^{\Pi}(S)$ that includes the
edges $e_{\alpha}$ for all
$\alpha \in \Reactions$.
By repeatedly traversing $C$, we obtain a weakly fair infinite path in
$D^{\Pi}$.
\end{proof}
\LongVersionEnd 

\LongVersion 
We can now express the stable/halting correctness of CRNs in terms of their
configuration digraphs:
\Lem{}~\ref{lemma:correctness-via-configuration-digraph} follows from
\Lem{}~\ref{lemma:non-escaping-scc-and-fair-path} by the definitions of stable
correctness and halting correctness.
\LongVersionEnd 

\begin{lemma} \label{lemma:correctness-via-configuration-digraph}
A CRN protocol
$\Pi = (\Species, \Reactions)$
is stably (resp., haltingly) correct with respect to an interface
$\mathcal{I} = (\mathcal{U}, \mu, \mathcal{C})$
under a weakly fair scheduler
if and only if
for every valid initial configuration
$\c^{0} \in \Naturals^{\Species}$,
every component $S$ of $D^{\Pi}_{\c^{0}}$ satisfies (at least) one of the
following two conditions:
(1)
$S$ admits some (at least one) escaping reaction;
or
(2)
$S \subseteq \Stab(Z_{\mathcal{I}}(\c^{0}))$
(resp.,
$S \subseteq \Halt(Z_{\mathcal{I}}(\c^{0}))$),
where
$Z_{\mathcal{I}}(\c^{0}) = \{
\c \in \Naturals^{\Species} \mid (\mu(\c^{0}), \mu(\c)) \in \mathcal{C}
\}$.
\end{lemma}

To complement \Lem{}~\ref{lemma:correctness-via-configuration-digraph}, we
also express the stable/halting correctness of CRNs in terms of their
configuration digraphs under a strongly fair scheduler%
\LongVersion 
:
\Lem{}~\ref{lemma:correctness-via-configuration-digraph-strong-fairness}
follows from the same line of arguments as Lemma 1 in
\cite{AngluinADFP2006computation} by the definitions of stable correctness and
halting correctness.
\LongVersionEnd 
\ShortVersion 
.
\ShortVersionEnd 

\begin{lemma}
\label{lemma:correctness-via-configuration-digraph-strong-fairness}
A CRN protocol
$\Pi = (\Species, \Reactions)$
is stably (resp., haltingly) correct with respect to an interface
$\mathcal{I} = (\mathcal{U}, \mu, \mathcal{C})$
under a strongly fair scheduler
if and only if
for every valid initial configuration
$\c^{0} \in \Naturals^{\Species}$,
every component $S$ of $D^{\Pi}_{\c^{0}}$ satisfies (at least) one of the
following two conditions:
(1)
$S$ admits some (at least one) edge outgoing to another component;
or
(2)
$S \subseteq \Stab(Z_{\mathcal{I}}(\c^{0}))$
(resp.,
$S \subseteq \Halt(Z_{\mathcal{I}}(\c^{0}))$),
where
$Z_{\mathcal{I}}(\c^{0}) = \{
\c \in \Naturals^{\Species} \mid (\mu(\c^{0}), \mu(\c)) \in \mathcal{C}
\}$.
\end{lemma}

Combining \Lem{}\
\ref{lemma:correctness-via-configuration-digraph} and
\ref{lemma:correctness-via-configuration-digraph-strong-fairness}, we obtain
the following corollary.

\begin{corollary} \label{corollary:weak-correctness-implies-strong-correctness}
If a CRN protocol
$\Pi = (\Species, \Reactions)$
is stably (resp., haltingly) correct with respect to an interface
$\mathcal{I}$ under a weakly fair scheduler, then $\Pi$ is also stably (resp.,
haltingly) correct with respect to $\mathcal{I}$ under a strongly fair
scheduler.
\end{corollary}

\LongVersion 
%
\subparagraph{Two Protocols in One Test Tube.}
A common technique in the design of CRN (and population) protocols is to
simulate two protocols
$\Pi_{1} = (\Species_{1}, \Reactions_{1})$
and
$\Pi_{2} = (\Species_{2}, \Reactions_{2})$
running ``in the same test tube''.
This is often done by constructing a ``combined'' protocol
$\Pi_{\times} = (\Species_{\times}, \Reactions_{\times})$
whose species set $\Species_{\times}$ is the Cartesian product 
$\Species_{1} \times \Species_{2}$
so that each reaction
$\alpha \in \Reactions_{\times}$
operates independently on the two ``sub-species''.
While this design pattern is very effective with strong fairness (it is
used, e.g., in
\cite{AngluinADFP2006computation,
AngluinAE2008fast}),
it turns out that the weakly fair adversarial scheduler may exploit the
Cartesian product construction to introduce ``livelocks'', preventing
$\Pi_{\times}$ from stabilizing/halting;
an example that demonstrates this phenomenon is presented in
Appendix~\ref{appendix:curse-cartesian-product}.

Consequently, the current paper uses a different type of construction when we
wish to simulate $\Pi_{1}$ and $\Pi_{2}$ in the same test tube:
We simply produce two separated sets of molecules, one for the $\Pi_{1}$
species and the other for the $\Pi_{2}$ species, and allow the two protocols
to run side-by-side.
Care must be taken though with regard to reactions that involve species from
both $\Species_{1}$ and $\Species_{2}$ as the weakly fair adversarial
scheduler may exploit those to interfere with the executions of the
individual protocols;
see our treatment of this issue in
\Sect{}
\ref{section:runtime:tools:ignition},
\ref{section:amplification:application}, and
\ref{section:crd:semilinear:closure}.
\LongVersionEnd 

\section{The Runtime of Adversarially Scheduled
Executions}
\label{section:runtime}
So far, the literature on CRN
\LongVersion 
(and population)
\LongVersionEnd 
protocols operating under an adversarial scheduler focused mainly on
computability, leaving aside, for the most part, complexity considerations.%
\footnote{%
The one exception in this regard is the work of Chen et
al.~\cite{ChenCDS2017speed-faults} on speed faults --- see
\Sect{} \ref{section:runtime:speed-faults} and \ref{section:crd}.
}
This is arguably unavoidable when working with the strong fairness condition
which is inherently oblivious to the chain of reactions that realizes the
reachability of one configuration from another.
In the current paper, however, we adopt the weak fairness condition which
facilitates the definition of a quantitative measure for the runtime of
adversarially scheduled executions, to which this section is dedicated.
(Refer to \Fig{}\
\ref{figure:crn-kill-b:runtime-policy}--%
\ref{figure:crn-skipping-policy:runtime-policy}
for illustrations of the notions presented in this section.)

Consider a stably (resp., haltingly) correct CRN protocol
$\Pi = (\Species, \Reactions)$%
\LongVersion 
\ and recall that every weakly fair valid execution of $\Pi$ is guaranteed to
stabilize (resp., halt)%
\LongVersionEnd 
.
We make extensive use of the following operator:
Given a weakly fair execution
$\eta = \langle \c^{t}, \alpha^{t} \rangle_{t \geq 0}$,
a step
$t \geq 0$,
and
a reaction (sub)set
$Q \subseteq \Reactions$,
let
$\Stop(\eta, t, Q)$
be the earliest step
$s > t$
such that at least one of the following two conditions is satisfied:
\\
(I)
$\alpha^{s - 1} \in Q$;
or
\\
(II)
$Q \subseteq \bigcup_{t \leq t' \leq s} \overline{\Applicable}(\c^{t'})$.
\\
(This operator is well defined by the weak fairness of $\eta$.)

Intuitively, we think of the operator
$\Stop(\eta, t, Q)$
as a process that tracks $\eta$ from step $t$ onward and stops once any $Q$
reaction is scheduled (condition (I)).
This by itself is not well defined as the scheduler may avoid scheduling the
$Q$ reactions from step $t$ onward.
However, the scheduler must prevent the starvation of any continuously
applicable reaction in $Q$, so we also stop the $\Stop$-process once the
adversary ``fulfills this commitment'' (condition (II)).

\subparagraph{The Policies.}
Our runtime measure is based on partitioning a given weakly fair execution
$\eta = \langle \c^{t}, \alpha^{t} \rangle_{t \geq 0}$
into \emph{rounds}.
This is done by means of two policies:
a \emph{runtime policy} $\RunPol$, determined by the protocol designer,
that maps each configuration
$\c \in \Naturals^{\Species}$
to a non-void reaction (sub)set
$\RunPol(\c) \subseteq \NonVoid(\Reactions)$,
referred to as the \emph{target reaction set} of $\c$ under
$\RunPol$;
and a \emph{skipping policy} $\SkipPol$, determined by the adversarial
scheduler (in conjunction with the execution $\eta$), that maps each step
$t \geq 0$
to a step
$\SkipPol(t) \geq t$.

Round
$i = 0, 1, \dots$
spans the step interval
$[\InitStep(i), \InitStep(i + 1))$
and includes a designated \emph{effective step}
$\InitStep(i) \leq \EffStep(i) < \InitStep(i + 1)$.
The partition of execution $\eta$ into rounds is defined inductively by
setting
\[
\InitStep(i)
\, = \,
\begin{cases}
0
\, , &
i = 0
\\
\Stop \left(
\eta,
\EffStep(i - 1),
\RunPol \left( \c^{\EffStep(i - 1)} \right)
\right)
\, , &
i > 0
\end{cases}
\qquad\text{and}\qquad
\EffStep(i)
\, = \,
\SkipPol(\InitStep(i))
\, .
\]
Put differently, for every round
$i \geq 0$
with initial step $t(i)$, the adversarial scheduler first determines the
round's effective step
$\EffStep(i) = \SkipPol(\InitStep(i)) \geq \InitStep(i)$
by means of the skipping policy $\SkipPol$.
Following that, we apply the runtime policy $\RunPol$ (chosen by the protocol
designer) to the configuration
$\EffConf^{i} = \c^{\EffStep(i)}$,
referred to as the round's \emph{effective configuration}, and obtain the
target reaction set
$Q = \RunPol(\EffConf^{i})$.
The latter is then plugged into the operator $\Stop$ to determine
$\InitStep(i + 1) = \Stop(\eta, \EffStep(i), Q)$.
Round $i$ is said to be \emph{target-accomplished} if
$\alpha^{\InitStep(i + 1) - 1} \in Q$;
otherwise, it is said to be \emph{target-deprived}.

\begin{remark*}
Our definition of the runtime policy $\RunPol$ does not require that the
reactions included in the target reaction set $\RunPol(\c)$ are applicable to
the configuration
$\c \in \Naturals^{\Species}$.
Notice though that if
\LongVersion 
$\RunPol(\c) \subseteq \overline{\Applicable}(\c)$
i.e.,
\LongVersionEnd 
all target reactions are inapplicable to $\c$ (which is bound to be the
case if $\c$ is halting), then a round whose effective configuration is $\c$
is destined to be target deprived and end immediately after the effective
step, regardless of the reaction scheduled in that step.
In
\LongVersion 
\Sect{}~\ref{section:justify-runtime-policy}%
\LongVersionEnd 
\ShortVersion 
the attached full version%
\ShortVersionEnd 
, we investigate several other ``natural restrictions'' of the runtime policy
definition, including fixed policies and singleton target reaction sets,
showing that they all lead to significant efficiency loss.
\end{remark*}

\subparagraph{Temporal Cost.}
We define the \emph{temporal cost} of a configuration
$\c \in \Naturals^{\Species}$
under a runtime policy $\RunPol$, denoted by
$\TemporalCost^{\RunPol}(\c)$,
as follows:
Let
$\eta_{r} = \langle \c_{r}^{t}, \alpha_{r}^{t} \rangle_{t \geq 0}$
be a stochastic execution emerging from the initial configuration
$\c_{r}^{0} = \c$
and define
\[\textstyle
\TemporalCost^{\RunPol}(\c)
\, \triangleq \,
\Ex \left(
\sum_{t = 0}^{\Stop(\eta_{r}, 0, \RunPol(\c)) - 1} 1 / \Prp_{\c_{r}^{t}}
\right)
\, = \,
\Theta (1 / \| \c \|)
\cdot
\Ex \left( \Stop(\eta_{r}, 0, \RunPol(\c)) \right)
\, ,
\]
where the expectation is over the random choice of $\eta_{r}$ and the second
transition is due to \eqref{equation:time-span}.
That is, the temporal cost of $\c$ under $\RunPol$ is defined to be the
expected stochastic runtime of round $0$ of $\eta_{r}$ with respect to the
runtime policy $\RunPol$ and the identity skipping policy
$\SkipPol_{\Identity}$ that maps each step
$t \geq 0$
to
$\SkipPol_{\Identity}(t) = t$
(which means that the effective step of each round is its initial step).
\LongVersion 
The following observation stems from the Markovian nature of the stochastic
scheduler.

\begin{observation} \label{observation:time-span-stochastic-execution}
Fix an (arbitrary) runtime policy $\RunPol$.
Let
$\eta_{r} = \langle \c_{r}^{t}, \alpha_{r}^{t} \rangle_{t \geq 0}$
be a stochastic execution and let $\InitStep(i)$ be the initial step of round
$i \geq 0$
under $\RunPol$ and the identity skipping policy $\SkipPol_{\Identity}$.
For each
$i \geq 0$,
conditioned on $\c_{r}^{\InitStep(i)}$, the expected stochastic runtime of
$\langle \c_{r}^{t}, \alpha_{r}^{t} \rangle_{\InitStep(i) \leq t < \InitStep(i + 1)}$
is equal to
$\TemporalCost^{\RunPol}(\c_{r}^{\InitStep(i)})$.
\end{observation}
\LongVersionEnd 

\subparagraph{Execution Runtime.}
Consider a runtime policy $\RunPol$ and a skipping policy $\SkipPol$.
Let
$\eta = \langle \c^{t}, \alpha^{t} \rangle_{t \geq 0}$
be a weakly fair valid execution and let $\InitStep(i)$, $\EffStep(i)$, and
$\EffConf^{i} = \c^{\EffStep(i)}$
be the initial step, effective step, and effective configuration,
respectively, of round
$i \geq 0$
under $\RunPol$ and $\SkipPol$.
Fix some step
$t^{*} \geq 0$
and consider the execution prefix
$\eta^{*} = \langle \c^{t}, \alpha^{t} \rangle_{0 \leq t < t^{*}}$.
We define the \emph{(adversarial) runtime} of $\eta^{*}$ under $\RunPol$ and
$\SkipPol$, denoted by
$\RunTime^{\RunPol, \SkipPol}(\eta^{*})$,
by taking
$i^{*} = \min \{ i \geq 0 \mid \InitStep(i) \geq t^{*} \}$
and setting
\[\textstyle
\RunTime^{\RunPol, \SkipPol}(\eta^{*})
\, \triangleq \,
\sum_{i = 0}^{i^{*} - 1} \TemporalCost^{\RunPol} \left( \EffConf^{i} \right)
\, .
\]
The \emph{stabilization runtime} (resp., \emph{halting runtime}) of the
(entire) execution $\eta$ under $\RunPol$ and $\SkipPol$, denoted by
$\RunTime_{\Stab}^{\RunPol, \SkipPol}(\eta)$
(resp.,
$\RunTime_{\Halt}^{\RunPol, \SkipPol}(\eta)$),
is defined to be
$\RunTime^{\RunPol, \SkipPol} \left(
\langle \c^{t}, \alpha^{t} \rangle_{0 \leq t < t^{*}}
\right)$,
where
$t^{*} \geq 0$
is the stabilization (resp., halting) step of $\eta$.
In other words, we use $\RunPol$ and $\SkipPol$ to partition $\eta$ into
rounds and mark the effective steps.
Following that, we charge each round $i$ that starts before step $t^{*}$
according to the temporal cost (under $\RunPol$) of its effective
configuration $\EffConf^{i}$.

Looking at it from another angle,
\LongVersion 
by employing its
\LongVersionEnd 
\ShortVersion 
using the
\ShortVersionEnd 
skipping policy $\SkipPol$,
the adversarial scheduler determines the sequence
$\EffConf^{0}, \EffConf^{1}, \dots$
of effective configurations according to which the temporal cost
$\TemporalCost^{\RunPol}(\EffConf^{i})$
of each round
$i \geq 0$
is calculated.
By choosing an appropriate runtime policy $\RunPol$, the protocol designer may
(1)
ensure that progress is made from one effective configuration to the next,
thus advancing $\eta$ towards round
$i^{*} = \min \{ i \geq 0 \mid \InitStep(i) \geq t^{*} \}$;
and
(2)
bound the contribution
$\TemporalCost^{\RunPol}(\EffConf^{i})$
of each round
$0 \leq i < i^{*}$
to the stabilization runtime
$\RunTime_{\Stab}^{\RunPol, \SkipPol}(\eta)$
(resp., halting runtime
$\RunTime_{\Halt}^{\RunPol, \SkipPol}(\eta)$).
The crux of our runtime definition is that this contribution depends only on
the effective configuration $\EffConf^{i}$, irrespectively of how round $i$
actually develops (see, e.g., \Fig{}~\ref{figure:crn-kill-b:runtime-policy}).

\begin{remark*}
Using this viewpoint, it is interesting to revisit the definitions of
\LongVersion 
Awerbuch \cite{Awerbuch85} and Dolev et al.~\cite{Dolev-etal-97}
\LongVersionEnd 
\ShortVersion 
\cite{Awerbuch85} and \cite{Dolev-etal-97}
\ShortVersionEnd 
for the runtime of an asynchronous distributed protocol $\mathcal{P}$.
Following the discussion in \Sect{}~\ref{section:introduction}, this runtime
is defined as the length of the longest sequence
$\EffConf^{0}, \EffConf^{1}, \dots, \EffConf^{i^{*} - 1}$
of ``non-terminal'' configurations (of $\mathcal{P}$) such that $\EffConf^{i}$
is reachable from
$\EffConf^{i - 1}$
by an execution interval that lasts for at least one round (according to the
respective definitions of \cite{Awerbuch85} and \cite{Dolev-etal-97}).
Our adversarial runtime
\LongVersion 
notion
\LongVersionEnd 
is defined in the same manner,
taking 
$\EffConf^{0}, \EffConf^{1}, \dots, \EffConf^{i^{*} - 1}$
to be the first $i^{*}$ effective (CRN) configurations,
only that we charge each configuration $\EffConf^{i}$ according to its
temporal cost (rather than one ``runtime unit'' as in \cite{Awerbuch85} and
\cite{Dolev-etal-97}).
This difference is consistent with the different ``idealized benchmarks'':
a synchronous schedule in \cite{Awerbuch85} and \cite{Dolev-etal-97} vs.\ a
stochastic execution in the current paper.
The skipping policy $\SkipPol$ plays a key role in adversarially generating
the sequence
$\EffConf^{0}, \EffConf^{1}, \dots, \EffConf^{i^{*} - 1}$
as it ``decouples'' between the last step of round $i$, determined by the
runtime policy $\RunPol$, and the effective configuration
$\EffConf^{i + 1}$
of round
$i + 1$
(see, e.g., \Fig{}~\ref{figure:crn-skipping-policy:runtime-policy}).
\end{remark*}

\subparagraph{The Runtime Function.}
For
$n \geq 1$,
let $\mathcal{F}(n)$ denote the set of weakly fair valid executions
$\eta = \langle \c^{t}, \alpha^{t} \rangle_{t \geq 0}$
of initial molecular count
$\| \c^{0} \| = n$.
The \emph{stabilization runtime} (resp., \emph{halting runtime}) of the
CRN protocol $\Pi$ for executions in $\mathcal{F}(n)$, denoted by
$\RunTime_{\Stab}^{\Pi}(n)$
(resp.,
$\RunTime_{\Halt}^{\Pi}(n)$),
is defined to be
\[\textstyle
\RunTime_{\mathrm{x}}^{\Pi}(n)
\, \triangleq \,
\min_{\RunPol}
\max_{\eta \in \mathcal{F}(n), \, \SkipPol}
\,
\RunTime_{\mathrm{x}}^{\RunPol, \SkipPol}(\eta)
\, ,
\]
where $\mathrm{x}$ serves as a placeholder for $\Stab$ (resp., $\Halt$).
This formalizes the responsibility of the protocol designer to specify a
runtime policy $\RunPol$, in conjunction with the protocol $\Pi$, used for
up-bounding $\Pi$'s stabilization (resp., halting) runtime (see, e.g.,
\Fig{}~\ref{figure:crn-kill-b:runtime-policy}).

The following two lemmas establish the soundness of our adversarial runtime
definition:
\Lem{}~\ref{lemma:run-time-policy-for-all-executions} ensures that the
stabilization (resp., halting) runtime function is well defined%
\LongVersion 
;%
\LongVersionEnd 
\ShortVersion 
.%
\ShortVersionEnd 
\footnote{%
Note that in \Lem{}~\ref{lemma:run-time-policy-for-all-executions} we use a
universal runtime policy that applies to all choices of the initial molecular
count $n$.
This is stronger in principle than what the runtime definition actually
requires.}
\LongVersion 
its proof relies on some tools introduced in
\Sect{}~\ref{section:runtime:tools:reachability-via-avoiding-paths} and is
therefore deferred to that section.
\LongVersionEnd 
In \Lem{}~\ref{lemma:stochastic-benchmark}, we show that if the scheduler
generates the execution stochastically, then our (adversarial)
runtime measure agrees in expectation with the stochastic runtime measure.

\begin{lemma} \label{lemma:run-time-policy-for-all-executions}
Consider a stably (resp., haltingly) correct protocol
$\Pi = (\Species, \Reactions)$.
There exists a runtime policy $\RunPol$ such that for every
integer
$n \geq 1$,
execution
$\eta \in \mathcal{F}(n)$,
and skipping policy $\SkipPol$,
the stabilization runtime
$\RunTime_{\Stab}^{\RunPol, \SkipPol}(\eta)$
(resp., halting runtime
$\RunTime_{\Halt}^{\RunPol, \SkipPol}(\eta)$)
is up-bounded as a function of $n$.
\end{lemma}

\begin{lemma}
\label{lemma:stochastic-benchmark}
Consider a stably (resp., haltingly) correct protocol
$\Pi = (\Species, \Reactions)$.
Let
$\eta_{r}
=
\langle \c_{r}^{t}, \alpha_{r}^{t} \rangle_{t \geq 0}$
be a stochastic execution emerging from a valid initial configuration
$\c_{r}^{0}$
and let
$t^{*} \geq 0$
be the stabilization (resp., halting) step of $\eta_{r}$.
Then,
\[\textstyle
\min_{\RunPol}
\Ex_{\eta_{r}} \left(
\max_{\SkipPol}
\RunTime_{\mathrm{x}}^{\RunPol, \SkipPol}(\eta_{r})
\right)
\, = \,
\Ex_{\eta_{r}} \left(
\sum_{t = 0}^{t^{*} - 1} 1 \big/ \Prp_{\c_{r}^{t}}
\right)
\, ,
\]
where $\mathrm{x}$ serves as a placeholder for $\Stab$ (resp., $\Halt$).
\end{lemma}
\LongVersion 
\begin{proof}
Let $\RunPol_{\mathrm{f}}$ be the ``full'' runtime policy that maps each
configuration
$\c \in \Naturals^{\Species}$
to
$\RunPol_{\mathrm{f}}(\c) = \NonVoid(\Reactions)$
and let $\SkipPol_{\Identity}$ be the identity skipping policy.
We establish the assertion by proving the following three claims:
\\
(C1)
$\Ex_{\eta_{r}} \left(
\RunTime_{\mathrm{x}}^{\RunPol_{\mathrm{f}}, \SkipPol_{\Identity}}
(\eta_{r})
\right)
=
\Ex_{\eta_{r}} \left(
\sum_{t = 0}^{t^{*} - 1} 1 \big/ \Prp_{\c_{r}^{t}}
\right)$;
\\
(C2)
$\RunTime_{\mathrm{x}}^{\RunPol_{\mathrm{f}}, \SkipPol_{\Identity}}(\eta)
\geq
\RunTime_{\mathrm{x}}^{\RunPol_{\mathrm{f}}, \SkipPol}(\eta)$
for every execution
$\eta = \langle \c^{t}, \alpha^{t} \rangle_{t \geq 0} \in \mathcal{F}(n)$
and skipping policy $\SkipPol$;
and
\\
(C3)
$\Ex_{\eta_{r}} \left(
\RunTime_{\mathrm{x}}^{\RunPol, \SkipPol_{\Identity}}
(\eta_{r})
\right)
\geq
\Ex_{\eta_{r}} \left(
\sum_{t = 0}^{t^{*} - 1} 1 \big/ \Prp_{\c_{r}^{t}}
\right)$
for every runtime policy $\RunPol$.
\\
Indeed, claims (C1) and (C2) imply that
\begin{align*}
\min_{\RunPol}
\Ex_{\eta_{r}} \left(
\max_{\SkipPol}
\RunTime_{\mathrm{x}}^{\RunPol, \SkipPol}(\eta_{r})
\right)
\, \leq \, &
\Ex_{\eta_{r}} \left(
\max_{\SkipPol}
\RunTime_{\mathrm{x}}^{\RunPol_{\mathrm{f}}, \SkipPol}(\eta_{r})
\right)
\\
= \, &
\Ex_{\eta_{r}} \left(
\RunTime_{\mathrm{x}}^{\RunPol_{\mathrm{f}}, \SkipPol_{\Identity}}(\eta_{r})
\right)
\, = \,
\Ex_{\eta_{r}} \left(
\sum_{t = 0}^{t^{*} - 1} 1 \big/ \Prp_{\c_{r}^{t}}
\right)
\, ,
\end{align*}
whereas claim (C3) yields
\[
\min_{\RunPol}
\Ex_{\eta_{r}} \left(
\max_{\SkipPol}
\RunTime_{\mathrm{x}}^{\RunPol, \SkipPol}(\eta_{r})
\right)
\, \geq \,
\min_{\RunPol}
\Ex_{\eta_{r}} \left(
\RunTime_{\mathrm{x}}^{\RunPol, \SkipPol_{\Identity}}(\eta_{r})
\right)
\, \geq \,
\Ex_{\eta_{r}} \left(
\sum_{t = 0}^{t^{*} - 1} 1 \big/ \Prp_{\c_{r}^{t}}
\right)
\, .
\]

To prove the three claims, we start by deducing that claim (C3) follows from
\Obs{}~\ref{observation:time-span-stochastic-execution}, observing that the
inequality may become strict (only) due to the excessive contribution to
$\RunTime_{\mathrm{x}}^{\RunPol, \SkipPol_{\Identity}}(\eta_{r})$
of the temporal cost charged to the (unique) round
$i \geq 0$
that satisfies
$\InitStep(i) < t^{*} < \InitStep(i - 1)$
(if such a round exists).
As the target reaction sets exclude void reactions, we conclude by the
definition of operator $\Stop$ that under $\RunPol_{\mathrm{f}}$ and
$\SkipPol_{\Identity}$, there must exist a round
$i \geq 0$
such that
$t^{*} = \InitStep(i)$,
thus obtaining claim (C1).
For claim (C2), it suffices to observe that under $\RunPol_{\mathrm{f}}$, it
holds that
\[\textstyle
\InitStep(i + 1)
\, = \,
\min \{ t' > \EffStep(i) \mid \alpha^{t' - 1} \in \NonVoid(\Reactions) \}
\, = \,
\min \{ t' > \EffStep(i) \mid \c^{t'} \neq \c^{\EffStep(i)} \}
\]
for every round
$i \geq 0$.
\end{proof}
\LongVersionEnd 

\subsection{Speed Faults}
\label{section:runtime:speed-faults}
Consider a CRN protocol
$\Pi = (\Species, \Reactions)$
which is stably (resp., haltingly) correct with respect to an interface
$\mathcal{I} = (\mathcal{U}, \mu, \mathcal{C})$.
For a valid initial configuration
$\c^{0} \in \Naturals^{\Species}$,
let
$Z_{\mathcal{I}}(\c^{0}) = \{
\c \in \Naturals^{\Species}
\mid
(\mu(\c^{0}), \mu(\c)) \in \mathcal{C}
\}$
and recall that if a weakly fair execution $\eta$ of $\Pi$ emerges from
$\c^{0}$, then $\eta$ is guaranteed
to reach $\Stab(Z_{\mathcal{I}}(\c^{0}))$ (resp.,
$\Halt(Z_{\mathcal{I}}(\c^{0}))$).

Given a parameter
$s > 0$,
a configuration
$\c \in \Naturals^{\Species}$
is said to be a \emph{stabilization $s$-pitfall} (resp., a \emph{halting
$s$-pitfall}) of the valid initial configuration $\c^{0}$ if
$\c^{0} \Reaching \c$
and every path from $\c$ to $\Stab(Z_{\mathcal{I}}(\c^{0}))$
(resp., $\Halt(Z_{\mathcal{I}}(\c^{0}))$) in the digraph $D^{\Pi}$ includes
(an edge labeled by) a reaction whose propensity is at most
$s / \Vol$
(see, e.g., \Fig{} \ref{figure:crn-multiple-pitfalls:runtime-policy} and
\ref{figure:crn-skipping-policy:runtime-policy}).
When
$s = O (1)$,
we often omit the parameter and refer to $\c$ simply as a \emph{stabilization
pitfall} (resp., \emph{halting pitfall}).
Following the definition of Chen et al.~\cite{ChenCDS2017speed-faults}, we say
that an infinite family $\mathbf{C}^{0}$ of valid initial configurations
has a \emph{stabilization speed fault} (resp., \emph{halting speed
fault}) if
for every integer
$n_{0} > 0$,
there exists a configuration
$\c^{0} \in \mathbf{C}^{0}$
of molecular count
$\| \c^{0} \| = n \geq n_{0}$
that admits a stabilization (resp., halting) pitfall.

\begin{lemma} \label{lemma:speed-fault-lower-bound}
If an infinite family $\mathbf{C}^{0}$ of valid initial configurations has a
stabilization (resp., halting) speed fault, then for every integer
$n_{0} > 0$,
there exist
a configuration
$\c^{0} \in \mathcal{C}^{0}$
of molecular count
$\| \c^{0} \| = n \geq n_{0}$,
a weakly fair execution $\eta$ emerging from $\c^{0}$,
and a skipping policy $\SkipPol$,
such that
$\RunTime_{\mathrm{x}}^{\RunPol, \SkipPol}(\eta) \geq \Omega (n)$
for every runtime policy
$\RunPol$, where $\mathrm{x}$ serves as a placeholder for $\Stab$ (resp.,
$\Halt$).\footnote{%
As discussed in \cite{ChenCDS2017speed-faults}, a speed fault does not imply an
$\Omega (n)$
lower bound on the (stochastic) runtime of stochastically scheduled executions
since the probability of reaching a pitfall configuration may be small.
}
\end{lemma}
\LongVersion 
\begin{proof}
Let $\c^{0}$ be a configuration in $\mathcal{C}^{0}$
of molecular count
$\| \c^{0} \| = n \geq n_{0}$
that admits a stabilization (resp., halting) pitfall $\c$.
As observed by Chen et al.~\cite{ChenCDS2017speed-faults}, a stochastically
scheduled execution emerging from $\c$ needs, in expectation, at least
$\Omega (n)$
time to stabilize (resp., halt).
Therefore, \Lem{}~\ref{lemma:stochastic-benchmark} implies that there exists a
weakly fair execution $\eta_{\c}$ emerging from $\c$ and a skipping policy
$\SkipPol_{\c}$ such that
$\RunTime_{\mathrm{x}}^{\RunPol, \SkipPol_{\c}}(\eta_{\c}) \geq \Omega (n)$
for every runtime policy
$\RunPol$.
As
$\c^{0} \Reaching \c$,
we can devise $\eta$ and $\SkipPol$ so that
$\RunTime_{\mathrm{x}}^{\RunPol, \SkipPol}(\eta)
=
\RunTime_{\mathrm{x}}^{\RunPol, \SkipPol_{\c}}(\eta_{\c})$,
thus establishing the assertion.
\end{proof}
\LongVersionEnd 

\LongVersion 
\subsection{Useful Toolbox for Runtime Analyses}
\label{section:runtime:tools}

\subsubsection{Bounding the Temporal Cost by means of $\RunPol$-Avoiding Paths}
\label{section:runtime:tools:reachability-via-avoiding-paths}
The following definition plays a key role in the runtime analysis of the CRN
protocols presented in the sequel.
Given a runtime policy $\RunPol$ and two configurations
$\c, \c' \in \Naturals^{\Species}$,
we say that $\c'$ is reachable from $\c$ via a \emph{$\RunPol$-avoiding path},
denoted by
$\c \Reaching_{\langle \RunPol \rangle} \c'$,
if there exists a path $P$ from $\c$ to $\c'$ in the configuration
digraph $D^{\Pi}$ of $\Pi$ that satisfies
(1)
all edges in $P$ are labeled by the reactions in
$\Reactions - \RunPol(\c)$;
and
(2)
there exists some (at least one) reaction
$\alpha \in \RunPol(\c)$
such that
$\alpha \in \Applicable(\hat{\c})$
for every configuration $\hat{\c}$ in $P$.
Equivalently, the relation
$\c \Reaching_{\langle \RunPol \rangle} \c'$
holds if (and only if) there exists a weakly fair execution
$\eta = \langle \c^{t}, \alpha^{t} \rangle_{t \geq 0}$,
a skipping policy $\SkipPol$,
and a round
$i \geq 0$
(defined with respect to $\RunPol$ and $\SkipPol$)
such that
$\c = \c^{\EffStep(i)}$
and
$\c' = \c^{t'}$
for some
$\EffStep(i) \leq t' < \InitStep(i + 1)$.

The usefulness of the notion of reachability via avoiding paths is manifested
in the following important lemma.
Its proof is fairly straightforward under the continuous time Markov process
formulation of the standard stochastic model
\cite{Gillespie1977stochastic};
for completeness, we provide, in
Appendix~\ref{appendix:proof:lemma:runtime:tools:temporal-cost}, a proof for
the discrete scheduler interpretation adopted in the current paper.

\begin{lemma} \label{lemma:runtime:tools:temporal-cost}
Consider a runtime policy $\RunPol$ and a configuration
$\c \in \Naturals^{\Species}$
and assume that
$\Prp_{\c'}(\RunPol(\c)) \geq p$
for every configuration
$\c' \in \Naturals^{\Species}$
such that
$\c \Reaching_{\langle \RunPol \rangle} \c'$.
Then, the temporal cost of $\c$ under $\RunPol$ is up-bounded as
$\TemporalCost^{\RunPol}(\c) \leq 1 / p$.
\end{lemma}

Employing \Lem{}~\ref{lemma:runtime:tools:temporal-cost}, we can now establish
\Lem{}~\ref{lemma:run-time-policy-for-all-executions}.

\begin{proof}[Proof of \Lem{}~\ref{lemma:run-time-policy-for-all-executions}]
Let
$\mathcal{L}(n) \subset \Naturals^{\Species}$
denote the set of configurations
$\c^{0} \in \Naturals^{\Species}$
of molecular count
$\| \c^{0} \| = n$
that are valid as initial configurations (recall that $\mathcal{F}(n)$
is the set of weakly fair executions emerging from initial configurations
in $\mathcal{L}(n)$).
Fix an integer
$n \geq 1$
and a configuration
$\c \in \Naturals^{\Species}$
which is reachable from some (at least one) valid initial configuration in
$\mathcal{L}(n)$ and let $S$ be the component of $\c$ in the configuration
digraph $D^{\Pi}$.
If $S$ does not admit any escaping reaction, then
\Lem{}~\ref{lemma:correctness-via-configuration-digraph} implies that any
execution in $\mathcal{F}(n)$ that reaches $\c$ has already stabilized (resp.,
halted).
Therefore, we can take $\RunPol(\c)$ to be an arbitrary reaction set as this
choice does not affect
$\RunTime_{\Stab}^{\RunPol}(\eta)$
(resp.,
$\RunTime_{\Halt}^{\RunPol}(\eta)$).

So, assume that $S$ admits a non-empty set
$Q \subseteq \Reactions$
of escaping reactions.
By setting
$\RunPol(\c) = Q'$
for an arbitrary reaction set
$\emptyset \subset Q' \subseteq Q$,
we ensure that for any execution
$\eta \in \mathcal{F}(n)$,
if $\c$ is the effective configuration of a round of $\eta$, then by the time
the next round begins, $\eta$ no longer resides in $S$.
The assumption that $\Pi$ respects finite density implies that the number of
components of $D^{\Pi}$ that $\eta$ goes through before it stabilizes (resp.,
halts) is up-bounded as a function of $n$.
As the propensity of any non-empty set of applicable reactions is at least
$1 / \Vol = \Theta (1 / n)$,
we conclude by \Lem{}~\ref{lemma:runtime:tools:temporal-cost} that each
such component contributes
$O (n)$
to
$\RunTime_{\Stab}^{\RunPol}(\eta)$
(resp.,
$\RunTime_{\Halt}^{\RunPol}(\eta)$),
thus establishing the assertion.
\end{proof}

\subsubsection{The Ignition Gadget}
\label{section:runtime:tools:ignition}
It is often convenient to design CRN protocols so that the molecules present
in the initial configuration belong to designated species whose role is to set
the execution into motion by transforming into the actual species that
participate in the protocol.
As this design feature is widespread in the protocols presented in the sequel,
we introduce it here as a standalone \emph{ignition gadget} so that it can be
used subsequently in a ``black box'' manner.

Formally, the ignition gadget of a CRN protocol
$\Pi = (\Species, \Reactions)$
is defined over a set
$\Species_{\Ignt} \subset \Species$
of \emph{ignition species}, referring to the species in
$\Species - \Species_{\Ignt}$
as \emph{working species}.
Each ignition species
$A \in \Species_{\Ignt}$
is associated with a unimolecular \emph{ignition reaction}
$\iota_{A} : A \rightarrow W_{A}^{1} + \cdots + W_{A}^{k_{A}}$,
where
$W_{A}^{1} + \cdots + W_{A}^{k_{A}} \in \Naturals^{\Species - \Species_{\Ignt}}$
is a multiset (or vector) of working species.
The ignition gadget requires that besides $\iota_{A}$, any reaction in which
the ignition species $A$ participates, as a reactant or as a product, is a
void reaction;
that is, 
$\mathbf{r}(A) = \mathbf{p}(A) = 0$
for every
$(\mathbf{r}, \mathbf{p}) \in \NonVoid(\Reactions) - \{ \iota_{A} \}$.

Given a weakly fair execution
$\eta = \langle \c^{t}, \alpha^{t} \rangle_{t \geq 0}$
of protocol $\Pi$, we say that the ignition gadget is \emph{mature} in step
$t \geq 0$
if
$\c^{t}(\Species_{\Ignt}) = 0$,
observing that this means that the ignition gadget is mature in any step
$t' \geq t$.

\begin{lemma} \label{lemma:runtime:tools:ignition}
Let $\RunPol$ be a runtime policy for protocol $\Pi$, designed so that
$\RunPol(\c) = \{ \iota_{A} \mid A \in \Species_{\Ignt} \}$
for every configuration
$\c \in \Naturals^{\Species}$
with
$\c(\Species_{\Ignt}) > 0$.
Then, for every weakly fair execution
$\eta = \langle \c^{t}, \alpha^{t} \rangle_{t \geq 0}$
and skipping policy $\SkipPol$, there exists a step
$t_{\Ignt} \geq 0$
such that $\eta$ is mature in step $t_{\Ignt}$.
Moreover, it is guaranteed that
$\RunTime^{\RunPol, \SkipPol}(\langle \c^{t}, \alpha^{t} \rangle_{0 \leq t < t_{\Ignt}})
\leq
O (\log n)$,
where
$n = \| \c^{0} \|$
is $\eta$'s initial molecular count.
\end{lemma}
\begin{proof}
The fact that step $t_{\Ignt}$ exists follows since the ignition reactions
remain applicable as long as the molecular count of the ignition species is
positive and since the ignition species are not produced by any non-void
reaction.

To bound the runtime of the execution prefix
$\eta_{\Ignt}
=
\langle \c^{t}, \alpha^{t} \rangle_{0 \leq t < t_{\Ignt}}$
under $\RunPol$ and $\SkipPol$, let $\InitStep(i)$ and $\EffConf^{i}$ be the
initial step and effective configuration, respectively, of round
$i \geq 0$
and let
$i_{\Ignt} = \min \{ i \geq 0 \mid \InitStep(i) \geq t_{\Ignt} \}$.
Fix some round
$0 \leq i < i_{\Ignt}$
and let
$\ell_{i} = \EffConf^{i}(\Species_{\Ignt})$.
The definition of the ignition gadget ensures that round $i$ is
target-accomplished with
$\ell_{i + 1} < \ell_{i}$
and that
$\Prp_{\c}(\RunPol(\EffConf^{i}))
=
\Prp_{\EffConf^{i}}(\RunPol(\EffConf^{i}))$
for every configuration
$\c \in \Naturals^{\Species}$
reachable form $\EffConf^{i}$ via a $\RunPol$-avoiding path.
As
$\Prp_{\EffConf^{i}}(\RunPol(\EffConf^{i}))
=
\Prp_{\EffConf^{i}}(\{ \iota_{A} \mid A \in \Species_{\Ignt} \})
=
\ell_{i}$,
we can employ \Lem{}~\ref{lemma:runtime:tools:temporal-cost} to conclude that
$\TemporalCost^{\RunPol}(\EffConf^{i}) \leq 1 / \ell_{i}$.
Since
$\ell_{0} \leq n$,
it follows that the runtime of $\eta_{\Ignt}$ under $\RunPol$ and $\SkipPol$
is bounded as
\[\textstyle
\RunTime^{\RunPol, \SkipPol}(\eta_{\Ignt})
\, = \,
\sum_{i = 0}^{i_{\Ignt} - 1} \TemporalCost^{\RunPol}(\EffConf^{i})
\, \leq \,
\sum_{\ell = 1}^{n} 1 / \ell
\, = \,
O (\log n)
\, ,
\]
thus establishing the assertion.
\end{proof}
\LongVersionEnd 

\section{Predicate Decidability}
\label{section:crd}
An important class of CRN protocols is that of \emph{chemical reaction
deciders (CRDs)} whose goal is to determine whether the initial molecular
counts of certain species satisfy a given predicate.
In its most general form (see
\cite{ChenCDS2017speed-faults,
Brijder2019tutorial}),
a CRD is a CRN protocol
$\Pi = (\Species, \Reactions)$
augmented with
(1)
a set
$\Sigma \subset \Species$
of \emph{input species};
(2)
two disjoint sets
$\Upsilon_{0}, \Upsilon_{1} \subset \Species$
of \emph{voter species};
(3)
a designated \emph{fuel species}
$F \in \Species - \Sigma$;
and
(4)
a fixed \emph{initial context}
$\mathbf{k} \in \Naturals^{\Species - (\Sigma \cup \{ F \})}$.
\LongVersion 
To emphasize that the protocol $\Pi$ is a CRD, we often write
$\Pi
=
(\Species, \Reactions, \Sigma, \Upsilon_{0}, \Upsilon_{1}, F, \mathbf{k})$.
\LongVersionEnd 
The CRD is said to be \emph{leaderless} if its initial context is the zero
vector, i.e.,
$\mathbf{k} = \mathbf{0}$.

A configuration
$\c^{0} \in \Naturals^{\Species}$
is valid as an initial configuration of the CRD $\Pi$ if
$\c^{0}|_{\Species - (\Sigma \cup \{ F \})} = \mathbf{k}$;
to ensure that the initial molecular count
$\| \c^{0} \|$
is always at least $1$ (especially when the CRD is leaderless), we also
require that
$\c^{0}(F) \geq 1$.
In other words, a valid initial configuration $\c^{0}$ can be decomposed into
an \emph{input vector}
$\c^{0}|_{\Sigma} = \mathbf{x} \in \Naturals^{\Sigma}$,
the initial context
$\c^{0}|_{\Species - (\Sigma \cup \{ F \})} = \mathbf{k}$,
and any number
$\c^{0}(F) \geq 1$
of fuel molecules.
We emphasize that in contrast to the initial context, the protocol designer has
no control over the \emph{exact} molecular count of the fuel species in the
initial configuration.

For
$v \in \{ 0, 1 \}$,
let
\LongVersion 
\[\textstyle
\mathcal{D}_{v}
\, = \,
\left\{
\c \in \Naturals^{\Species}
\mid
\c(\Upsilon_{v}) > 0
\land
\c(\Upsilon_{1 - v}) = 0
\right\}
\]
\LongVersionEnd 
\ShortVersion 
$\mathcal{D}_{v}
=
\left\{
\c \in \Naturals^{\Species}
\mid
\c(\Upsilon_{v}) > 0
\land
\c(\Upsilon_{1 - v}) = 0
\right\}$
\ShortVersionEnd 
be the set of configurations that include
\LongVersion 
a positive molecular count of
\LongVersionEnd 
$v$-voters and no
$(1 - v)$-voters.
An input vector
$\mathbf{x} \in \Naturals^{\Sigma}$
is said to be \emph{stably accepted} (resp., \emph{haltingly accepted}) by
$\Pi$ if for every valid initial configuration
$\c^{0} \in \Naturals^{\Species}$
with
$\c^{0}|_{\Sigma} = \mathbf{x}$,
every weakly fair execution
$\eta = \langle \c^{t}, \alpha^{t} \rangle_{t \geq 0}$
\LongVersion 
emerging from $\c^{0}$
\LongVersionEnd 
stabilizes (resp., halts) into $\mathcal{D}_{1}$;
the input vector
$\mathbf{x} \in \Naturals^{\Sigma}$
is said to be \emph{stably rejected} (resp., \emph{haltingly rejected}) by
$\Pi$ if the same holds with $\mathcal{D}_{0}$.
The CRD $\Pi$ is stably (resp., haltingly) correct if every input vector
$\mathbf{x} \in \Naturals^{\Sigma}$
is either stably (resp., haltingly) accepted or stably (resp., haltingly)
rejected by $\Pi$.
In this case, we say that $\Pi$ \emph{stably decides} (resp., \emph{haltingly
decides}) the predicate
$\psi : \Naturals^{\Sigma} \rightarrow \{ 0, 1 \}$
defined so that
$\psi(\mathbf{x}) = 1$
if and only if
$\mathbf{x}$
is stably (resp., haltingly) accepted by $\Pi$.

By definition, the molecular count of the fuel species $F$ in the initial
configuration $\c^{0}$ does not affect the computation's outcome in terms of
whether the execution stabilizes (resp., halts) with $0$- or $1$-voters.
Consequently, one can increase the molecular count $\c^{0}(F)$ of the fuel
species in the initial configuration $\c^{0}$, thus increasing the initial
(total) molecular count
$n = \| \c^{0} \|$
for any given input vector
$\mathbf{x} \in \Naturals^{\Sigma}$.
Since the runtime of a CRN is expressed in terms of the initial molecular
count $n$, decoupling $\mathbf{x}$ from $n$ allows us to measure the
asymptotic runtime of the protocol while keeping $\mathbf{x}$ fixed.
In this regard, the CRD $\Pi$ is said to be \emph{stabilization speed fault
free} (resp., \emph{halting speed fault free})
\cite{ChenCDS2017speed-faults} if for
every input vector
$\mathbf{x} \in \Naturals^{\Sigma}$,
the family of valid initial configurations
$\c^{0} \in \Naturals^{\Species}$
with
$\c^{0}|_{\Sigma} = \mathbf{x}$
does not admit a stabilization (resp., halting) speed fault (as defined in
\Sect{}~\ref{section:runtime:speed-faults}).

\LongVersion 
Notice though that there is a caveat in the conception that $\c^{0}(F)$ can be
made arbitrarily large:
we can artificially drive the runtime of $\Pi$ (expressed as a function of
$n$) towards
$\RunTime^{\Pi}(n) = \Theta (n)$
simply by introducing an inert fuel species $F$ (i.e., a species that
participates only in void reactions) and ``pumping up'' its initial molecular
count $\c^{0}(F)$.
Indeed, this has the effect of
(1)
scaling the probability for choosing any ``meaningful'' reaction in a given
step as
$1 / n^{2}$;
and
(2)
scaling the time span of each step as
$1 / n$.
Consequently, the temporal cost associated with each round scales linearly
with $n$, whereas the number of rounds necessary for termination is
independent of $n$.

As a remedy, we subsequently allow for arbitrarily large initial molecular
counts $\c^{0}(F)$ of the fuel species $F$ only when we aim for sub-linear
runtime (upper) bounds, that is,
$\RunTime^{\Pi}(n) = o (n)$.
Otherwise, we restrict ourselves to \emph{fuel bounded} CRDs, namely,
CRDs that are subject to the (additional) requirement that
$\c^{0}(F) \leq O (\| \mathbf{x} \|)$,
ensuring that the fuel molecular count does not dominate (asymptotically) the
initial molecular count
$n = \| \c^{0} \|$.
\LongVersionEnd 

\subsection{Semilinear Predicates}
\label{section:crd:semilinear}
A predicate
$\psi : \Naturals^{\Sigma} \rightarrow \{ 0, 1 \}$
is \emph{linear} if there exist a finite set
$\mathcal{A} = \mathcal{A}(\psi) \subset \Naturals^{\Sigma}$
and a vector
$\mathbf{b} = \mathbf{b}(\psi) \in \Naturals^{\Sigma}$
such that
$\psi(\mathbf{x}) = 1$
if and only if
$\mathbf{x} = \mathbf{b} + \sum_{\mathbf{a} \in \mathcal{A}} k_{\mathbf{a}} \mathbf{a}$
for some coefficients
$k_{\mathbf{a}} = k_{\mathbf{a}}(\mathbf{x}) \in \Naturals$,
$\mathbf{a} \in \mathcal{A}$.
A predicate
$\psi : \Naturals^{\Sigma} \rightarrow \{ 0, 1 \}$
is \emph{semilinear} if it is the disjunction of finitely many linear
predicates.
The following theorem is established in the seminal work of Angluin et al.\
\cite{AngluinADFP2006computation,
AngluinAER2007computational}.

\begin{theorem}[%
\cite{AngluinADFP2006computation,
AngluinAER2007computational}]
\label{theorem:crd:semilinear-strong-fairness}
Fix a predicate
$\psi : \Naturals^{\Sigma} \rightarrow \{ 0, 1 \}$.
If $\psi$ is semilinear, then $\psi$ can be haltingly decided under a
strongly fair scheduler by a leaderless CRD.
If $\psi$ can be stably decided by a CRD under a strongly fair scheduler, then
$\psi$ is semilinear.
\end{theorem}

\LongVersion 
Our goal in this section is to
\LongVersionEnd 
\ShortVersion 
In the attached full version, we
\ShortVersionEnd 
extend \Thm{}~\ref{theorem:crd:semilinear-strong-fairness} to weak
fairness which allows us to bound the adversarial runtime of the corresponding
CRDs and establish the following theorem;
notice that the
$O (n)$
runtime bound is asymptotically tight
\LongVersion 
for general semilinear predicates
\LongVersionEnd 
--- see the speed fault freeness
discussion in \Sect{}~\ref{section:crd:detection}.

\begin{theorem} \label{theorem:crd:semilinear}
Fix a predicate
$\psi : \Naturals^{\Sigma} \rightarrow \{ 0, 1 \}$.
If $\psi$ is semilinear, then $\psi$ can be haltingly decided under a weakly
fair scheduler by a leaderless CRD whose halting runtime is
$O (n)$.
If $\psi$ can be stably decided by a CRD under a weakly fair scheduler, then
$\psi$ is semilinear.
\end{theorem}

\LongVersion 
The second claim of \Thm{}~\ref{theorem:crd:semilinear} follows immediately
from \Cor{}~\ref{corollary:weak-correctness-implies-strong-correctness} and
\Thm{}~\ref{theorem:crd:semilinear-strong-fairness}.
For the first claim, we define the following two predicate families (that
also play a crucial role in the proof of
\Thm{}~\ref{theorem:crd:semilinear-strong-fairness}
\cite{AngluinADFP2006computation}):
A predicate
$\psi : \Naturals^{\Sigma} \rightarrow \{ 0, 1 \}$
is a \emph{threshold} predicate if there exist a vector
$\mathbf{a} = \mathbf{a}(\psi) \in \Integers^{\Sigma}$
and a scalar
$b = b(\psi) \in \Integers$
such that
$\psi(\mathbf{x}) = 1$
if and only if
$\mathbf{a} \cdot \mathbf{x} < b$.
A predicate
$\psi : \Naturals^{\Sigma} \rightarrow \{ 0, 1 \}$
is a \emph{modulo} predicate if there exist a vector
$\mathbf{a} = \mathbf{a}(\psi) \in \Integers^{\Sigma}$
a scalar
$b = b(\psi) \in \Integers$
and a scalar
$m = m(\psi) \in \Integers_{> 0}$
such that
$\psi(\mathbf{x}) = 1$
if and only if
$\mathbf{a} \cdot \mathbf{x} = b \bmod{m}$.

A folklore result (see, e.g., \cite{GinsburgS1966semigroups}) states that a
predicate
$\psi : \Naturals^{\Sigma} \rightarrow \{ 0, 1 \}$
is semilinear if and only if it can be obtained from finitely many threshold
and modulo predicates through conjunction, disjunction, and negation
operations.
Consequently, we establish \Thm{}~\ref{theorem:crd:semilinear} by
proving the following three propositions.

\begin{proposition} \label{proposition:crd:semilinear:threshold-predicates}
For every threshold predicate
$\psi : \Naturals^{\Sigma} \rightarrow \{ 0, 1 \}$,
there exists a leaderless CRD that haltingly decides $\psi$ whose halting
runtime is
$O (n)$.
\end{proposition}

\begin{proposition} \label{proposition:crd:semilinear:modulo-predicates}
For every modulo predicate
$\psi : \Naturals^{\Sigma} \rightarrow \{ 0, 1 \}$,
there exists a leaderless CRD that haltingly decides $\psi$ whose halting
runtime is
$O (n)$.
\end{proposition}

\begin{proposition}
\label{proposition:crd:semilinear:closure}
For
$j \in \{ 1, 2 \}$,
let
$\Pi_{j}
=
(\Species_{j}, \Reactions_{j}, \Sigma, \Upsilon_{j, 0}, \Upsilon_{j, 1},
F_{j}, \mathbf{0})$
be a leaderless CRD that haltingly decides the predicate
$\psi_{j} : \Naturals^{\Sigma} \rightarrow \{ 0, 1 \}$.
Let
$\xi : \{ 0, 1 \} \times \{ 0, 1 \} \rightarrow \{ 0, 1 \}$
be a Boolean function and let
$\psi_{\xi} : \Naturals^{\Sigma} \rightarrow \{ 0, 1 \}$
be the predicate defined by setting
$\psi_{\xi}(\mathbf{x}) = \xi(\psi_{1}(\mathbf{x}), \psi_{2}(\mathbf{x}))$.
Then, there exists a leaderless CRD
$\Pi_{\xi}
=
(\Species_{\xi}, \Reactions_{\xi}, \Sigma, \Upsilon_{\xi, 0}, \Upsilon_{\xi, 1},
F_{\xi}, \mathbf{0})$
that haltingly decides $\psi_{\xi}$ whose halting runtime satisfies
$\RunTime_{\Halt}^{\Pi_{\xi}}(n)
\leq
O (\RunTime_{\Halt}^{\Pi_{1}}(n) + \RunTime_{\Halt}^{\Pi_{2}}(n) + n)$.
Moreover, $\Pi_{\xi}$ uses
$|\Species_{\xi}| = |\Species_{1}| + |\Species_{2}| + |\Sigma| + O (1)$
species.
\end{proposition}

\Prop{}\ \ref{proposition:crd:semilinear:threshold-predicates},
\ref{proposition:crd:semilinear:modulo-predicates}, and
\ref{proposition:crd:semilinear:closure} are established
in \Sect{}\ \ref{section:crd:semilinear:threshold},
\ref{section:crd:semilinear:modulo}, and
\ref{section:crd:semilinear:closure}, respectively.
The proofs borrow many ideas from the existing literature (particularly
\cite{AngluinADFP2006computation}) although some adaptations are needed to
accommodate the weak fairness condition as well as for the (adversarial)
runtime analysis.

\subsubsection{Threshold Predicates}
\label{section:crd:semilinear:threshold}
In this section, we establish
\Prop{}~\ref{proposition:crd:semilinear:threshold-predicates} by designing the
promised CRD $\Pi$.
Specifically, given a vector
$\mathbf{a} \in \Integers^{\Sigma}$
and a scalar
$b \in \Integers$,
the (leaderless) CRD
$\Pi
=
(\Species, \Reactions, \Sigma, \Upsilon_{0}, \Upsilon_{1}, F, \mathbf{0})$
haltingly decides the predicate
$\psi : \Naturals^{\Sigma} \rightarrow \{ 0, 1 \}$
defined so that
$\psi(\mathbf{x}) = 1$
if and only if
$\mathbf{a} \cdot \mathbf{x} < b$.
Moreover, the halting runtime of $\Pi$ is
$\RunTime_{\Halt}^{\Pi}(n) = O (n)$.

Taking
$s = \max \left\{ |b| + 1, \max_{A \in \Sigma} |\mathbf{a}(A)| \right\}$,
the species set of protocol $\Pi$ is defined to be
\[\textstyle
\Species
\, = \,
\Sigma
\cup
\{ F \}
\cup
\left\{ L_{u} \mid -s \leq u \leq s \right\}
\cup
\left\{ Y_{-1}, Y_{0}, Y_{+1} \right\}
\]
The species in
$\Sigma \cup \{ F \}$
are regarded as the ignition species of the ignition gadget presented in
\Sect{}~\ref{section:runtime:tools:ignition}, taking the ignition reaction
associated with species
$A \in \Sigma$
to be
\\
$\iota_{A}$:
$A \rightarrow L_{\mathbf{a}(A)}$;
\\
and the ignition reaction associated with species $F$ to be
\\
$\iota_{F}$:
$F \rightarrow L_{0}$.

Semantically, we think of the molecules of the different species as
carrying an abstract charge that may be positive, negative, or neutral
(i.e., zero):
each molecule of species
$A \in \Sigma$
carries
$\chi(A) = \mathbf{a}(A)$
units of charge;
each fuel molecule carries
$\chi(F) = 0$
units of charge;
each molecule of species $L_{u}$,
$-s \leq u \leq s$,
carries
$\chi(L_{u}) = u$
units of charge;
and
each molecule of species $Y_{j}$,
$j \in \{ -1, 0, +1 \}$,
carries
$\chi(Y_{j}) = j$
units of charge.
From this point of view, the ignition reactions can be interpreted as
transferring the charge from the ignition species in
$\Sigma \cup \{ F \}$
to the working species in
$\left\{ L_{u} \mid -s \leq u \leq s \right\}
\cup
\left\{ Y_{-1}, Y_{0}, Y_{+1} \right\}$.

We design the reaction set $\Reactions$ so that the total charge remains
invariant throughout the execution (see
\Obs{}~\ref{observation:crd:semilinear:threshold:charge-invariant}).
Moreover, when the execution halts, there is exactly one $L$ molecule
left (i.e., a leader) and we can determine whether or not the total charge is
below the threshold $b$ based solely on the charge of this $L$ molecule.
Following this logic, the voter species are defined as
\[\textstyle
\Upsilon_{0}
\, = \,
\left\{ L_{u} \mid u \geq b \right\}
\quad \text{and} \quad
\Upsilon_{1}
\, = \,
\left\{ L_{u} \mid u < b \right\}
\, .
\]

Concretely, the non-void reaction set $\NonVoid(\Reactions)$ of protocol $\Pi$
includes the following reactions on top of the aforementioned ignition
reactions:
\begin{itemize}

\item
$\beta_{u, u'}$:
$L_{u} + L_{u'}
\rightarrow L_{u + u'} + Y_{0}$
for every
$-s \leq u, u' \leq s$
such that
$|u + u'| \leq s$;

\item
$\hat{\beta}_{u, u'}$:
$L_{u} + L_{u'} \rightarrow
L_{\Sign(u + u') \cdot s} + (|u + u'| - s) \cdot Y_{\Sign(u + u')}$
for every
$-s \leq u, u' \leq s$
such that
$|u + u'| > s$;

\item
$\gamma$:
$Y_{-1} + Y_{+1} \rightarrow
2 Y_{0}$;
and

\item
$\delta_{u, j}$:
$L_{u} + Y_{j} \rightarrow
L_{u + j} + Y_{0}$
for every
$-s \leq u \leq s$
and
$j \in \{ -1, +1 \}$
such that
$|u + j| \leq s$.

\end{itemize}
In other words, the $\beta$ and $\hat{\beta}$ reactions decrement the number
of $L$ molecules, where the latter reactions introduce an appropriate number
of $Y_{-1}$ or $Y_{+1}$ molecules so as to maintain the total charge;
reaction $\gamma$ cancels a negative unit of charge with a positive unit of
charge held by the $Y$ molecules;
and
the $\delta$ reactions shift a (negative or positive) unit of charge from the
$Y$ molecules to the $L$ molecules.

\subparagraph{Analysis.}
For the analysis of protocol $\Pi$, fix some input vector
$\mathbf{x} \in \Naturals^{\Sigma}$
and let
$\c^{0} \in \Naturals^{\Species}$
be a valid initial configuration of $\Pi$ with
$\c^{0}|_{\Sigma} = \mathbf{x}$.
Consider a weakly fair execution
$\eta = \langle \c^{t}, \zeta^{t} \rangle_{t \geq 0}$
emerging from $\c^{0}$.

\begin{observation}
\label{observation:crd:semilinear:threshold:charge-invariant}
For every step
$t \geq 0$,
we have
$\sum_{A \in \Species} \chi(A) \cdot \c^{t}(A)
=
\mathbf{a} \cdot \mathbf{x}$.
\end{observation}
\begin{proof}
Follows from the design of $\Reactions$ ensuring that
(1)
$\sum_{A \in \Species} \chi(A) \cdot \c^{0}(A)
=
\mathbf{a} \cdot \mathbf{x}$;
and
(2)
$\sum_{A \in \Species} \chi(A) \cdot \c^{t}(A)
=
\sum_{A \in \Species} \chi(A) \cdot \c^{t - 1}(A)$
for every
$t > 0$.
\end{proof}

We make extensive use of the notation
\[\textstyle
\chi^{+}(\c)
\, = \,
\c(Y_{+1}) + \sum_{1 \leq u \leq s} u \cdot \c(L_{u})
\qquad \text{and} \qquad
\chi^{-}(\c)
\, = \,
\c(Y_{-1}) + \sum_{-s \leq u \leq -1} -u \cdot \c(L_{u})
\, ,
\]
as well as
$\c(L) = \c(\{ L_{u} \mid -s \leq u \leq s \})$,
defined for each configuration
$\c \in \Naturals^{\Species}$.
The following steps play a key role in the analysis:
\begin{itemize}

\item
let
$t_{\Ignt} \geq 0$
be the earliest step in which the ignition gadget matures in $\eta$ (as
promised in \Lem{}~\ref{lemma:runtime:tools:ignition});

\item
let $t_{\Sign}$ be the earliest step
$t \geq t_{\Ignt}$
such that
$\chi^{+}(\c^{t}) = 0$
or
$\chi^{-}(\c^{t}) = 0$;

\item
let $t_{\mathrm{leader}}$ be the earliest step
$t \geq t_{\Sign}$
such that
$\c^{t}(L) = 1$;
and

\item
let $t_{\Halt}$
be the earliest step
$t \geq t_{\mathrm{leader}}$
such that $\Applicable(\c^{t})$ does not include any $\delta$ reaction.

\end{itemize}
The existence of steps $t_{\Sign}$, $t_{\mathrm{leader}}$, and $t_{\Halt}$ is
established (implicitly) in the sequel as part of the runtime analysis.
Here, we prove the following three observations.

\begin{observation} \label{observation:crd:semilinear:threshold:t-sign}
$\chi^{+}(\c^{t}) = 0$
or
$\chi^{-}(\c^{t}) = 0$
for all
$t \geq t_{\Sign}$.
In particular, reaction $\gamma$ is inapplicable from step $t_{\Sign}$ onward.
\end{observation}
\begin{proof}
Follows by noticing that
$\chi^{+}(\c^{t + 1}) \leq \chi^{+}(\c^{t})$
and
$\chi^{-}(\c^{t + 1}) \leq \chi^{-}(\c^{t})$
for every
$t \geq 0$.
\end{proof}

\begin{observation} \label{observation:crd:semilinear:threshold:t-leader}
$\c^{t}(L) = 1$
for all
$t \geq t_{\mathrm{leader}}$.
In particular, the $\beta$ and $\hat{\beta}$ reactions are inapplicable from
step $t_{\mathrm{leader}}$ onward.
\end{observation}
\begin{proof}
Reaction $\gamma$ and the $\delta$ reactions do not change
$\c^{t}(L)$,
whereas each application of a $\beta$ or $\hat{\beta}$ reaction decreases
$\c^{t}(L)$
while still producing one $L$ molecule.
The assertion is established by recalling that $L_{0}$ is produced by the fuel
ignition reaction $\iota_{F}$, hence
$\c^{t_{\Ignt}}(L)
\geq
\c^{t_{\Ignt}}(L_{0})
\geq 1$.
\end{proof}

\begin{observation} \label{observation:crd:semilinear:threshold:t-trm}
Exactly one of the following two properties holds:
(1)
$\c^{t_{\Halt}}(L_{u}) = 1$
for some
$-s + 1 \leq u \leq s - 1$
and
$\c^{t_{\Halt}}(Y_{-1}) = \c^{t_{\Halt}}(Y_{+1}) = 0$;
or
(2)
$\c^{t_{\Halt}}(L_{u}) = 1$
for some
$u \in \{ -s, +s \}$
and
$\c^{t_{\Halt}}(Y_{-\Sign(u)}) = 0$.
In particular, the $\delta$ reactions are inapplicable in step $t_{\Halt}$.
\end{observation}
\begin{proof}
By \Obs{}\ \ref{observation:crd:semilinear:threshold:t-sign} and
\ref{observation:crd:semilinear:threshold:t-leader}, from step
$t_{\mathrm{leader}}$ onward, there is a single $L$ molecule $L_{u}$ present in
the configuration and the only non-void reactions that may still be applicable
are the
$\delta_{u, j}$
reactions for
$j = \Sign(u)$.
\end{proof}

\Lem{}~\ref{lemma:runtime:tools:ignition} and \Obs{}\
\ref{observation:crd:semilinear:threshold:t-sign},
\ref{observation:crd:semilinear:threshold:t-leader}, and
\ref{observation:crd:semilinear:threshold:t-trm} imply that $\eta$
halts in step $t_{\Halt}$, thus establishing
\Cor{}~\ref{corollary:crd:semilinear:threshold:correctness} due to
\Obs{}~\ref{observation:crd:semilinear:threshold:charge-invariant}.
The halting correctness of protocol $\Pi$ follows by the choice of
$\Upsilon_{0}$ and $\Upsilon_{1}$.

\begin{corollary} \label{corollary:crd:semilinear:threshold:correctness}
Execution $\eta$ halts in a configuration $\c$ that includes a
single $L$ molecule $L_{u}$ whose index $u$ satisfies:
(1)
if
$|\mathbf{a} \cdot \mathbf{x}| \leq s$,
then
$u = \mathbf{a} \cdot \mathbf{x}$;
and
(2)
if
$|\mathbf{a} \cdot \mathbf{x}| > s$,
then
$u = \Sign(\mathbf{a} \cdot \mathbf{x}) \cdot s$.
\end{corollary}

For the halting runtime analysis, let
$n = \| \c^{0} \|$
denote the molecular count of the initial configuration and fix some
skipping policy $\SkipPol$.
We prove that
$\RunTime_{\Halt}^{\Pi}(n) \leq O (n)$
by presenting a runtime policy $\RunPol$ for $\Pi$ (defined independently of
$\eta$ and $\SkipPol$) and showing that
$\RunTime_{\Halt}^{\RunPol, \SkipPol}(\eta) \leq O (n)$.
Given a configuration
$\c \in \Naturals^{\Species}$,
the runtime policy $\RunPol$ is defined as follows:
\begin{itemize}

\item
if
$\c(\Sigma \cup \{ F \}) > 0$,
then $\RunPol(\c)$ consists of the ignition reactions;

\item
else if
$\chi^{+}(\c) > 0$
and
$\chi^{-}(\c) > 0$,
then
$\RunPol(\c)
=
\left\{ \beta_{u, u'} \mid \Sign(u) \cdot \Sign(u') = -1 \right\}
\cup
\left\{ \gamma \right\}
\cup
\left\{ \delta_{u, j} \mid \Sign(u) \cdot \Sign(j) = -1 \right\}$;

\item
else if
$\c(L) > 1$,
then $\RunPol(\c)$ consists of the $\beta$ and $\hat{\beta}$ reactions;

\item
else $\RunPol(\c)$ consists of the $\delta$ reactions.

\end{itemize}

Let $\EffStep(i)$ and
$\EffConf^{i} = \c^{\EffStep(i)}$
be the effective step and effective configuration, respectively, of round
$i \geq 0$
under $\RunPol$ and $\SkipPol$.
We establish the desired upper bound on
$\RunTime_{\Halt}^{\RunPol, \SkipPol}(\eta)$
by introducing the following four rounds:
\begin{itemize}

\item
$i_{\Ignt}
=
\min \{ i \geq 0 \mid \EffStep(i) \geq t_{\Ignt} \}$;

\item
$i_{\Sign}
=
\min \{ i \geq i_{\Ignt} \mid \EffStep(i) \geq t_{\Sign} \}$;

\item
$i_{\mathrm{leader}}
=
\min \{ i \geq i_{\Sign} \mid \EffStep(i) \geq t_{\mathrm{leader}} \}$;
and

\item
$i_{\Halt}
=
\min \{ i \geq i_{\mathrm{leader}} \mid \EffStep(i) \geq t_{\Halt} \}$.

\end{itemize}
\Lem{}~\ref{lemma:runtime:tools:ignition} guarantees that the total
contribution of rounds
$0 \leq i < i_{\Ignt}$
to
$\RunTime_{\Halt}^{\RunPol, \SkipPol}(\eta)$
is up-bounded by
$O (\log n)$.

For the contribution of the subsequent rounds to
$\RunTime_{\Halt}^{\RunPol, \SkipPol}(\eta)$,
we need to define the following notation:
Let
$K^{+}_{i} \in \{ 0, 1 \}^{s}$
and
$K^{-}_{i} \in \{ 0, 1 \}^{s}$
be the binary vectors defined so that
$K^{+}_{i}(u) = \Indicator_{\EffConf^{i}(u) > 0}$
and
$K^{-}_{i}(u) = \Indicator_{\EffConf^{i}(-u) > 0}$
for each
$1 \leq u \leq s$.
Let $\prec$ denote the lexicographic (strict) order over
$\{ 0, 1 \}^{s}$
in decreasing index significance;
that is, for every
$\mathbf{f}, \mathbf{g} \in \{ 0, 1 \}^{s}$,
the relation
$\mathbf{f} \prec \mathbf{g}$
holds if and only if there exists an integer
$1 \leq u \leq s$
such that
$\mathbf{f}(u) < \mathbf{g}(u)$
and
$\mathbf{f}(u') = \mathbf{g}(u')$
for every
$u < u' \leq s$.
Define the binary relations $\succ$ and $\succeq$ over
$\{ 0, 1 \}^{s}$
so that
$\mathbf{f} \succ \mathbf{g}
\Longleftrightarrow
\mathbf{g} \prec \mathbf{f}$
and
$\mathbf{f} \succeq \mathbf{g}
\Longleftrightarrow
[\mathbf{f} \succ \mathbf{g} \lor \mathbf{f} = \mathbf{g}]$.
We can now establish the following three lemmas.

\begin{lemma} \label{lemma:crd:semilinear:threshold:temp-cost-sign}
The total contribution of rounds
$i_{\Ignt} \leq i < i_{\Sign}$
to
$\RunTime_{\Halt}^{\RunPol, \SkipPol}(\eta)$
is up-bounded by
$O (n)$.
\end{lemma}
\begin{proof}
Fix a round
$i_{\Ignt} \leq i < i_{\Sign}$
and recall that the runtime policy $\RunPol$ is designed so that
$\RunPol(\EffConf^{i})
=
Q
=
\{ \beta_{u, u'} \mid \Sign(u) \cdot \Sign(u') = -1 \}
\cup
\{ \gamma \}
\cup
\{ \delta_{u, j} \mid \Sign(u) \cdot \Sign(j) = -1 \}$.
Consider a configuration $\c$ reachable from $\EffConf^{i}$ via a
$\RunPol$-avoiding path.
Since each reactant of a $Q$ reaction carries at least $1$ and at most $s$
units of charge, it follows that the propensity $\Prp_{\c}(Q)$ satisfies
\[\textstyle
\frac{\chi^{+}(\c) \cdot \chi^{-}(\c)}{\Vol \cdot s^{2}}
\, \leq \,
\Prp_{\c}(Q)
\, \leq \,
\frac{\chi^{+}(\c) \cdot \chi^{-}(\c)}{\Vol}
\, ,
\]
thus
$\Prp_{\c}(Q)
=
\Theta \left( \frac{\chi^{+}(\c) \cdot \chi^{-}(\c)}{n} \right)$
as
$s = O (1)$.
Inspecting the reactions in
$\Reactions - Q$,
we deduce that
$\chi^{+}(\c) = \chi^{+}(\EffConf^{i})$
and
$\chi^{-}(\c) = \chi^{-}(\EffConf^{i})$,
hence we can employ \Lem{}~\ref{lemma:runtime:tools:temporal-cost} to conclude
that the contribution of round $i$ to
$\RunTime_{\Halt}^{\RunPol}(\eta)$
is
$\TemporalCost^{\RunPol}(\EffConf^{i})
\leq
O \left( \frac{n}{\chi^{+}(\EffConf^{i}) \cdot \chi^{-}(\EffConf^{i})} \right)$.

If round $i$ is target-accomplished, then
$\chi^{+}(\EffConf^{i + 1}) < \chi^{+}(\EffConf^{i})$
and
$\chi^{-}(\EffConf^{i + 1}) < \chi^{-}(\EffConf^{i})$
This is no longer guaranteed if round $i$ is target-deprived, however, we
argue that if
$\chi^{+}(\EffConf^{i'}) = \chi^{+}(\EffConf^{i})$
or
$\chi^{-}(\EffConf^{i'}) = \chi^{-}(\EffConf^{i})$
for some
$i' > i$,
then
$i' - i \leq 2^{s} = O (1)$.\footnote{%
Using a more delicate argument, one can improve this bound to
$i' - i \leq O (s)$.}
Indeed, if
$\chi^{+}(\EffConf^{i + 1}) = \chi^{+}(\EffConf^{i})$
and
$\chi^{-}(\EffConf^{i + 1}) = \chi^{-}(\EffConf^{i})$,
then
$K^{+}_{i + 1} \succeq K^{+}_{i}$
and
$K^{-}_{i + 1} \succeq K^{-}_{i}$,
while at least one of the two relations must be strict.

Taking
$\ell_{i} = \min \{ \chi^{+}(\EffConf^{i}), \chi^{-}(\EffConf^{i}) \}$
for each
$i_{\Ignt} \leq i < i_{\Sign}$,
we conclude that
\\
(1)
$\TemporalCost^{\RunPol}(\EffConf^{i})
\leq
O (n / \ell_{i}^{2})$;
\\
(2)
$\ell_{i + 1} \leq \ell_{i}$;
and
\\
(3)
there exists a constant
$h \geq 1$
such that
$\ell_{i + h} < \ell_{i}$.
\\
As
$\ell_{0} \leq n \cdot s$,
we can bound the total contribution of rounds
$i_{\Ignt} \leq i < i_{\Sign}$
to
$\RunTime_{\Halt}^{\RunPol, \SkipPol}(\eta)$
by
\[\textstyle
\sum_{j = 1}^{n \cdot s} O (n / j^{2})
\, \leq \,
O (n) \cdot \sum_{j = 1}^{\infty} 1 / j^{2}
\, \leq \,
O (n)
\, ,
\]
thus establishing the assertion.
\end{proof}

\begin{lemma} \label{lemma:crd:semilinear:threshold:temp-cost-leader}
The total contribution of rounds
$i_{\Sign} \leq i < i_{\mathrm{leader}}$
to
$\RunTime_{\Halt}^{\RunPol, \SkipPol}(\eta)$
is up-bounded by
$O (n)$.
\end{lemma}
\begin{proof}
Fix a round
$i_{\Sign} \leq i < i_{\mathrm{leader}}$
and assume without loss of generality that
$\chi^{-}(\EffConf^{i}) = 0$
(the case where
$\chi^{+}(\EffConf^{i}) = 0$
is proved symmetrically).
Recall that the runtime policy $\RunPol$ is designed so that
$\RunPol(\EffConf^{i})
=
Q
=
\{ \beta_{u, u'}, \hat{\beta}_{u, u'} \mid -s \leq u, u' \leq s \}$.
Consider a configuration $\c$ reachable from $\EffConf^{i}$ via a
$\RunPol$-avoiding path and notice that the propensity $\Prp_{\c}(Q)$
satisfies
$\Prp_{\c}(Q)
\geq
\Omega((\c(L))^{2} / n)$.
Inspecting the reactions in
$\Reactions - Q$,
we deduce that
$\c(L) = \EffConf^{i}(L)$,
hence we can employ \Lem{}~\ref{lemma:runtime:tools:temporal-cost} to conclude
that the contribution of round $i$ to
$\RunTime_{\Halt}^{\RunPol, \SkipPol}(\eta)$
is
$\TemporalCost^{\RunPol}(\EffConf^{i})
\leq
O (n / (\EffConf^{i}(L))^{2})$.

If round $i$ is target-accomplished, then
$\EffConf^{i + 1}(L) < \EffConf^{i}(L)$.
This is no longer the case if round $i$ is target-deprived, however, we argue
that if
$\chi^{+}(\EffConf^{i'}) = \chi^{+}(\EffConf^{i})$
for some
$i' > i$,
then
$i' - i \leq 2^{s} = O (1)$.\footnote{%
Using a more delicate argument, one can improve this bound to
$i' - i \leq O (s^{2})$.}
Indeed, if
$\chi^{+}(\EffConf^{i + 1}) = \chi^{+}(\EffConf^{i})$,
then
$K^{+}_{i + 1} \succ K^{+}_{i}$.

Taking
$\ell_{i} = \EffConf^{i}(L)$
for each
$i_{\Sign} \leq i < i_{\mathrm{leader}}$,
we conclude that
\\
(1)
$\TemporalCost^{\RunPol}(\EffConf^{i})
\leq
O (n / \ell_{i}^{2})$;
\\
(2)
$\ell_{i + 1} \leq \ell_{i}$;
and
\\
(3)
there exists a constant
$h \geq 1$
such that
$\ell_{i + h} < \ell_{i}$.
\\
As
$\ell_{i_{\Sign}} \leq n$,
we can bound the total contribution of rounds
$i_{\Sign} \leq i < i_{\mathrm{leader}}$
to
$\RunTime_{\Halt}^{\RunPol, \SkipPol}(\eta)$
by
\[\textstyle
\sum_{j = 1}^{n} O (n / j^{2})
\, \leq \,
O (n) \cdot \sum_{j = 1}^{\infty} 1 / j^{2}
\, \leq \,
O (n)
\, ,
\]
thus establishing the assertion.
\end{proof}

\begin{lemma} \label{lemma:crd:semilinear:threshold:temp-cost-trm}
The total contribution of rounds
$i_{\mathrm{leader}} \leq i < i_{\Halt}$
to
$\RunTime_{\Halt}^{\RunPol, \SkipPol}(\eta)$
is up-bounded by
$O (n)$.
\end{lemma}
\begin{proof}
\Obs{}\ \ref{observation:crd:semilinear:threshold:t-sign} and
\ref{observation:crd:semilinear:threshold:t-leader} imply that
$\Applicable(\EffConf^{i})$
consists only of $\delta$ reactions for each round
$i_{\mathrm{leader}} \leq i < i_{\Halt}$.
Since
$i_{\mathrm{leader}} \geq i_{\Sign}$,
it follows that there are at most
$s - 1 = O (1)$
such rounds.
The assertion follows as each round contributes at most
$O (n)$
temporal cost to
$\RunTime_{\Halt}^{\RunPol, \SkipPol}(\eta)$.
\end{proof}

Combining \Lem{}~\ref{lemma:runtime:tools:ignition} with
\Lem{}\ \ref{lemma:crd:semilinear:threshold:temp-cost-sign},
\ref{lemma:crd:semilinear:threshold:temp-cost-leader}, and
\ref{lemma:crd:semilinear:threshold:temp-cost-trm},
we conclude that
$\RunTime_{\Halt}^{\RunPol, \SkipPol}(\eta)
=
O (n)$,
which yields \Prop{}~\ref{proposition:crd:semilinear:threshold-predicates}.

\subsubsection{Modulo Predicates}
\label{section:crd:semilinear:modulo}
In this section, we establish
\Prop{}~\ref{proposition:crd:semilinear:modulo-predicates} by designing the
promised CRD $\Pi$.
Specifically, given a vector
$\mathbf{a} \in \Integers^{\Sigma}$
and scalars
$b \in \Integers$
and
$m \in \Integers_{> 0}$,
the (leaderless) CRD
$\Pi
=
(\Species, \Reactions, \Sigma, \Upsilon_{0}, \Upsilon_{1}, F, \mathbf{0})$
haltingly decides the predicate
$\psi : \Naturals^{\Sigma} \rightarrow \{ 0, 1 \}$
defined so that
$\psi(\mathbf{x}) = 1$
if and only if
$\mathbf{a} \cdot \mathbf{x} = b \bmod{m}$.
Moreover, the halting runtime of $\Pi$ is
$\RunTime_{\Halt}^{\Pi}(n) = O (n)$.

The species set of protocol $\Pi$ is defined to be
\[\textstyle
\Species
\, = \,
\Sigma
\cup
\{ F \}
\cup
\left\{ L_{u} \mid 0 \leq u \leq m - 1 \right\}
\cup
\left\{ Y \right\}
\, .
\]
The species in
$\Sigma \cup \{ F \}$
are regarded as the ignition species of the ignition gadget presented in
\Sect{}~\ref{section:runtime:tools:ignition}, taking the ignition reaction
associated with species
$A \in \Sigma$
to be
\\
$\iota_{A}$:
$A \rightarrow L_{\mathbf{a}(A) \bmod m}$;
\\
and the ignition reaction associated with species $F$ to be
\\
$\iota_{F}$:
$F \rightarrow L_{0}$.

Semantically, we think of the molecules of species
$A \in \Sigma$
as carrying $\mathbf{a}(A)$ units of an abstract charge.
Each molecule of species $L_{u}$,
$0 \leq u \leq m - 1$,
encodes the consumption of $\chi$ units of charge for some
$\chi = u \bmod m$,
whereas the $F$ and $Y$ molecules carry a neutral charge.
From this point of view, the ignition reactions can be interpreted as
transferring the charge (modulo $m$) from the ignition species to the working
species.

We design the reaction set $\Reactions$ so that the total charge remains
invariant modulo $m$ throughout the execution (see
\Obs{}~\ref{observation:crd:semilinear:modulo:charge-invariant}).
Moreover, when the execution halts, there is exactly one $L$ molecule left
(i.e., a leader) and we can determine whether or not the total charge modulo
$m$ is $b$ based solely on the species of the remaining $L$ molecule.
Following this logic, the voter species are defined as
\[\textstyle
\Upsilon_{0}
\, = \,
\left\{ L_{u} \mid u \neq b \right\}
\quad \text{and} \quad
\Upsilon_{1}
\, = \,
\left\{ L_{b} \right\}
\, .
\]

Concretely, the non-void reaction set $\NonVoid(\Reactions)$ of protocol $\Pi$
includes the following reactions on top of the aforementioned ignition
reactions:
\\
$\beta_{u, u'}$:
$L_{u} + L_{u'} \rightarrow
L_{u + u' \bmod m} + Y$
for every
$0 \leq u, u' \leq m - 1$.
\\
In other words, the $\beta$ reactions decrement the number of $L$ molecules
while maintaining the total charge modulo $m$.

\subparagraph{Analysis.}
For the analysis of protocol $\Pi$, fix some input vector
$\mathbf{x} \in \Naturals^{\Sigma}$
and let
$\c^{0} \in \Naturals^{\Species}$
be a valid initial configuration of $\Pi$ with
$\c^{0}|_{\Sigma} = \mathbf{x}$.
Consider a weakly fair execution
$\eta = \langle \c^{t}, \zeta^{t} \rangle_{t \geq 0}$
emerging from $\c^{0}$.

\begin{observation}
\label{observation:crd:semilinear:modulo:charge-invariant}
For every step
$t \geq 0$,
we have
$\sum_{A \in \Sigma} \mathbf{a}(A) \cdot \c^{t}(A)
+
\sum_{0 \leq u \leq m - 1} u \cdot \c^{t}(L_{u})
=
\mathbf{a} \cdot \mathbf{x}
\bmod m$.
\end{observation}
\begin{proof}
Follows from the design of $\Reactions$, ensuring that
(1)
$\sum_{A \in \Sigma} \mathbf{a}(A) \cdot \c^{0}(A)
=
\mathbf{a} \cdot \mathbf{x}
\bmod m$;
and
(2)
$\sum_{A \in \Sigma} \mathbf{a}(A) \cdot \c^{t}(A)
+
\sum_{0 \leq u \leq m - 1} u \cdot \c^{t}(L_{u})
=
\sum_{A \in \Sigma} \mathbf{a}(A) \cdot \c^{t - 1}(A)
+
\sum_{0 \leq u \leq m - 1} u \cdot \c^{t - 1}(L_{u})
\bmod m$
for every
$t > 0$.
\end{proof}

We subsequently use the notation
$\c(L) = \sum_{0 \leq u \leq m - 1} \c(L_{u})$
defined for each configuration
$\c \in \Naturals^{\Species}$.
The following steps play a key role in the analysis:
\begin{itemize}

\item
let
$t_{\Ignt} \geq 0$
be the earliest step in which the ignition gadget matures in $\eta$ (as
promised in \Lem{}~\ref{lemma:runtime:tools:ignition});
and

\item
let $t_{\mathrm{leader}}$ be the earliest step
$t \geq t_{\Ignt}$
such that $\c^{t}(L) = 1$.

\end{itemize}
The existence of step $t_{\mathrm{leader}}$ is established (implicitly) in the
sequel as part of the runtime analysis.
Here, we prove the following observation.

\begin{observation} \label{observation:crd:semilinear:modulo:t-beta}
$\c^{t}(L) = 1$
for all
$t \geq t_{\mathrm{leader}}$.
In particular, the $\beta$ reactions are inapplicable from step
$t_{\mathrm{leader}}$ onward.
\end{observation}
\begin{proof}
Each application of a $\beta$ reaction decreases $\c^{t}(L)$ while still
producing one $L$ molecule.
The assertion is established by recalling that $L_{0}$ is produced by the fuel
ignition reaction $\iota_{F}$, hence
$\c^{t_{\Ignt}}(L)
\geq
\c^{t_{\Ignt}}(L_{0})
\geq 1$.
\end{proof}

\Lem{}~\ref{lemma:runtime:tools:ignition} and
\Obs{}~\ref{observation:crd:semilinear:modulo:t-beta} imply that $\eta$ halts
in step $t_{\mathrm{leader}}$, thus establishing
\Cor{}~\ref{corollary:crd:semilinear:modulo:correctness} due to
\Obs{}~\ref{observation:crd:semilinear:modulo:charge-invariant}.
The halting correctness of protocol $\Pi$ follows by the choice of
$\Upsilon_{0}$
and
$\Upsilon_{1}$.

\begin{corollary} \label{corollary:crd:semilinear:modulo:correctness}
Execution $\eta$ halts in a configuration $\c$ that includes a
single $L$ molecule $L_{u}$ whose index $u$ satisfies
$u = \mathbf{a} \cdot \mathbf{x} \bmod m$.
\end{corollary}

For the halting runtime analysis, let
$n = \| \c^{0} \|$
denote the molecular count of the initial configuration and fix some skipping
policy $\SkipPol$.
We prove that
$\RunTime_{\Halt}^{\Pi}(n) \leq O (n)$
by presenting a runtime policy $\RunPol$ for $\Pi$ (defined independently of
$\eta$ and $\SkipPol$) and showing that
$\RunTime_{\Halt}^{\RunPol, \SkipPol}(\eta) \leq O (n)$.
Given a configuration
$\c \in \Naturals^{\Species}$,
the runtime policy $\RunPol$ is defined as follows:
\begin{itemize}

\item
if
$\c(\Sigma \cup \{ F \}) > 0$,
then $\RunPol(\c)$ consists of the ignition reactions;

\item
else $\RunPol(\c)$ consists of the $\beta$ reactions.

\end{itemize}

Let $\EffStep(i)$ and
$\EffConf^{i} = \c^{\EffStep(i)}$
be the effective step and effective configuration, respectively, of round
$i \geq 0$
under $\RunPol$ and $\SkipPol$.
We establish the desired upper bound on
$\RunTime_{\Halt}^{\RunPol, \SkipPol}(\eta)$
by introducing the following two rounds:
\begin{itemize}

\item
$i_{\Ignt}
=
\min \{ i \geq 0 \mid \EffStep(i) \geq t_{\Ignt} \}$;
and

\item
$i_{\mathrm{leader}}
=
\min \{ i \geq i_{\Ignt} \mid \EffStep(i) \geq t_{\mathrm{leader}} \}$.

\end{itemize}
\Lem{}~\ref{lemma:runtime:tools:ignition} guarantees that the total
contribution of rounds
$0 \leq i < i_{\Ignt}$
to
$\RunTime_{\Halt}^{\RunPol, \SkipPol}(\eta)$
is up-bounded by
$O (\log n)$.
For the contribution of the subsequent rounds to
$\RunTime_{\Halt}^{\RunPol, \SkipPol}(\eta)$,
we establish the following lemma.

\begin{lemma} \label{lemma:crd:semilinear:modulo:temp-cost-beta}
The total contribution of rounds
$i_{\Ignt} \leq i < i_{\mathrm{leader}}$
to
$\RunTime_{\Halt}^{\RunPol, \SkipPol}(\eta)$
is up-bounded by
$O (n)$.
\end{lemma}
\begin{proof}
Fix a round
$i_{\Ignt} \leq i \leq i_{\beta}$
and recall that the runtime policy $\RunPol$ is designed so that
$\RunPol(\EffConf^{i})
=
Q
=
\{ \beta_{u, u'} \mid 0 \leq u, u' \leq m - 1 \}$.
Since
$Q = \NonVoid(\Applicable(\EffConf^{i}))$,
it follows that $\EffConf^{i}$ is the only configuration reachable from
$\EffConf^{i}$ via a $\RunPol$-avoiding path.
Let
$\ell_{i} = \EffConf^{i}(L)$.
As
$\Prp_{\EffConf^{i}}(Q)
=
\frac{1}{\Vol} \cdot \binom{\ell_{i}}{2}
\geq
\Omega (\ell_{i}^{2} / n)$,
we can employ \Lem{}~\ref{lemma:runtime:tools:temporal-cost} to conclude that
the contribution of round $i$ to
$\RunTime_{\Halt}^{\RunPol, \SkipPol}(\eta)$
is
$\TemporalCost^{\RunPol}(\EffConf^{i}) \leq O (n / \ell_{i}^{2})$.
Observing that
$\ell_{i + 1} < \ell_{i}$
and
$\ell_{0} \leq n$,
we can bound the total contribution of rounds
$i_{\Ignt} \leq i < i_{\mathrm{leader}}$
to
$\RunTime_{\Halt}^{\RunPol, \SkipPol} (\eta)$
by
\[\textstyle
\sum_{\ell = 2}^{n} O \left( \frac{n}{\ell^{2}} \right)
\, \leq \,
O (n) \cdot \sum_{\ell = 1}^{\infty} \frac{1}{\ell^{2}}
\, \leq \,
O (n)
\, ,
\]
thus establishing the assertion.
\end{proof}

Combining \Lem{}~\ref{lemma:runtime:tools:ignition} with
\Lem{}~\ref{lemma:crd:semilinear:modulo:temp-cost-beta}, we conclude that
$\RunTime_{\Halt}^{\RunPol, \SkipPol}(\eta)
=
O (n)$,
which yields \Prop{}~\ref{proposition:crd:semilinear:modulo-predicates}.

\subsubsection{Closure under Boolean Operations}
\label{section:crd:semilinear:closure}
In this section, we establish
\Prop{}~\ref{proposition:crd:semilinear:closure} by
designing the promised CRD $\Pi_{\xi}$.
Intuitively, we employ the ignition gadget to produce two copies of each input
molecule, one for protocol $\Pi_{1}$ and one for protocol $\Pi_{2}$;
following that, the two protocols run in parallel, each on its own
molecules.
The ignition gadget is further employed to produce ``global voter'' molecules
whose role is to interact with the ``local voter'' molecules of $\Pi_{1}$ and
$\Pi_{2}$, recording their votes.
To ensure that the runtime overhead is
$O (1)$,
we invoke a leader election process on the global voters so that a single
global voter survives.

Formally, for
$j \in \{ 1, 2 \}$,
let
$\Pi'_{j}
=
(\Species'_{j}, \Reactions'_{j}, \Sigma'_{j}, \Upsilon'_{j, 0}, \Upsilon'_{j, 1},
F'_{j}, \mathbf{0})$
be the (leaderless) CRD derived from $\Pi_{j}$ by replacing each species
$A \in \Species_{j}$
with a $\Pi'_{j}$ designated species
$A'_{j} \in \Species'_{j}$;
in particular, the CRD $\Pi'_{j}$ is defined over the input species
$\Sigma'_{j} = \{ A'_{j} \mid A \in \Sigma \}$.

The species set $\Species_{\xi}$ of protocol $\Pi_{\xi}$ is defined to be
\[\textstyle
\Species_{\xi}
\, = \,
\Sigma
\cup
\{ F_{\xi} \}
\cup
\Species'_{1}
\cup
\Species'_{2}
\cup
\{ G_{0, 0}, G_{0, 1}, G_{1, 0}, G_{1, 1}, W \}
\, .
\]
The species in
$\Sigma \cup \{ F_{\xi} \}$
are regarded as the ignition species of the ignition gadget presented in
\Sect{}~\ref{section:runtime:tools:ignition}, taking the ignition reaction
associated with species
$A \in \Sigma$
to be
\\
$\iota_{A}$:
$A \rightarrow A'_{1} + A'_{2}$;
\\
and the ignition reaction associated with species $F_{\xi}$ to be
\\
$\iota_{F_{\xi}}$:
$F_{\xi} \rightarrow F'_{1} + F'_{2} + G_{0, 0}$.

On top of the aforementioned ignition reactions, the non-void reaction set
$\NonVoid(\Reactions_{\xi})$ of protocol $\Pi_{\xi}$ consists of the
(non-void) reactions in
$\NonVoid(\Reactions'_{1}) \cup \NonVoid(\Reactions'_{2})$
as well as the following reactions:
\begin{itemize}

\item
$\beta_{u_{1}, u_{2}, w_{1}, w_{2}}$:
$G_{u_{1}, u_{2}} + G_{w_{1}, w_{2}} \rightarrow G_{0, 0} + W$
for every
$u_{1}, u_{2}, w_{1}, w_{2} \in \{ 0, 1 \}$;

\item
$\gamma_{u_{1}, u_{2}}^{V'_{1}}$:
$G_{u_{1}, u_{2}} + V'_{1} \rightarrow G_{1 - u_{1}, u_{2}} + V'_{1}$
for every
$u_{1}, u_{2} \in \{ 0, 1 \}$
and
$V'_{1} \in \Upsilon'_{1, 1 - u_{1}}$;
and

\item
$\gamma_{u_{1}, u_{2}}^{V'_{2}}$:
$G_{u_{1}, u_{2}} + V'_{2} \rightarrow G_{u_{1}, 1 - u_{2}} + V'_{2}$
for every
$u_{1}, u_{2} \in \{ 0, 1 \}$
and
$V'_{2} \in \Upsilon'_{2, 1 - u_{2}}$.

\end{itemize}

Finally, the voter species of $\Pi_{\xi}$ are defined as
\[\textstyle
\Upsilon_{\xi, 0}
\, = \,
\left\{ G_{u_{1}, u_{2}} \mid \xi(u_{1}, u_{2}) = 0 \right\}
\qquad \text{and} \qquad
\Upsilon_{\xi, 1}
\, = \,
\left\{ G_{u_{1}, u_{2}} \mid \xi(u_{1}, u_{2}) = 1 \right\}
\, .
\]

\subparagraph{Analysis.}
For the analysis of protocol $\Pi_{\xi}$, fix some input vector
$\mathbf{x} \in \Naturals^{\Sigma}$
and let
$\c^{0} \in \Naturals^{\Species_{\xi}}$
be a valid initial configuration with
$\c^{0}|_{\Sigma} = \mathbf{x}$.
Consider a weakly fair execution
$\eta = \langle \c^{t}, \zeta^{t} \rangle_{t \geq 0}$
of $\Pi_{\xi}$ emerging from $\c^{0}$.

Let
$t_{\Ignt} \geq 0$
be the earliest step in which the ignition gadget matures in $\eta$ (as
promised in \Lem{}~\ref{lemma:runtime:tools:ignition}).
For
$j \in \{ 1, 2 \}$,
let $t_{j}$ be the earliest step
$t \geq t_{\Ignt}$
such that
$\Applicable(\c^{t}) \cap \NonVoid(\Reactions'_{j}) = \emptyset$.
The halting correctness of $\Pi'_{j}$ and the fact that the species in
$\Species'_{j}$ are catalysts for any reaction in
$\Reactions - \Reactions'_{j}$
ensure that $t_{j}$ exists.
They also yield the following observation.

\begin{observation} \label{observation:crd:semilinear:closure:local-halting}
For each
$j \in \{ 1, 2 \}$,
we have
$\c^{t_{j}}(\Upsilon_{j, v}) > 0$
and
$\c^{t_{j}}(\Upsilon_{j, 1 - v}) = 0$,
where
$v = \psi_{j}(\mathbf{x})$.
Moreover,
$\c^{t}|_{\Species'_{j}} = \c^{t_{j}}|_{\Species'_{j}}$
for every
$t \geq t_{j}$.
\end{observation}

Let
$t_{\max} = \max \{ t_{1}, t_{2} \}$.
The only non-void reactions that can be applicable from step $t_{\max}$ onward
are the $\beta$ and $\gamma$ reactions.
Moreover, \Obs{}~\ref{observation:crd:semilinear:closure:local-halting}
implies that the number of $\gamma$ reactions that can be scheduled between
any two consecutive $\beta$ reactions is up-bounded by a linear function of
the molecular count of the $G$ species.
Since each $\beta$ reaction decreases the molecular count of the $G$ species
and since this molecular count never increases, it follows that there exists a
step
$t_{\mathrm{leader}} \geq t_{\max}$
such that
$\c^{t_{\mathrm{leader}}}(\{ G_{0, 0}, G_{0, 1}, G_{1, 0}, G_{1, 1}\}) = 1$.

From step $t_{\mathrm{leader}}$ onward, the only non-void reactions that can
be applicable are $\gamma$ reactions and these can be scheduled at most twice
(in total) until $\eta$ reaches a halting configuration in step
$t^{*} \geq t_{\mathrm{leader}}$.
\Cor{}~\ref{corollary:crd:semilinear:closure:correctness} follows by the choice
of
$\Upsilon_{\xi, 0}$
and
$\Upsilon_{\xi, 1}$.

\begin{corollary} \label{corollary:crd:semilinear:closure:correctness}
Execution $\eta$ halts in a configuration that includes a single $G$ molecule
that belongs to
$\Upsilon_{\xi, v}$,
where
$v = \xi(\psi_{1}(\mathbf{x}), \psi_{2}(\mathbf{x}))$.
\end{corollary}

For the halting runtime analysis, let
$n = \| \c^{0} \|$
denote the molecular count of the initial configuration and fix some
skipping policy $\SkipPol$.
For
$j \in \{ 1, 2 \}$,
let $\RunPol_{j}$ be a runtime policy for the CRD $\Pi_{j}$ that realizes
$\RunTime_{\Halt}^{\Pi_{j}}(n)$
and let $\RunPol'_{j}$ be the runtime policy for $\Pi'_{j}$ derived from
$\RunPol_{j}$ by replacing each species
$A \in \Species_{j}$
with the $\Pi'_{j}$ designated species $A'_{j}$.
We shall bound
$\RunTime_{\Halt}^{\Pi_{\xi}}(n)$
by introducing a runtime policy $\RunPol$ for $\Pi_{\xi}$ (defined
independently of $\eta$ and $\SkipPol$) and showing that
$\RunTime_{\Halt}^{\RunPol, \SkipPol}(\eta)
\leq
O (\RunTime_{\Halt}^{\Pi_{1}}(n) + \RunTime_{\Halt}^{\Pi_{2}}(n) + n)$.

The runtime policy $\RunPol$ is defined as follows for every configuration
$\c \in \Naturals^{\Species_{\xi}}$:
\begin{itemize}

\item
if
$\c(\Sigma \cup \{ F_{\xi} \}) > 0$,
then $\RunPol(\c)$ consists of the ignition reactions;

\item
else if
$\c|_{\Species'_{1}}$
is not a halting configuration of $\Pi'_{1}$, then
$\RunPol(\c) = \RunPol'_{1}(\c|_{\Species'_{1}})$;

\item
else if
$\c|_{\Species'_{2}}$
is not a halting configuration of $\Pi'_{2}$, then
$\RunPol(\c) = \RunPol'_{2}(\c|_{\Species'_{2}})$;

\item
else if
$\c(\{ G_{0, 0}, G_{0, 1}, G_{1, 0}, G_{1, 1} \}) > 1$,
then $\RunPol(\c)$ consists of the $\beta$ reactions;

\item
else $\RunPol(\c)$ consists of the $\gamma$ reactions.

\end{itemize}

Let $\EffStep(i)$ and
$\EffConf^{i} = \c^{\EffStep(i)}$
be the effective step and effective configuration, respectively, of round
$i \geq 0$
under $\RunPol$ and $\SkipPol$.
We establish the desired upper bound on
$\RunTime_{\Halt}^{\RunPol, \SkipPol}(\eta)$
by introducing the following four rounds:
\begin{itemize}

\item
$i_{\Ignt}
=
\min \{ i \geq 0 \mid \EffStep(i) \geq t_{\Ignt} \}$;

\item
$i_{\max} = \min \{ i \geq i_{\Ignt} \mid \EffStep(i) \geq t_{\max} \}$;

\item
$i_{\mathrm{leader}}
=
\min \{ i \geq i_{\max} \mid \EffStep(i) \geq t_{\mathrm{leader}} \}$;
and

\item
$i^{*} = \min \{ i \geq i_{\mathrm{leader}} \mid \EffStep(i) \geq t^{*} \}$.

\end{itemize}
\Lem{}~\ref{lemma:runtime:tools:ignition} guarantees that the total
contribution of rounds
$0 \leq i < i_{\Ignt}$
to
$\RunTime_{\Halt}^{\RunPol}(\eta)$
is up-bounded by
$O (\log n)$.
For the contribution of the subsequent rounds to
$\RunTime_{\Halt}^{\RunPol}(\eta)$,
we establish the following three lemmas.

\begin{lemma} \label{lemma:crd:semilinear:closure:bound-contribution-i-max}
The total contribution of rounds
$i_{\Ignt} \leq i < i_{\max}$
to
$\RunTime_{\Halt}^{\RunPol, \SkipPol}(\eta)$
is up-bounded by
$O (\RunTime_{\Halt}^{\Pi_{1}}(n) + \RunTime_{\Halt}^{\Pi_{2}}(n))$.
\end{lemma}
\begin{proof}
Follows by the halting runtime bound of $\Pi_{1}$ and $\Pi_{2}$ and the fact
that the species in $\Species'_{1}$ and $\Species'_{2}$ are catalysts for any
reaction in
$\Reactions - \Reactions'_{1}$
and
$\Reactions - \Reactions'_{2}$,
respectively.
\end{proof}

\begin{lemma} \label{lemma:crd:semilinear:closure:bound-contribution-i-leader}
The total contribution of rounds
$i_{\max} \leq i < i_{\mathrm{leader}}$
to
$\RunTime_{\Halt}^{\RunPol, \SkipPol}(\eta)$
is up-bounded by
$O (n)$.
\end{lemma}
\begin{proof}
Given a round
$i_{\max} \leq i < i_{\mathrm{leader}}$,
let
$\ell_{i} = \EffConf^{i}(\{ G_{0, 0}, G_{0, 1}, G_{1, 0}, G_{1, 1} \})$
and recall that
$\RunPol(\EffConf^{i})
=
\left\{
\beta_{u_{1}, u_{2}, w_{1}, w_{2}}
\mid
u_{1}, u_{2}, w_{1}, w_{2} \in \{ 0, 1 \}
\right\}$.
Employing \Lem{}~\ref{lemma:runtime:tools:temporal-cost}, we conclude that the
contribution of round $i$ to
$\RunTime_{\Halt}^{\RunPol, \SkipPol}(\eta)$
is up-bounded as
$\TemporalCost^{\RunPol}(\EffConf^{i}) \leq O (n / \ell_{i}^{2})$.
Notice that
$\ell_{i + 1} \leq \ell_{i}$
and that the inequality is strict if round $i$ is target-accomplished.
This is no longer guaranteed if round $i$ is target-deprived, however we argue
that if
$\ell_{i'} = \ell_{i}$
for some
$i' > i$,
then
$i' - i \leq O (1)$.
Indeed, this is ensured by
\Obs{}~\ref{observation:crd:semilinear:closure:local-halting} as
$i \geq i_{\max}$.
Therefore, the total contribution of rounds
$i_{\max} \leq i < i_{\mathrm{leader}}$
to
$\RunTime_{\Halt}^{\RunPol, \SkipPol}(\eta)$
is up-bounded by
\[\textstyle
\sum_{j = 2}^{n} \frac{n}{j^{2}}
\, \leq \,
O (n) \cdot \sum_{j = 1}^{\infty} \frac{1}{j^{2}}
\, \leq \,
O (n)
\, ,
\]
thus establishing the assertion.
\end{proof}

\begin{lemma} \label{lemma:crd:semilinear:closure:bound-contribution-i-star}
The total contribution of rounds
$i_{\mathrm{leader}} \leq i < i^{*}$
to
$\RunTime_{\Halt}^{\RunPol, \SkipPol}(\eta)$
is up-bounded by
$O (n)$.
\end{lemma}
\begin{proof}
Follows since there can be at most two such rounds.
\end{proof}

Combining \Lem{}~\ref{lemma:runtime:tools:ignition} with
\Lem{}\ \ref{lemma:crd:semilinear:closure:bound-contribution-i-max},
\ref{lemma:crd:semilinear:closure:bound-contribution-i-leader}, and
\ref{lemma:crd:semilinear:closure:bound-contribution-i-star},
we conclude that
$\RunTime_{\Halt}^{\RunPol, \SkipPol}(\eta)
=
O (\RunTime_{\Halt}^{\Pi_{1}}(n) + \RunTime_{\Halt}^{\Pi_{2}}(n) + n)$,
which yields
\Prop{}~\ref{proposition:crd:semilinear:closure}.
\LongVersionEnd 

\subsection{Detection Predicates}
\label{section:crd:detection}
For a vector
$\mathbf{x} \in \Naturals^{\Sigma}$,
let
$\mathbf{x}_{\downarrow} \in \{ 0, 1 \}^{\Sigma} \subset \Naturals^{\Sigma}$
be the vector defined
\LongVersion 
by setting
$\mathbf{x}_{\downarrow}(A) = 1$
if
$\mathbf{x}(A) > 0$;
and
$\mathbf{x}_{\downarrow}(A) = 0$
otherwise%
\LongVersionEnd 
\ShortVersion 
so that
$\mathbf{x}_{\downarrow}(A) = 1
\Longleftrightarrow
\mathbf{x}(A) > 0$%
\ShortVersionEnd 
.
A predicate
$\psi : \Naturals^{\Sigma} \rightarrow \{ 0, 1 \}$
is a \emph{detection} predicate
if
$\psi(\mathbf{x}) = \psi(\mathbf{x}_{\downarrow})$
for
\LongVersion 
every vector
\LongVersionEnd 
\ShortVersion 
all
\ShortVersionEnd 
$\mathbf{x} \in \Naturals^{\Sigma}$
(cf.\
\cite{AlistarhDKSU2017robust,
ChenCDS2017speed-faults,
DudekK2018universal}).
Chen et al.~\cite{ChenCDS2017speed-faults} prove that
\LongVersion 
in the context of the strongly fair adversarial scheduler,
\LongVersionEnd 
a predicate
$\psi : \Naturals^{\Sigma} \rightarrow \{ 0, 1 \}$
can be stably decided
\ShortVersion 
under the strongly fair adversarial scheduler
\ShortVersionEnd 
by a stabilization speed fault free CRD if and only if it is a detection
predicate.
\Cor{}~\ref{corollary:weak-correctness-implies-strong-correctness} ensures
that the only if direction translates to our weakly fair adversarial
scheduler;
employing \Lem{}~\ref{lemma:speed-fault-lower-bound}, we conclude that a
non-detection predicate cannot be decided by a CRD whose stabilization
\LongVersion 
(and hence also halting)
\LongVersionEnd 
runtime is better than
$\Omega (n)$.
For the if direction, the construction in \cite{ChenCDS2017speed-faults}
yields leaderless CRDs that haltingly decide $\psi$ whose expected halting
runtime under the stochastic scheduler is
$O (\log n)$.
The following theorem states that the same (asymptotic) runtime upper bound
can be obtained under the weakly fair adversarial scheduler%
\ShortVersion 
;
the theorem is proved in the attached full version, where we also explain why
the promised upper bound is asymptotically tight%
\ShortVersionEnd 
.

\begin{theorem} \label{theorem:crd:detection}
For every detection predicate
$\psi : \Naturals^{\Sigma} \rightarrow \{ 0, 1 \}$,
there exists a leaderless CRD that haltingly decides $\psi$ whose
halting runtime is
$O (\log n)$.
Moreover, the CRD is designed so that all molecules in the halting
configuration are voters.
\end{theorem}

\LongVersion 
A standard probabilistic argument reveals that in a stochastically scheduled
execution, the expected stochastic runtime until each molecule, present in the
initial configuration, reacts at least once is
$\Omega (\log n)$.
Employing Lemma~\ref{lemma:stochastic-benchmark}, we deduce the same
asymptotic lower bound for the stabilization (and thus also halting) runtime
of any protocol whose outcome may be altered by a single input molecule.
Since CRD protocols that decide detection predicates satisfy this property, it
follows that the runtime upper bound promised in
\Thm{}~\ref{theorem:crd:detection} is (asymptotically) tight.

We establish \Thm{}~\ref{theorem:crd:detection} by presenting a
(leaderless) CRD
$\Pi
=
(\Species, \Reactions, \Sigma, \Upsilon_{0}, \Upsilon_{1}, F, \mathbf{0})$
that haltingly decides a given detection predicate
$\psi : \Naturals^{\Sigma} \rightarrow \{ 0, 1 \}$.
As promised in the theorem, the halting runtime of $\Pi$ is
$\RunTime_{\Halt}^{\Pi}(n) = O (\log n)$.
We note that both the construction and the runtime analysis of $\Pi$ are
heavily inspired by the construction of \cite{ChenCDS2017speed-faults} for a
CRD that decides $\psi$ (although the runtime in
\cite{ChenCDS2017speed-faults} is analyzed assuming a stochastic scheduler).

The species set of protocol $\Pi$ is defined to be
\[\textstyle
\Species
\, = \,
\Sigma
\cup
\{ F \}
\cup
\left\{ D_{\mathbf{u}} \mid \mathbf{u} \in \{ 0, 1 \}^{\Sigma} \right\}
\, .
\]
The species in
$\Sigma \cup \{ F \}$
are regarded as the ignition species of the ignition gadget presented in
\Sect{}~\ref{section:runtime:tools:ignition}, taking the ignition reaction
associated with species
$A \in \Sigma$
to be
\\
$\iota_{A}$:
$A \rightarrow D_{\mathbf{z}^{A}}$,
\\
where
$\mathbf{z}^{A} \in \{ 0, 1 \}^{\Sigma}$
denotes the (unit) vector that corresponds to the multiset
$1 A$;
and the ignition reaction associated with species $F$ to be
\\
$\iota_{F}$:
$F \rightarrow D_{\mathbf{0}}$.

Semantically, the presence of species $D_{\mathbf{u}}$ in the configuration
indicates that it has already been detected that input species $A$ was present
in the initial configuration for each
$A \in \Sigma$
such that
$\mathbf{u}(A) = 1$.
Following this semantics, the voter species are defined as
\[\textstyle
\Upsilon_{0}
\, = \,
\left\{ D_{\mathbf{u}} \mid \psi(\mathbf{u}) = 0 \right\}
\qquad \text{and} \qquad
\Upsilon_{1}
\, = \,
\left\{ D_{\mathbf{u}} \mid \psi(\mathbf{u}) = 1 \right\}
\, .
\]

Concretely, given two vectors
$\mathbf{u}, \mathbf{u}' \in \{ 0, 1 \}^{\Sigma}$,
let
$\mathbf{u} \lor \mathbf{u}' \in \{ 0, 1 \}^{\Sigma}$
denote the bitwise logical or of $\mathbf{u}$ and $\mathbf{u}'$.
The non-void reaction set $\NonVoid(\Reactions)$ of protocol $\Pi$ includes
the following reactions on top of the aforementioned ignition reactions:
\\
$\beta_{\mathbf{u}, \mathbf{u}'}$:
$D_{\mathbf{u}} + D_{\mathbf{u'}}
\rightarrow
2 D_{\mathbf{u} \lor \mathbf{u}'}$
for every distinct
$\mathbf{u}, \mathbf{u}' \in \{ 0, 1 \}^{\Sigma}$.
\\
In other words, the $\beta$ reactions ``spread'' the detection of the various
input species among all (working) molecules.

\subparagraph{Analysis.}
For the analysis of protocol $\Pi$, fix some input vector
$\mathbf{x} \in \Naturals^{\Sigma}$
and let
$\c^{0} \in \Naturals^{\Species}$
be a valid initial configuration with
$\c^{0}|_{\Sigma} = \mathbf{x}$.
Consider a weakly fair execution
$\eta = \langle \c^{t}, \zeta^{t} \rangle_{t \geq 0}$
of $\Pi$ emerging from $\c^{0}$.

Let
$t_{\Ignt} \geq 0$
be the earliest step in which the ignition gadget matures in $\eta$ (as
promised in \Lem{}~\ref{lemma:runtime:tools:ignition}).
For a configuration
$\c \in \Naturals^{\Species}$,
let
$\OR(\c)$ denote the result of the bitwise logical or of all vectors
$\mathbf{u} \in \{ 0, 1 \}^{\Sigma}$
such that
$\c(D_{\mathbf{u}}) > 0$.

\begin{observation} \label{observation:crd:detection:or-invariant}
For every step
$t \geq t_{\Ignt}$,
we have
$\OR(\c^{t})
=
\mathbf{x}_{\downarrow}$.
\end{observation}
\begin{proof}
Follows by the choice of the ignition reactions as
$\OR(\c^{t + 1})
=
\OR(\c^{t})$
for every
$t \geq t_{\Ignt}$.
\end{proof}

Let
$n = \| \c^{0} \|$
be the molecular count of the initial configuration and observe that
$\| \c^{t} \| = n$
for all
$t \geq 0$
as $\Pi$ is density preserving.
Given a configuration
$\c \in \Naturals^{\Species}$
and an input species
$A \in \Sigma$,
let
\[\textstyle
w_{A}(\c)
\, = \,
\sum_{\mathbf{u} \in \{0, 1 \}^{\Sigma} \, : \, \mathbf{u}(A) = 1}
\,
\c(D_{\mathbf{u}})
\qquad \text{and} \qquad
w(\c)
\, = \,
\sum_{A \in \Sigma} w_{A}(\c)
\, .
\]

\begin{observation} \label{observation:crd:detection:weight-progress}
The following three properties hold for every step
$t \geq t_{\Ignt}$
and input species
$A \in \Sigma$:
\\
(1)
$w_{A}(\c^{t + 1}) \geq w_{A}(\c^{t})$;
\\
(2)
if
$\zeta^{t} = \beta_{\mathbf{u}, \mathbf{u}'}$
and
$\mathbf{u}(A) \neq \mathbf{u}'(A)$,
then
$w_{A}(\c^{t + 1}) = w_{A}(\c^{t}) + 1$;
and
\\
(3)
$0 \leq w(\c^{t}) \leq n \cdot |\Sigma|$.
\end{observation}

As
$\Applicable(\c^{t})$
includes a $\beta$ reaction if and only if
$\c^{t}(D_{\mathbf{u}}) > 0$
for some
$\mathbf{u} \neq \OR(\c^{t})$,
we obtain the following observation.

\begin{observation} \label{observation:crd:detection:halting-configuration}
There exists a step
$t^{*} \geq t_{\Ignt}$
such that
$\c^{t^{*}}$
is halting and
$\c^{t^{*}}(D_{\mathbf{u}}) > 0$
implies that
$\mathbf{u} = \OR(\c^{t^{*}})$.
\end{observation}

\Cor{}~\ref{corollary:crd:detection:correctness} now follows by the definition
of
$\Upsilon_{0}$
and
$\Upsilon_{1}$
due to \Obs{}\ \ref{observation:crd:detection:or-invariant} and
\ref{observation:crd:detection:halting-configuration}.

\begin{corollary} \label{corollary:crd:detection:correctness}
Protocol $\Pi$ is haltingly correct.
\end{corollary}

For the halting runtime analysis, we fix some skipping policy $\SkipPol$ and
prove that
$\RunTime_{\Halt}^{\Pi}(n) \leq O (\log n)$
by presenting a runtime policy $\RunPol$ for $\Pi$ (defined independently of
$\eta$ and $\SkipPol$) and showing that
$\RunTime_{\Halt}^{\RunPol, \SkipPol}(\eta) \leq O (\log n)$.
Given a configuration
$\c \in \Naturals^{\Species}$,
the runtime policy $\RunPol$ is defined as follows:
\begin{itemize}

\item
if
$\c(\Sigma \cup \{ F \}) > 0$,
then $\RunPol(\c)$ consists of the ignition reactions;

\item
else
$\RunPol(\c) = \NonVoid(\Reactions)$.

\end{itemize}

\begin{lemma} \label{lemma:crd:detection:runtime}
$\RunTime_{\Halt}^{\RunPol, \SkipPol}(\eta) \leq O (\log n)$.
\end{lemma}
\begin{proof}
Let $\EffStep(i)$ and
$\EffConf^{i} = \c^{\EffStep(i)}$
be the effective step and effective configuration, respectively, of round
$i \geq 0$
under $\RunPol$ and $\SkipPol$.
Let
$i_{\Ignt} = \min \{ i \geq 0 \mid \EffStep(i) \geq t_{\Ignt} \}$
and let
$i^{*} = \min \{ i \geq i_{\Ignt} \mid \EffStep(i) \geq t^{*} \}$.
\Lem{}~\ref{lemma:runtime:tools:ignition} ensures that
$\sum_{i = 0}^{i_{\Ignt} - 1} \TemporalCost^{\RunPol}(\EffConf^{i})
\leq
O (\log n)$,
so it remains to prove that
$\sum_{i = i_{\Ignt}}^{i^{*} - 1} \TemporalCost^{\RunPol}(\EffConf^{i})
\leq
O (\log n)$.

Fix a round
$i_{\Ignt} \leq i < i^{*}$
and notice that the definition of $\RunPol$ guarantees that round $i$ is
target-accomplished.
For an input species
$A \in \Sigma$,
denote
$w_{A}(i) = w_{A}(\EffConf^{i})$.
Let $A(i)$ be the first (according to an arbitrary total order on $\Sigma$)
input species
$A  \in \Sigma$
that satisfies
$w_{A}(i + 1) > w_{A}(i)$.

The key observation now is that the temporal cost associated with round $i$ is
up-bounded as
\[\textstyle
\TemporalCost^{\RunPol}(\EffConf^{i})
\, \leq \,
O \left(
\frac{n}{w_{A(i)}(i) \cdot (n - w_{A(i)}(i))}
\right)
\, .
\]
Therefore, we can establish the assertion by developing
\begin{align*}\textstyle
\sum_{i = i_{\Ignt}}^{i^{*} - 1} \TemporalCost^{\RunPol}(\EffConf^{i})
\, \leq \, &
\sum_{A \in \Sigma} \,
\sum_{i_{\Ignt} \leq i < i^{*} \, : \, A_{i} = A} \,
O \left( \frac{n}{w_{A}(i) \cdot (n - w_{A}(i))} \right)
\\
\leq \, &
\sum_{A \in \Sigma} \,
\sum_{i_{\Ignt} \leq i < i^{*} \, : \, w_{A}(i + 1) = w_{A}(i) + 1} \,
O \left( \frac{n}{w_{A}(i) \cdot (n - w_{A}(i))} \right)
\\
\leq \, &
O (n) \cdot |\Sigma| \cdot \sum_{i = 1}^{n - 1} \frac{1}{i \cdot (n - i)}
\\
= \, &
O (n) \cdot |\Sigma| \cdot \left(
\sum_{i = 1}^{\lfloor n / 2 \rfloor} \frac{1}{i \cdot (n - i)}
\, + \,
\sum_{i = \lfloor n / 2 \rfloor + 1}^{n - 1} \frac{1}{i \cdot (n - i)}
\right)
\\
\leq \, &
O (1) \cdot |\Sigma| \cdot \left(
\sum_{i = 1}^{\lfloor n / 2 \rfloor} \frac{1}{i}
\right)
\, \leq \,
O (\log n)
\, ,
\end{align*}
where
the third transition follows from
\Obs{}~\ref{observation:crd:detection:weight-progress}
and
the last transition holds as
$|\Sigma| = O (1)$.
\end{proof}

\Thm{}~\ref{theorem:crd:detection} follows from
\Cor{}~\ref{corollary:crd:detection:correctness} and
\Lem{}~\ref{lemma:crd:detection:runtime}.
\LongVersionEnd 

\section{Vote Amplification}
\label{section:amplification}
Recall that CRDs are required to stabilize/halt into configurations $\c$
that include a positive number of $v$-voter molecules and zero
$(1 - v)$-voter
molecules, where
$v \in \{ 0, 1 \}$
is determined by the decided predicate according to the input vector.
This requirement alone does not rule out the possibility of having a small
(yet positive) voter molecular count in $\c$.
Indeed, the semilinear predicate CRDs promised in
\Thm{}~\ref{theorem:crd:semilinear} are designed so that
the configuration $\c$ includes a single voter molecule (this is in contrast
to the detection predicate CRDs promised in
\Thm{}~\ref{theorem:crd:detection}, where all molecules in $\c$ are voters).

In practice though, it may be difficult to obtain a meaningful signal from
small molecular counts.
Consequently, we aim for \emph{vote amplified} CRDs, namely, CRDs
that guarantee to stabilize/halt into configurations in which the voter
molecules take all but an $\epsilon$-fraction of the total molecular count for
an arbitrarily small constant
$\epsilon > 0$.
These are obtained by means of a ``generic compiler'' that can be applied, in
a black-box manner, to any existing CRD, turning it into a vote amplified CRD
while preserving the original stabilization/halting correctness.
At the heart of this compiler lies a CRN protocol for a standalone
computational task, referred to as \emph{vote amplification (VA)}, whose
runtime dominates the runtime overhead of the compiler, as stated in the
following theorem
\LongVersion 
(proved in \Sect{}~\ref{section:amplification:application})%
\LongVersionEnd 
\ShortVersion 
(proved in the attached full version)%
\ShortVersionEnd 
.

\begin{theorem} \label{theorem:amplification:application}
Consider a predicate
$\psi : \Naturals^{\Sigma} \rightarrow \{ 0, 1 \}$
that can be haltingly decided by a (leaderless) CRD in $T_{\psi}(n)$ time.
The existence of a VA protocol that stabilizes (resp., halts) in $T_{\Amp}(n)$
time implies
\LongVersion 
that for any constant
$\epsilon > 0$,
there exists
\LongVersionEnd 
\ShortVersion 
the existence of
\ShortVersionEnd 
a (leaderless)
\ShortVersion 
vote amplified
\ShortVersionEnd 
CRD that stably (resp., haltingly) decides $\psi$ in
$T_{\psi}(O (n)) + T_{\Amp}(O (n)) + O (\log n)$
time%
\LongVersion 
so that the non-voter molecules take at most an $\epsilon$-fraction
of the molecular count of the configuration(s) into which the CRD stabilizes
(resp., halts)%
\LongVersionEnd 
.
\end{theorem}

Assuming a stochastic scheduler, Angluin et al.~\cite{AngluinAE2008fast}
develop a VA protocol that halts in
$O (n)$
time.
Unfortunately, the protocol of \cite{AngluinAE2008fast} does not meet the
topological conditions of
\Lem{}~\ref{lemma:correctness-via-configuration-digraph}, hence the (weakly
fair) adversarial scheduler can prevent this protocol from stabilizing (see
\LongVersion 
Appendix~\ref{appendix:random-walk}
\LongVersionEnd 
\ShortVersion 
the attached full version
\ShortVersionEnd 
for more details).
Using a completely different technique, we develop a VA protocol whose
guarantees are cast in the following theorem.

\begin{theorem} \label{theorem:amplification:protocol}
There exists a VA protocol (operating under the weakly fair scheduler) that
stabilizes in
$O (n)$
time and halts in
$O (n \log n)$
time.
\end{theorem}

Combined with \Thm{}~\ref{theorem:amplification:application}, we obtain a
compiler whose stabilization and halting runtime overheads are
$O (n)$
and
$O (n \log n)$,
respectively.
Applying this compiler to the CRDs promised in
\Thm{}~\ref{theorem:crd:semilinear} results in vote amplified CRDs whose
stabilization runtime remains
$O (n)$,
however their halting runtime increases to
$O (n \log n)$.
The excessive
$\log n$
factor would be shaved by a VA protocol that halts in
$O (n)$
time whose existence remains an open question.

\subparagraph{Task Formalization.}
%
\LongVersion 
As stated in \Thm{}~\ref{theorem:amplification:application}, our
compiler is formalized by means of the VA task.
\LongVersionEnd 
A VA protocol is a CRN protocol
$\Pi = (\Species, \Reactions)$
whose species set $\Species$ is partitioned into the pairwise disjoint sets
$
\PermVoters_{0} \cup \PermVoters_{1}
\cup
\FluidVoters_{0} \cup \FluidVoters_{1}
=
\Species$,
where for
$v \in \{ 0, 1 \}$,
the species in $\PermVoters_{v}$ are referred to as \emph{permanent $v$-voters}
and the species in $\FluidVoters_{v}$ are referred to as \emph{fluid
$v$-voters}.
The permanent voters are regarded as part of the task specification and can
participate in the reactions of $\Pi$ only as catalysts (which means that the
molecular count of each permanent voter remains invariant throughout the
execution).

A configuration
$\c^{0} \in \Naturals^{\Species}$
is
\LongVersion 
valid as an 
\LongVersionEnd 
\ShortVersion 
a valid
\ShortVersionEnd 
initial configuration for the VA task if there exists a vote
$v \in \{ 0, 1 \}$
such that
$\c^{0}(\PermVoters_{v}) > 0$
and
$\c^{0}(\PermVoters_{1 - v}) = 0$,
in which case we refer to $\c^{0}$ as a \emph{$v$-voting} initial
configuration.
\LongVersion 
For convenience, we further require that
$\c^{0}(\{ \PermVoters_{0}, \PermVoters_{1} \})
\leq
\c^{0}(\{ \FluidVoters_{0}, \FluidVoters_{1} \})$,
which means that the permanent voters (in fact, the voters in
$\PermVoters_{v}$) do not dominate the initial molecular count.\footnote{%
This requirement is not inherent to the task formulation and we present it
solely for the sake of simplifying the analysis.}
\par
\LongVersionEnd 
A configuration
$\c \in \Naturals^{\Species}$
is
\LongVersion 
said to be
\LongVersionEnd 
an \emph{amplification} of a $v$-voting initial configuration $\c^{0}$
if
(1)
$\c(A) = \c^{0}(A)$
for every
$A \in \PermVoters_{0} \cup \PermVoters_{1}$;
(2)
$\c(\FluidVoters_{v}) = \c^{0}(\{ \FluidVoters_{0} \cup \FluidVoters_{1} \})$;
and
(3)
$\c(\FluidVoters_{1 - v}) = 0$.
In other words, an amplification of a $v$-voting initial configuration keeps
the original permanent voter molecules and shifts all fluid voter molecules to
the $v$-voting side.

The VA protocol $\Pi$ is stably (resp., haltingly) correct if every weakly
fair valid execution
$\eta = \langle \c^{t}, \alpha^{t} \rangle_{t \geq 0}$
stabilizes (resp., halts) into the (set of)
amplifications of
\LongVersion 
its initial configuration
\LongVersionEnd 
$\c^{0}$.
The typical scenario involves a small number of permanent $v$-voter molecules
and the challenge is to ensure that all
\LongVersion 
(asymptotically many)
\LongVersionEnd 
fluid voter molecules ``end up'' in $\FluidVoters_{v}$.
We emphasize that for $\Pi$ to be correct, the protocol should handle any
initial configuration
$\c^{0}|_{\FluidVoters_{0} \cup \FluidVoters_{1}}$
of the fluid voters.

\subparagraph{The VA Protocol.}
We now turn to develop the VA protocol
$\Pi = (\Species, \Reactions)$
promised in
\Thm{}~\ref{theorem:amplification:protocol}.
\LongVersion 
Fix some sets $\PermVoters_{0}$ and $\PermVoters_{1}$ of permanent $0$- and
$1$-voters, respectively.
\LongVersionEnd 
\ShortVersion 
For simplicity, assume in this extended abstract that $\PermVoters_{0}$ and
$\PermVoters_{1}$ are singleton sets with
$\PermVoters_{0} = \{ P_{0} \}$
and
$\PermVoters_{1} = \{ P_{1} \}$;
the general case is handled in the attached full version.
\ShortVersionEnd 
Protocol $\Pi$ is defined over the fluid voter sets
$\FluidVoters_{0} = \{ H_{0}, L_{0} \}$
and
$\FluidVoters_{1} = \{ H_{1}, L_{1} \}$.
Semantically, we think of the $H$ (resp., $L$) fluid voters as having a high
(resp., low) confidence level in their vote.
The reaction set $\Reactions$ of $\Pi$ includes the following non-void
reactions:
\\
\LongVersion 
$\beta_{v, P_{v}}^{A}$%
\LongVersionEnd 
\ShortVersion 
$\beta_{v}^{A}$%
\ShortVersionEnd 
:
$P_{v} + A \rightarrow P_{v} + H_{v}$
for every
$v \in \{ 0, 1 \}$%
\LongVersion 
,
$P_{v} \in \PermVoters_{v}$,
\LongVersionEnd 
\ and
$A \in \{ H_{1 - v}, L_{0}, L_{1} \}$;
\\
$\gamma$:
$H_{0} + H_{1} \rightarrow L_{0} + L_{1}$;
and
\\
$\delta_{v}$:
$H_{v} + L_{1 - v} \rightarrow 2 L_{v}$
for every
$v \in \{ 0, 1 \}$.
\LongVersion 
\par
In other words, the $\beta_{v}$ reactions turn any fluid voter into a high
confidence fluid $v$-voter;
reaction $\gamma$ turns high confidence fluid voters with opposite votes into
low confidence fluid voters with opposite votes;
and
reaction $\delta_{v}$ turns a high confidence fluid $v$-voter and a low
confidence fluid
$(1 - v)$-voter
into two low confidence fluid $v$-voters.
\LongVersionEnd 
\ShortVersion 
\\
\ShortVersionEnd 
Informally, these reactions guarantee that the adversary has little leverage
because, as we show soon, \emph{all} of the non-void reactions make nontrivial
progress in their own different ways.

For the runtime analysis of protocol $\Pi$, consider a weakly fair valid
execution
$\eta = \langle \c^{t}, \zeta^{t} \rangle_{t \geq 0}$
of initial molecular count
$\| \c^{0} \| = n$.
Assume for simplicity that the initial configuration $\c^{0}$ is $1$-voting
which means that
\LongVersion 
all permanent voters present in $\c^{0}$ (and in $\c^{t}$ for any
$t \geq 0$)
belong to $\PermVoters_{1}$%
\LongVersionEnd 
\ShortVersion 
$\c^{t}(P_{1}) > 0$
and
$\c^{t}(P_{0}) = 0$
for all
$t \geq 0$%
\ShortVersionEnd 
;
the case where $\c^{0}$ is $0$-voting is analyzed symmetrically.
Let
$m = \c^{0}(\{ H_{0}, L_{0}, L_{1}, H_{1} \})$
be the initial molecular count of the fluid voters and observe that
$\c^{t}(\{ H_{0}, L_{0}, L_{1}, H_{1} \}) = m$
for every
$t \geq 0$.

To capture progress
\LongVersion 
made as execution $\eta$ advances%
\LongVersionEnd 
, we assign an integral score $s(\cdot)$ to each fluid voter by setting
\LongVersion 
\[\textstyle
s(H_0) = -4,
\quad
s(L_0) = -1,
\quad
s(L_1) = 1,
\quad\text{and}\quad
s(H_1) =  2
\, .
\]
\LongVersionEnd 
\ShortVersion 
$s(H_0) = -4$,
$s(L_0) = -1$,
$s(L_1) = 1$,
and
$s(H_1) =  2$.
\ShortVersionEnd 
Substituting the $s(\cdot)$ scores into each reaction
$\alpha \in \NonVoid(\Reactions)$
reveals that the sum of scores of $\alpha$'s fluid reactants is strictly
smaller than the sum of scores of $\alpha$'s fluid products.
Denoting the total score in a configuration
$\c \in \Naturals^{\Species}$
by
$s(\c) = \sum_{A \in \{ H_{0}, L_{0}, L_{1}, H_{1} \}} \c(A) \cdot s(A)$,
we deduce that
$s(\c^{t + 1}) \geq s(\c^{t})$
and that
$\zeta^{t} \in \NonVoid(\Reactions)
\Longrightarrow
s(\c^{t + 1}) > s(\c^{t})$
for every
$t \geq 0$.
Since 
$-4 m
\leq
s(\c^{t})
\leq
2 m$
for every
$t \geq 0$,
it follows that $\eta$ includes, in total, at most
$O(m) \leq O (n)$
non-void reactions until it
\LongVersion 
halts%
\LongVersionEnd 
\ShortVersion 
stabilizes%
\ShortVersionEnd 
.

The last bound ensures that progress is made
\LongVersion 
whenever a non-void reaction is scheduled%
\LongVersionEnd 
\ShortVersion 
on each non-void reaction%
\ShortVersionEnd 
.
Accordingly, we choose the runtime policy $\RunPol$ so that
$\RunPol(\c) = \NonVoid(\Reactions)$
for all configurations
$\c \in \Naturals^{\Species}$.\footnote{%
Although it serves its purpose in the current analysis, for many CRN protocols,
a runtime policy whose targets cover all non-void reactions is suboptimal;
this is elaborated in
\LongVersion 
\Sect{}~\ref{section:justify-runtime-policy}%
\LongVersionEnd 
\ShortVersion 
the attached full version%
\ShortVersionEnd 
.
}
\LongVersion 
This means in particular that for every configuration
$\c \in \Naturals^{\Species}$,
the only configuration reachable from $\c$ via a $\RunPol$-avoiding path is
$\c$ itself.
\LongVersionEnd 

Fix some skipping policy $\SkipPol$ and let $\EffConf^{i}$ be the effective
configuration of round
$i \geq 0$
under $\RunPol$ and $\SkipPol$.
Let
$i^{*} = \min \{
i \geq 0 \mid \EffConf^{i}(\{ H_{0}, L_{0}\}) = 0
\}$
be the first round whose effective step appears after $\eta$ stabilizes%
\LongVersion 
\ and
let
$i^{**} = \min \{
i \geq 0 \mid \EffConf^{i}(\{ H_{0}, L_{0}, L_{1} \}) = 0
\}$
be the first round whose effective step appears after $\eta$ halts%
\LongVersionEnd 
.
Since the choice of $\RunPol$ ensures that each round
\LongVersion 
$0 \leq i < i^{**}$
\LongVersionEnd 
\ShortVersion 
$0 \leq i < i^{*}$
\ShortVersionEnd 
is target-accomplished, ending with a non-void reaction, it follows that
\LongVersion 
$i^{*} \leq i^{**} \leq O (n)$%
\LongVersionEnd 
\ShortVersion 
$i^{*} \leq O (n)$%
\ShortVersionEnd 
.

To bound the stabilization runtime of execution $\eta$ under $\RunPol$ and
$\SkipPol$, we argue that
$\Prp_{\EffConf^{i}}(\NonVoid(\Reactions))
\geq
\Omega (1)$
for every
$0 \leq i < i^{*}$;
\LongVersion 
this allows us to employ \Lem{}~\ref{lemma:runtime:tools:temporal-cost} and
\LongVersionEnd 
\ShortVersion 
by a simple probabilistic argument (elaborated in the attached full version),
this allows us to
\ShortVersionEnd 
conclude that
$\TemporalCost^{\RunPol}(\EffConf^{i}) \leq O (1)$
for every
$0 \leq i < i^{*}$.
To this end, notice that if
$\EffConf^{i}(H_{1}) \geq m / 2$,
then
\[\textstyle
\Prp_{\EffConf^{i}} \left( \{ \gamma, \delta_{1} \} \right)
\, = \,
\frac{1}{\Vol}
\cdot
\EffConf^{i}(H_{1}) \cdot \EffConf^{i}(\{ H_{0}, L_{0} \})
\, \geq \,
\Omega (m / n)
\, = \,
\Omega (1)
\, .
\]
Otherwise
($\EffConf^{i}(H_{1}) < m / 2$),
we know that
$\EffConf^{i}(\{ H_{0}, L_{0}, L_{1} \}) > m / 2$,
hence
\LongVersion 
\begin{align*}
\Prp_{\EffConf^{i}} \left(
\left\{ \beta_{1, P_{1}}^{A}
\mid
P_{1} \in \PermVoters_{1}, A \in \{ H_{0}, L_{0}, L_{1} \} \right\}
\right)
\, = \, &
\tfrac{1}{\Vol}
\cdot
\EffConf^{i}(\{ H_{0}, L_{0}, L_{1} \}) \cdot \EffConf^{i}(\PermVoters_{1})
\\
\geq \, &
\Omega (m / n)
\, = \,
\Omega (1)
\, ,
\end{align*}
\LongVersionEnd 
\ShortVersion 
\[\textstyle
\Prp_{\EffConf^{i}} \left(
\left\{ \beta_{1}^{A}
\mid
A \in \{ H_{0}, L_{0}, L_{1} \} \right\}
\right)
\, = \,
\tfrac{1}{\Vol}
\cdot
\EffConf^{i}(\{ H_{0}, L_{0}, L_{1} \}) \cdot \EffConf^{i}(\PermVoters_{1})
\, \geq \,
\Omega (m / n)
\, = \,
\Omega (1)
\, ,
\]
\ShortVersionEnd 
thus establishing the argument.
Therefore, the stabilization runtime of $\eta$ satisfies
\[\textstyle
\RunTime_{\Stab}^{\RunPol, \SkipPol}(\eta)
\, = \,
\sum_{i = 0}^{i^{*} - 1} \TemporalCost^{\RunPol}(\EffConf^{i})
\, \leq \,
\sum_{i = 0}^{O (n)} O (1)
\, = \,
O (n)
\, .
\]

\ShortVersion 
The proof of \Thm{}~\ref{theorem:amplification:protocol} is completed by
showing that protocol $\Pi$ halts in
$O (n \log n)$
time.
This part of the proof is deferred to the attached full version.
\ShortVersionEnd 

\LongVersion 
To bound the halting runtime of $\eta$ under $\RunPol$ and $\SkipPol$, fix
some round
$i^{*} \leq i < i^{**}$
and observe that
$\EffConf^{i}(\{ H_{0}, L_{0}\}) = 0$
and that
$\Applicable(\EffConf^{i}) \cap \NonVoid(\Reactions)
=
\{ \beta_{1, P_{1}}^{L_{1}} \mid P_{1} \in \PermVoters_{1} \}$.
Let
$\ell_{i} = \EffConf^{i}(L_{1})$
and notice that
$\Prp_{\EffConf^{i}}(\{ \beta_{1, P_{1}}^{L_{1}} \mid P_{1} \in \PermVoters_{1} \})
\geq
\ell_{i} / \Vol
\geq
\Omega (\ell_{i} / n)$.
Therefore, we can employ \Lem{}~\ref{lemma:runtime:tools:temporal-cost} to
conclude that
$\TemporalCost^{\RunPol}(\EffConf^{i}) \leq O (n / \ell_{i})$.
Since
$\ell_{i + 1} < \ell_{i}$
for every
$i^{*} \leq i < i^{**}$
and since
$\ell_{i^{*}} \leq m \leq n$
and
$\ell_{i^{**}} = 0$,
it follows that the halting runtime of $\eta$ satisfies
\begin{align*}
\RunTime_{\Halt}^{\RunPol, \SkipPol}(\eta)
\, = \, &
\RunTime_{\Stab}^{\RunPol, \SkipPol}(\eta)
+
\sum_{i = i^{*}}^{i^{**} - 1} \TemporalCost^{\RunPol}(\EffConf^{i})
\\
\leq \, &
O (n)
+
\sum_{\ell = 1}^{n} O (n / \ell)
\, = \,
O (n) \cdot \sum_{\ell = 1}^{n} 1 / \ell
\, = \,
O (n \log n)
\, ,
\end{align*}
thus establishing \Thm{}~\ref{theorem:amplification:protocol}.
\LongVersionEnd 

\LongVersion 
\subsection{Obtaining Vote Amplified CRDs}
\label{section:amplification:application}
Let
$\Pi_{\psi}
=
(\Species_{\psi}, \Reactions_{\psi}, \Sigma, \Upsilon_{\psi, 0},
\Upsilon_{\psi, 1}, F_{\psi}, \mathbf{k}_{\psi})$
be a CRD protocol that haltingly decides the predicate
$\psi : \Naturals^{\Sigma} \rightarrow \{ 0, 1 \}$
in $T_{\psi}(n)$ time using the runtime policy $\RunPol_{\psi}$.
Let
$\Pi'_{\psi}
=
(\Species'_{\psi}, \Reactions'_{\psi}, \Sigma', \Upsilon'_{\psi, 0},
\Upsilon'_{\psi, 1}, F'_{\psi}, \mathbf{k}'_{\psi})$
be the CRD derived from $\Pi_{\psi}$ by replacing each species
$A \in \Species_{\psi}$
with a $\Pi'_{\psi}$ designated species
$A' \in \Species'_{\psi}$;
in particular, the CRD $\Pi'_{\psi}$ is defined over the input species
$\Sigma' = \{ A' \mid A \in \Sigma \}$.
Let $\RunPol'_{\psi}$ be the runtime policy for $\Pi'_{\psi}$ derived from
$\RunPol_{\psi}$ by replacing each species
$A \in \Species_{\psi}$
with the corresponding species
$A' \in \Species'_{\psi}$.
Consider a VA protocol
$\Pi_{\Amp} = (\Species_{\Amp}, \Reactions_{\Amp})$
that stabilizes (resp., halts)
in $T_{\Amp}(n)$ time using the runtime policy $\RunPol_{\Amp}$ and assume
that the permanent $v$-voters of $\Pi_{\Amp}$ are identified with the species
in
$\Upsilon'_{\psi, v}$
for
$v \in \{ 0, 1 \}$.

We construct the CRD
$\Pi
=
(\Species, \Reactions, \Sigma, \Upsilon_{0}, \Upsilon_{1}, F, \mathbf{k})$
promised in \Thm{}~\ref{theorem:amplification:application} as
follows:
The species set of $\Pi$ is taken to be
$\Species = \Species_{\psi} \cup \Species'_{\psi} \cup \Species_{\Amp}$.
The species in $\Species_{\psi}$ are regarded as the ignition species of the
ignition gadget presented in \Sect{}~\ref{section:runtime:tools:ignition},
taking the ignition reaction associated with species
$A \in \Species_{\psi}$
to be
\\
$\iota_{A}$:
$A \rightarrow A' + \lceil 1 / \epsilon \rceil \cdot B$,
\\
where
$B \in \Species_{\Amp}$
is an arbitrary fluid voter of $\Pi_{\Amp}$.
The remaining non-void reactions of $\Pi$ are the non-void reactions of
$\Pi'_{\psi}$ and $\Pi_{\Amp}$ so that
$\NonVoid(\Reactions)
=
\NonVoid(\Reactions'_{\psi})
\cup
\NonVoid(\Reactions_{\Amp})
\cup
\{ \iota_{A} \mid A \in \Species_{\psi} \}$.
For
$v \in \{ 0, 1 \}$,
the $v$-voters of $\Pi$ are taken to be the (permanent and fluid)
$v$-voters of $\Pi_{\Amp}$.
Finally, the context $\mathbf{k}$ of $\Pi$ is taken to be
$\mathbf{k} = \mathbf{k}_{\psi}$.

Intuitively, the construction of $\Pi$ ensures that once the ignition gadget
has matured, all but an $\epsilon$-fraction of the molecules in the test tube
are fluid voters (of $\Pi_{\Amp}$) and that this remains the case
subsequently.
The fluid voters may interact with the voters of $\Pi'_{\psi}$ --- playing the
role of the permanent voters of $\Pi_{\Amp}$ --- and consequently ``switch
side'' back and forth.
However, once the (projected) execution of $\Pi'_{\psi}$ halts, all permanent
voters present in the test tube have the same vote, so $\Pi_{\Amp}$ can now
run in accordance with the definition of the VA task.

Formally, to establish
\Thm{}~\ref{theorem:amplification:application}, we construct the
runtime policy $\RunPol$ for $\Pi$ by setting $\RunPol(\c)$ as follows for
each configuration
$\c \in \Naturals^{\Species}$:
\begin{itemize}

\item
if
$\c(\Species_{\psi}) > 0$,
then
$\RunPol(\c) = \{ \iota_{A} \mid A \in \Species_{\psi} \}$;

\item
else if
$\c|_{\Species'_{\psi}}$
is not a halting configuration of $\Pi'_{\psi}$, then
$\RunPol(\c) = \RunPol'_{\psi}(\c|_{\Species'_{\psi}})$;

\item
else
$\RunPol(\c) = \RunPol_{\Amp}(\c|_{\Species_{\Amp}})$.

\end{itemize}
Consider a weakly fair valid execution
$\eta = \langle \c^{t}, \zeta^{t} \rangle_{t \geq 0}$
of $\Pi$ of initial molecular
count
$\| \c^{0} \| = n$
and fix a skipping policy $\SkipPol$.

By Lemma~\ref{lemma:runtime:tools:ignition}, the construction of the runtime
policy $\RunPol$ guarantees that the ignition gadget matures in $\eta$ within
$O (\log n)$
time;
following that, the configurations of $\eta$ consist only of $\Pi'_{\psi}$
and $\Pi_{\Amp}$ molecules.
Recall that if a $\Pi'_{\psi}$ species
$A' \in \Species'_{\psi}$
participates in a $\Pi_{\Amp}$ reaction
$\alpha = (\mathbf{r}, \mathbf{p}) \in \Reactions_{\Amp}$,
then $A'$ is a catalyst for $\alpha$, i.e.,
$\mathbf{r}(A') = \mathbf{p}(A')$,
hence the execution of $\Pi'_{\psi}$ is not affected by that of $\Pi_{\Amp}$.
In particular, the construction of $\RunPol$ guarantees that $\Pi'_{\psi}$
halts within
$T_{\psi}(O (n))$
time.
Once the execution of $\Pi'_{\psi}$ halts, that of $\Pi_{\Amp}$ sets into
motion, exploiting the fact that the molecular counts of the voters of
$\Pi'_{\psi}$ (that play the role of the permanent voters in $\Pi_{\Amp}$)
remain fixed.
the construction of $\RunPol$ guarantees that $\Pi_{\Amp}$ stabilizes (resp.,
halts) within
$T_{\Amp}(O (n))$
time.
\LongVersionEnd 

\LongVersion 
\section{Justifying the Runtime Policy Definition}
\label{section:justify-runtime-policy}
Consider a CRN protocol
$\Pi = (\Species, \Reactions)$.
Recall that as defined in \Sect{}~\ref{section:runtime}, a runtime policy
$\RunPol$ for $\Pi$ can map a given configuration
$\c \in \Naturals^{\Species}$
to any subset
$\RunPol(\c) \subseteq \NonVoid(\Reactions)$
of non-void target reactions.
This definition is fairly general and the reader may wonder whether it can be
restricted without hurting the (asymptotic) runtime efficiency of $\Pi$.
In the current section, we show that this is not the case for four ``natural''
such restrictions as stated in \Prop{}
\ref{proposition:justify-runtime-policy:fixed},
\ref{proposition:justify-runtime-policy:singleton},
\ref{proposition:justify-runtime-policy:subset-escaping},
and
\ref{proposition:justify-runtime-policy:superset-escaping}.

\begin{proposition}
\label{proposition:justify-runtime-policy:fixed}
There exists a CRN protocol
$\Pi = (\Species, \Reactions)$
such that
$\RunTime_{\Halt}^{\Pi}(n) = O (n)$,
however if we restrict the runtime policies $\RunPol$ so that
$\RunPol(\c) = Q$
for every configuration
$\c \in \Naturals^{\Species}$,
where $Q$ is a fixed subset of $\NonVoid(\Reactions)$, then
$\RunTime_{\Halt}^{\Pi}(n)$
cannot be bounded as a function of $n$.
\end{proposition}

\begin{proposition}
\label{proposition:justify-runtime-policy:singleton}
There exists a CRN protocol
$\Pi = (\Species, \Reactions)$
such that
$\RunTime_{\Stab}^{\Pi}(n) = O (\log n)$,
however if we restrict the runtime policies $\RunPol$ so that
$|\RunPol(\c)| \leq 1$
for every configuration
$\c \in \Naturals^{\Species}$,
then
$\RunTime_{\Stab}^{\Pi}(n) = \Omega (n)$.
\end{proposition}

The following two propositions rely on the notation $\mathcal{E}(\c)$,
denoting the set of reactions that escape from the component of configuration
$\c \in \Naturals^{\Species}$
in the configuration digraph $D^{\Pi}$.

\begin{proposition}
\label{proposition:justify-runtime-policy:subset-escaping}
There exists a CRN protocol
$\Pi = (\Species, \Reactions)$
such that
$\RunTime_{\Stab}^{\Pi}(n) = O (\log n)$,
however if we restrict the runtime policies $\RunPol$ so that
$\RunPol(\c) \subseteq \mathcal{E}(\c)$
for every configuration
$\c \in \Naturals^{\Species}$,
then
$\RunTime_{\Stab}^{\Pi}(n) = \Omega (n)$.
\end{proposition}

\begin{proposition}
\label{proposition:justify-runtime-policy:superset-escaping}
There exists a CRN protocol
$\Pi = (\Species, \Reactions)$
such that
$\RunTime_{\Halt}^{\Pi}(n) = O (n)$,
however if we restrict the runtime policies $\RunPol$ so that
$\RunPol(\c) \supseteq \mathcal{E}(\c)$
for every configuration
$\c \in \Naturals^{\Species}$,
then
$\RunTime_{\Halt}^{\Pi}(n) = \Omega (n \log n)$.
\end{proposition}

\Prop{}
\ref{proposition:justify-runtime-policy:fixed},
\ref{proposition:justify-runtime-policy:singleton},
\ref{proposition:justify-runtime-policy:subset-escaping},
and
\ref{proposition:justify-runtime-policy:superset-escaping}
are proved in \Sect{}
\ref{section:justify-runtime-policy:fixed},
\ref{section:justify-runtime-policy:singleton},
\ref{section:justify-runtime-policy:subset-escaping}, and
\ref{section:justify-runtime-policy:superset-escaping},
respectively.
Each proof is followed by a short discussion explaining why the adversarial
runtime obtained with our general definition is intuitively more plausible
than that obtained with the more restricted definition.

\subsection{Fixed Policies}
\label{section:justify-runtime-policy:fixed}
In this section, we prove
\Prop{}~\ref{proposition:justify-runtime-policy:fixed}.
To this end, consider the CRN protocol
$\Pi = (\Species, \Reactions)$
defined over the species set
$\Species = \{ A_{0}, A_{1}, X_{0}, X_{1}, W \}$
and the following non-void reactions:
\\
$\beta$:
$X_{0} + X_{1} \rightarrow 2 W$;
\\
$\gamma_{0}$:
$A_{1} + X_{0} \rightarrow A_{0} + X_{0}$;
and
\\
$\gamma_{1}$:
$A_{0} + X_{1} \rightarrow A_{1} + X_{1}$.
\\
A configuration
$\c^{0} \in \Naturals^{\Species}$
is valid as an initial configuration of $\Pi$ if
$\c^{0}(\{ A_{0}, A_{1} \}) = 1$
and
$\c^{0}(X_{0}) \neq \c^{0}(X_{1})$.

\begin{observation}
\label{observation:justify-runtime-policy:fixed:prefix}
For every weakly fair valid execution
$\eta = \langle \c^{t}, \alpha^{t} \rangle_{t \geq 0}$,
there exists a step
$\hat{t} \geq 0$
such that
\\
(1)
$\min \{ \c^{t}(X_{0}), \c^{t}(X_{1}) \} > 0$
for every
$0 \leq t < \hat{t}$;
and
\\
(2)
$\min \{ \c^{t}(X_{0}), \c^{t}(X_{1}) \} = 0$
for every
$t \geq \hat{t}$.
\end{observation}

\begin{observation}
\label{observation:justify-runtime-policy:fixed:halting}
Every weakly fair execution emerging from a valid initial configuration
$\c^{0} \in \Naturals^{\Species}$
with
$\c^{0}(X_{j}) > \c^{0}(X_{1 - j})$,
$j \in \{ 0, 1 \},$
halts into a configuration
$\c \in \Naturals^{\Species}$
that satisfies
\begin{itemize}

\item
$\c(A_{j}) = 1$;

\item
$\c(X_{j}) = \c^{0}(X_{j}) - \c^{0}(X_{1 - j})$;
and

\item
$\c(X_{1 - j}) = \c(A_{1 - j}) = 0$.

\end{itemize}
\end{observation}

Consider the runtime policy $\RunPol$ defined as
\[\textstyle
\RunPol(\c)
\, = \,
\begin{cases}
\{ \beta \}
\, , &
\c(X_{0}) > 0 \, \land \, \c(X_{1}) > 0
\\
\{ \gamma_{0}, \gamma_{1} \}
\, , &
\text{otherwise}
\end{cases}
\, .
\]
Fix a weakly fair valid execution
$\eta = \langle \c^{t}, \alpha^{t} \rangle_{t \geq 0}$
of initial molecular count
$\| \c^{0} \| = n$
and a skipping policy $\SkipPol$.

\begin{lemma}
\label{lemma:justify-runtime-policy:fixed:good-policy}
$\RunTime_{\Halt}^{\RunPol, \SkipPol}(\eta) \leq O (n)$.
\end{lemma}
\begin{proof}
Let $\EffStep(i)$ and
$\EffConf^{i} = \c^{\EffStep(i)}$
be the effective step and effective configuration, respectively, of round
$i \geq 0$
under $\RunPol$ and $\SkipPol$.
Let
$\hat{i}
=
\min \{ i \geq 0 \mid \EffStep(i) \geq \hat{t} \}$,
where $\hat{t}$ is the step promised in
\Obs{}~\ref{observation:justify-runtime-policy:fixed:prefix}, and let
$i^{*}
=
\min \{ i \geq 0 \mid \EffStep(i) \geq t^{*} \}$,
where $t^{*}$ is the halting step promised in
\Obs{}~\ref{observation:justify-runtime-policy:fixed:halting}.
We establish the assertion by proving the following two claims:
\\
(C1)
the total contribution of rounds
$0 \leq i < \hat{i}$
to
$\RunTime_{\Halt}^{\RunPol, \SkipPol}(\eta)$
is up-bounded by
$O (n)$;
and
\\
(C2)
$i^{*} - \hat{i} \leq 1$.
\\
Indeed, the assertion follows as the temporal cost charged to a single round
is at most
$O (n)$.

To establish claim (C1), let
$\ell_{i} = \min \{ \EffConf^{i}(X_{0}), \EffConf^{i}(X_{1}) \}$
for each round
$0 \leq i < \hat{i}$.
Notice that
$\Prp_{\EffConf^{i}}(\beta)
=
\frac{\EffConf^{i}(X_{0}) \cdot \EffConf^{i}(X_{1})}{\Vol}
\geq
\Omega (\ell_{i}^{2} / n)$
and that
$\Prp_{\c}(\beta) = \Prp_{\EffConf^{i}}(\beta)$
for every configuration
$\c \in \Naturals^{\Species}$
reachable form $\EffConf^{i}$ via a $\RunPol$-avoiding path.
Employing \Lem{}~\ref{lemma:runtime:tools:temporal-cost}, we conclude that
$\TemporalCost^{\RunPol}(\EffConf^{i})
\leq
O (n / \ell_{i}^{2})$.
Since
$\ell_{i + 1} > \ell_{i}$
for every
$0 \leq i < \hat{i}$,
it follows that
\[\textstyle
\sum_{i = 0}^{\hat{i}} \TemporalCost^{\RunPol}(\EffConf^{i})
\, \leq \,
\sum_{i = 0}^{\hat{i}} O (n / \ell_{i}^{2})
\, \leq \,
O (n) \cdot \sum_{\ell = 1}^{\infty} \frac{1}{\ell^{2}}
\, = \,
O (n)
\, .
\]

To establish claim (C2), it suffices to observe that from step $\hat{t}$
onward, only one of the reactions $\gamma_{0}$ and $\gamma_{1}$ may be
applicable and that the execution halts once this reaction is scheduled.
\end{proof}

Now, consider a runtime policy $\RunPol_{Q}$ with a fixed target reaction set
$Q \subseteq \{ \beta, \gamma_{0}, \gamma_{1} \}$,
that is,
$\RunPol_{Q}(\c) = Q$
for every configuration
$\c \in \Naturals^{\Species}$.
We argue that
$\gamma_{j} \in Q$
for each
$j \in \{ 0, 1 \}$.
Indeed, if
$\c(A_{1 - j}) = 1$,
$\c(X_{j}) > 0$,
and
$\c(X_{1 - j}) = 0$,
then
$\Applicable(\c) = \{ \gamma_{j} \}$,
hence $\RunPol_{Q}(\c)$ must include $\gamma_{j}$ in order to bound the halting
runtime.

However, if
$\{ \gamma_{0}, \gamma_{1} \} \subseteq Q$,
then the adversarial scheduler can construct a weakly fair valid execution
$\eta$ of initial molecular count $n$ in a manner that forces the protocol to
go through arbitrarily many rounds under $\RunPol_{Q}$ and the identity
skipping policy $\SkipPol_{\Identity}$ before the execution halts, charging an
$\Omega (1 / n)$
temporal cost to each one of them.
This is done simply by starting with a configuration that includes both
$X_{0}$ and $X_{1}$ molecules and then scheduling reactions $\gamma_{0}$ and
$\gamma_{1}$ in alternation.
We conclude that
$\RunTime_{\Halt}^{\RunPol_{Q}, \SkipPol_{\Identity}}(\eta)$
is unbounded (as a function of $n$), thus establishing
\Prop{}~\ref{proposition:justify-runtime-policy:fixed}.

In summary, a policy with a fixed target reaction set can inappropriately
reward an adversarial scheduler that delays progress indefinitely:
the longer the delay, the larger the runtime.
This runs counter to the philosophy of adversarial runtimes that are standard
in distributed computing, discussed in \Sect{}~\ref{section:introduction}.

\subsection{Singleton Target Reaction Sets}
\label{section:justify-runtime-policy:singleton}
In this section, we prove
\Prop{}~\ref{proposition:justify-runtime-policy:singleton}.
To this end, consider the CRN protocol
$\Pi = (\Species, \Reactions)$
defined over the species set
$\Species = \{ A, B, B' X, X', Y \}$
and the following non-void reactions:
\\
$\beta$:
$A + A \rightarrow 2 B$;
\\
$\gamma$:
$A + X \rightarrow A + X'$;
\\
$\gamma'$:
$A + X' \rightarrow A + X$;
\\
$\delta$:
$X + Y \rightarrow 2 Y$;
\\
$\delta'$:
$X' + Y \rightarrow 2 Y$;
\\
$\chi$:
$B + B \rightarrow 2 B'$;
and
\\
$\chi'$:
$B' + B' \rightarrow 2 B$.
\\
A configuration
$\c^{0} \in \Naturals^{\Species}$
is valid as an initial configuration of $\Pi$ if
$\c^{0}(A) = 2$,
$\c^{0}(Y) = 1$,
and
$\c^{0}(\{ B, B' \}) = 0$.

\begin{observation}
\label{observation:justify-runtime-policy:singleton:basic-properties}
The following properties hold for every weakly fair valid execution
$\eta = \langle \c^{t}, \alpha^{t} \rangle_{t \geq 0}$
and for every
$t \geq 0$:
\begin{itemize}

\item
$\c^{t}(A), \c^{t}(B), \c^{t}(B') \in \{ 0, 2 \}$

\item
$\c^{t + 1}(A) \leq \c^{t}(A)$;

\item
if
$\c^{t}(A) = 0$,
then
$\c^{t}(\{ B, B' \}) = 2$;

\item
$\c^{t + 1}(\{ X, X' \}) \leq \c^{t}(\{ X, X' \})$;
and

\item
$\c^{t + 1}(Y) \geq \c^{t}(Y)$.

\end{itemize}
\end{observation}

\begin{observation}
\label{observation:justify-runtime-policy:singleton:stabilization}
Every weakly fair execution emerging from a valid initial configuration
$\c^{0} \in \Naturals^{\Species}$
of molecular count
$\| \c^{0} \| = n$
stabilizes into the configurations $\c$ satisfying
$\c(A) = \c(\{ X, X' \}) = 0$
and
$\c(Y) = n - 2$.
\end{observation}

Consider the runtime policy $\RunPol$ defined as
\[\textstyle
\RunPol(\c)
\, = \,
\begin{cases}
\{ \beta, \delta, \delta' \}
\, , &
\c(A) = 2
\\
\{ \delta, \delta' \}
\, , &
\c(A) = 0
\end{cases}
\, .
\]
Fix a weakly fair valid execution
$\eta = \langle \c^{t}, \alpha^{t} \rangle_{t \geq 0}$
of initial molecular count
$\| \c^{0} \| = n$
and let $t^{*}$ be the stabilization step of $\eta$.
Fix a skipping policy $\SkipPol$ and let $\EffStep(i)$ and
$\EffConf^{i} = \c^{\EffStep(i)}$
be the effective step and effective configuration, respectively, of round
$i \geq 0$
under $\RunPol$ and $\SkipPol$.
Let
$i^{*}
=
\min \{ i \geq 0 \mid \EffStep(i) \geq t^{*} \}$.

\begin{lemma}
\label{lemma:justify-runtime-policy:singleton:temp-cost-mult}
The total contribution of rounds
$0 \leq i < i^{*}$
to
$\RunTime_{\Stab}^{\RunPol, \SkipPol}(\eta)$
is up-bounded by
$O (\log n)$.
\end{lemma}
\begin{proof}
Let
$\ell_{i} = \EffConf^{i}(Y)$
for each round
$0 \leq i < i^{*}$.
The key observation here is that every round
$0 \leq i < i^{*}$
is target-accomplished with
$\ell_{i + 1} \leq \ell_{i}$;
moreover, there exists at most one round
$0 \leq i < i^{*}$
such that
$\ell_{i + 1} = \ell_{i}$.
Since
$\Prp_{\c}(\RunPol(\EffConf^{i}))
=
\Prp_{\EffConf^{i}}(\RunPol(\EffConf^{i}))$
for every configuration
$\c \in \Naturals^{\Species}$
reachable form $\EffConf^{i}$ via a $\RunPol$-avoiding path and since
$\Prp_{\EffConf^{i}}(\RunPol(\EffConf^{i}))
\geq
\frac{\ell_{i} \cdot (n - \ell_{i} - 2)}{\Vol}
\geq
\Omega \left( \frac{\ell_{i} \cdot (n - \ell_{i})}{n} \right)$,
it follows by \Lem{}~\ref{lemma:runtime:tools:temporal-cost} that
$\TemporalCost^{\RunPol}(\EffConf^{i})
\leq
O \left( \frac{n}{\ell_{i} \cdot (n - \ell_{i})} \right)$.
As
$\ell_{0} \geq 1$
and
$\ell_{i^{*} - 1} \leq n - 3$,
we can bound the runtime of $\eta$ under $\RunPol$ and $\SkipPol$ as
\[\textstyle
\RunTime^{\RunPol, \SkipPol}(\eta)
\, = \,
\sum_{i = 0}^{i^{*} - 1} \TemporalCost^{\RunPol}(\EffConf^{i})
\, \leq \,
\sum_{\ell = 1}^{n - 3} O \left( \frac{n}{\ell \cdot (n - \ell)} \right)
\, \leq \,
O (n) \cdot \sum_{\ell = 1}^{n - 1} \frac{1}{\ell \cdot (n - \ell)}
\, = \,
O (\log n)
\, ,
\]
thus establishing the assertion.
\end{proof}

Next, let us examine the efficiency of the runtime policies all of whose
target reaction sets are of size (at most) $1$.
To this end, consider such a runtime policy $\RunPol$ and the configuration
set
\[\textstyle
S_{0}
\, = \,
\left\{
\c \in \Naturals^{\Species}
\mid
\c(A) = 2
\land
\c(Y) = 1
\land
\c(\{ B, B' \}) = 0
\land
\| \c \| = n
\right\}
\, .
\]
By definition, every configuration in $S_{0}$ is a valid initial configuration
of $\Pi$.
Moreover, the set $S_{0}$ forms a component of the configuration digraph
$D^{\Pi}$.
As $\beta$ is the only configuration that escapes $S_{0}$, we deduce that
there must exist a configuration
$\c \in S_{0}$
such that
$\RunPol(\c) = \{ \beta \}$;
indeed, if
$\beta \notin \RunPol(\c)$,
then the adversarial scheduler can generate an arbitrarily long sequence of
rounds all of whose effective configurations are $\c$, thus pumping up the
stabilization runtime of $\Pi$.
Let
$\hat{\c} \in S_{0}$
be such a configuration.

To low-bound the stabilization runtime of $\Pi$, construct a weakly fair valid
execution $\eta$ and a skipping policy $\SkipPol$ such that the effective
configuration of round $0$ is
$\EffConf^{0} = \hat{\c}$.
Notice that
$\Prp_{\hat{\c}}(\beta) = 2 / \Vol$
and that
$\Prp_{\c}(\beta) = \Prp_{\hat{\c}}(\beta)$
for every configuration $\c$ reachable from $\hat{\c}$ via a
$\RunPol$-avoiding path.
Therefore, we can employ \Lem{}~\ref{lemma:runtime:tools:temporal-cost} to
conclude that
$\TemporalCost^{\RunPol}(\hat{\c})
=
\Vol / 2
=
\Omega (n)$.
\Prop{}~\ref{proposition:justify-runtime-policy:singleton} follows as the
temporal cost of round $0$ is clearly (weakly) dominated by the stabilization
runtime of the entire execution.

In this example, restricting the runtime policy to a singleton target set
limits the meaning of ``progress'' in an artificial way.
Specifically, it that does not give the protocol designer credit for ensuring
that progress is indeed possible from the effective configuration identified
by the adversary.

\subsection{Target Reactions $\subseteq$ Escaping Reactions}
\label{section:justify-runtime-policy:subset-escaping}
Our goal in this section is to prove
\Prop{}~\ref{proposition:justify-runtime-policy:subset-escaping}.
To this end, consider the CRN protocol
$\Pi = (\Species, \Reactions)$
introduced in \Sect{}~\ref{section:justify-runtime-policy:singleton} and
recall that
\[\textstyle
\RunTime_{\Stab}^{\Pi}(n)
\, = \,
O (\log n)
\, .
\]
Consider the configuration set $S_{0}$ introduced in
\Sect{}~\ref{section:justify-runtime-policy:singleton} and recall that
$S_{0}$ forms a component of the configuration digraph $D^{\Pi}$ and that
$\beta \in \Reactions$
is the only reaction that escapes from $S_{0}$.
Moreover, in \Sect{}~\ref{section:justify-runtime-policy:singleton}, we
prove that if a runtime policy $\RunPol$ satisfies
$\RunPol(\c) = \{ \beta \}$
for some configuration
$\c \in S_{0}$,
then there exist a weakly fair valid execution $\eta$ of initial molecular
count $n$ and a skipping policy $\SkipPol$ such that
\[\textstyle
\RunTime_{\Stab}^{\RunPol, \SkipPol}(\eta) = \Omega (n)
\, ,
\]
thus establishing
\Prop{}~\ref{proposition:justify-runtime-policy:subset-escaping}.

\subsection{Target Reactions $\supseteq$ Escaping Reactions}
\label{section:justify-runtime-policy:superset-escaping}
This section is dedicated to proving
\Prop{}~\ref{proposition:justify-runtime-policy:superset-escaping}.
To this end, consider the CRN protocol
$\Pi = (\Species, \Reactions)$
defined over the species set
$\Species = \{  L, X, W \}$
and the following non-void reactions:
\\
$\alpha$: $L + X \rightarrow L + W$;
and
\\
$\beta$: $L + L \rightarrow 2W$.
\\
A configuration
$\c^{0} \in \Naturals^{\Species}$
is valid as an initial configuration of $\Pi$ if
$\c^{0}(L) = 2$
and
$\c^{0}(W) = 0$.

\begin{observation}
\label{observation:justify-runtime-policy:superset-escaping:halting}
For every weakly fair valid execution
$\eta = \langle \c^{t}, \zeta^{t} \rangle_{t \geq 0}$,
there exists a step
$t^{*} > 0$
such that
\\
(1)
$\beta \in \Applicable(\c^{t})$
for every
$0 \leq t < t^{*}$;
\\
(2)
$\zeta^{t^{*} - 1} = \beta$;
and
\\
(3)
$\c^{t^{*}}$ is a halting configuration.
\end{observation}

Notice that each configuration
$\c \in \Naturals^{\Species}$
forms a singleton component in the configuration digraph $D^{\Pi}$ and that
every applicable non-void reaction is escaping;
that is, if
$\c(L) = 2$,
then
$\mathcal{E}(\c) = \{ \alpha, \beta \}$.
Therefore, the only runtime policy that satisfies the restriction presented in
\Prop{}~\ref{proposition:justify-runtime-policy:superset-escaping} is the
runtime policy $\RunPol_{\supseteq}$ defined so that
$\RunPol_{\supseteq}(\c) = \{ \alpha, \beta \}$
for every configuration
$\c \in \Naturals^{\Species}$
with
$\c(L) = 2$.

\begin{lemma}
\label{lemma:justify-runtime-policy:superset-escaping:superset-policy}
For every sufficiently large $n$, there exist a weakly fair valid execution
$\eta$ of initial molecular count $n$ and a skipping policy $\SkipPol$ such
that
$\RunTime_{\Halt}^{\RunPol_{\supseteq}, \SkipPol}(\eta) = \Omega (n \log n)$.
\end{lemma}
\begin{proof}
Construct the execution
$\eta = \langle \c^{t}, \zeta^{t} \rangle_{t \geq 0}$
by scheduling
$\zeta^{t} = \alpha$
for
$t = 0, 1, \dots, n - 3$,
i.e., as long as
$\c^{t}(X) > 0$,
and then scheduling
$\zeta^{n - 2} = \beta$.
This means that
$\c^{t} = 2 L + (n - 2 - t) X + t W$
for every
$0 \leq t \leq n - 2$
and that
$\c^{t} = n W$
for every
$t \geq n - 1$.

Let $\SkipPol$ be the identity skipping policy mapping each step
$t \geq 0$
to
$\SkipPol(t) = t$.
Under $\RunPol_{\supseteq}$ and $\SkipPol$, each step constitutes a full
round, so each configuration $\c^{t}$ is the effective configuration of its
own round.
Since
$\Prp_{\c^{t}}(\{ \alpha, \beta \}) = \frac{1 + 2 (n - 2 - t)}{\Vol}$,
it follows that
$\TemporalCost^{\RunPol_{\supseteq}}(\c^{t})
=
\frac{\Vol}{1 + 2 (n - 2 - t)}$
for each
$0 \leq t \leq n - 2$.
We conclude that
\[\textstyle
\RunTime_{\Halt}^{\RunPol_{\supseteq}, \SkipPol}(\eta)
\, \geq \,
\sum_{t = 0}^{n - 2} \frac{\Vol}{1 + 2 (n - 2 - t)}
\, = \,
\Omega (n) \cdot \sum_{\ell = 1}^{n} \frac{1}{\ell}
\, = \,
\Omega (n \log n)
\, ,
\]
thus establishing the assertion.
\end{proof}

Next, consider the runtime policy $\RunPol$ defined so that
$\RunPol(\c) = \{ \beta \}$
for every configuration
$\c \in \Naturals^{\Species}$
with
$\c(L) = 2$.
Fix a weakly fair valid execution
$\eta = \langle \c^{t}, \zeta^{t} \rangle_{t \geq 0}$
of initial molecular count $n$ and a skipping policy $\SkipPol$ and let
$\EffStep(i)$ and
$\EffConf^{i} = \c^{\EffStep(i)}$
be the effective step and effective configuration, respectively, of round
$i \geq 0$
under $\RunPol$ and $\SkipPol$.

The key observation now is that when round $0$ ends, the execution must have
halted, i.e.,
$\EffStep(1) \geq t^{*}$,
where $t^{*}$ is the step promised in
\Obs{}~\ref{observation:justify-runtime-policy:superset-escaping:halting}.
As the temporal cost charged to a single round is always
$O (n)$,
we conclude that
\[\textstyle
\RunTime_{\Halt}^{\RunPol, \SkipPol}(\eta)
\, \leq \,
O (n)
\, ,
\]
thus establishing
\Prop{}~\ref{proposition:justify-runtime-policy:superset-escaping}.

The intuition behind this example is that by repeatedly scheduling $\alpha$,
the adversary can postpone halting, but not indefinitely.
Should the protocol designer be charged the temporal cost of \emph{every}
adversarially-scheduled postponing step?
It is reasonable to argue that the answer is no:
The adversary drives the execution to a pitfall, and the protocol designer
should pay for that (to the tune of
$O (n)$),
but should not have to pay for every step of the execution that leads to the
pitfall.
\LongVersionEnd 

\LongVersion 
\section{Large Adversarial Runtime vs.\ Small Expected
Stochastic Runtime}
\label{section:large-adversarial-small-stochastic}
This section focuses on the ability of the weakly fair adversarial scheduler to
slow down the execution of CRN protocols.
In particular, we show that the stabilization/halting runtime of a CRN
protocol $\Pi$ operating under the weakly fair adversarial scheduler may be
significantly larger than the expected stochastic stabilization/halting
runtime of the same protocol $\Pi$ when it operates under the stochastic
scheduler.
This phenomenon is demonstrated by two CRN protocols in which the
aforementioned gap is obtained using different strategies:
In \Sect{}~\ref{section:large-adversarial-small-stochastic:multiple-pitfalls},
we present a protocol designed so that the (weakly fair) adversarial scheduler
can lead the execution through a sequence of (asymptotically) many pitfall
configurations before it stabilizes;
a stochastic execution, on the other hand, avoids all those pitfall
configurations with high probability and thus, stabilizes much faster.
In contrast, the protocol presented in
\Sect{}~\ref{section:large-adversarial-small-stochastic:round-inflation} does
not admit any pitfall configurations;
rather, this protocol is designed so that the adversarial scheduler can
``extend'' the execution far beyond what one would expect from a
stochastically generated execution, thus charging the protocol's runtime for
an inflated number of rounds.

\subsection{Reaching Multiple Pitfall Configurations}
\label{section:large-adversarial-small-stochastic:multiple-pitfalls}
This section presents a CRN protocol for which there exists a weakly fair
execution that reaches $\Theta (n)$ pitfall configurations before stabilization.
However, under a stochastic scheduler, the execution reaches stabilization with a high probability of avoiding any pitfall configurations.\footnote{%
We say that event occurs with high probability if its probability is at least
$1 - n^{-c}$
for an arbitrarily large constant $c$.}
Consider the CRN protocol
$\Pi = (\Species, \Reactions)$
presented in \Fig{}~\ref{figure:crn-multiple-pitfalls:protocol} and the valid initial configuration
$\c^{0} = L_{0} + C + (n - 2) X$.
In \Fig{}~\ref{figure:crn-multiple-pitfalls:runtime-policy},
we devise a runtime policy $\RunPol$ demonstrating that the stabilization runtime of $\Pi$ is
$\RunTime_{\Stab}^{\Pi}(n) = \Theta (n^{2})$.
Next, we show that on stochastically generated execution
$\eta_{r} = \langle \c^{t}, \zeta^{t} \rangle_{0 \leq t}$
of initial molecular count $n$, the expected (stochastic) stabilization runtime of $\Pi$ is
$O (n)$.

Notice that after every consumption of an $X$ molecule (and production of a $Y$ molecule), the resulting configuration $\c$ satisfies
$\c(L_{0}) = 1$.
Recalling that $\c^{0}(L_{0}) = 1$ and that $\c^{0}(X) = n - 2$, we conclude that until stabilizing, the stochastic execution reaches exactly $\frac{1}{2} n$ configurations
$\c_{L_{0}}^{0}, \c_{L_{0}}^{1}, \dots \c_{L_{0}}^{\frac{1}{2} n - 1}$
such that $\c_{L_{0}}^{j}(L_{0}) = 1$
for every
$0 \leq j \leq \frac{1}{2} n - 1$.
\begin{observation}
\label{observation:large-adversarial-small-stochastic:multiple-pitfalls:X}
For every
$0 \leq j \leq \frac{1}{2} n - 1$
it holds that
$\c_{L_{0}}^{j}(X) = n - 2 - j$.
\end{observation}

Fix some
$0 \leq j \leq \frac{1}{2} n - 2$.
We analyze the contribution to the expected stochastic runtime of the step interval
$[t_{j}, t_{j + 1} - 1]$
where
$\c^{t_{j}} = \c_{L_{0}}^{j}$
and
$\c^{t_{j + 1}} = \c_{L_{0}}^{j + 1}$.
Observe that exactly one of the following holds:
(I)
$1 \leq t_{j + 1} - t_{j} \leq k$; or
(II)
$t_{j + 1} - t_{j} = k + 2$
(refer to the configuration digraph $D^{\Pi}$ for $k = 2$ shown in \Fig~\ref{figure:crn-multiple-pitfalls:digraph}).
Recalling that
$\Prp_{c^{t}}(\{ \gamma_{0}, \dots \gamma_{k} \}) \leq \frac{1}{\Vol}$
for every
$t \geq 0$,
combined with \Obs~\ref{observation:large-adversarial-small-stochastic:multiple-pitfalls:X},
we conclude that the probability of
$t_{j + 1} - t_{j} = q$
for
$1 \leq q \leq k$
(event (I))
is
$\left( \frac{1}{n - 1 - j} \right)^{q - 1} \cdot \frac{n - 2 - j}{n - 1 - j}$,
and probability of event (II) is
$\left( \frac{1}{n - 1 -j} \right)^{k}$.
Therefore we get
\begin{align*}
&
\Ex (t_{j + 1} - t_{j})
\\
= \, &
n \left(
\sum_{q = 1}^{k} \left(
\left( \tfrac{1}{n - 1 - j} \right)^{q - 1}
\cdot
\tfrac{n - 2 - j}{n - 1 - j}
\cdot
\left( (q - 1) n + \tfrac{n}{n - 2 - j} \right)
\right)
+
(k + 1) n \cdot \left( \tfrac{1}{n - 1 - j} \right)^{k}
\right)
\\
= \, &
\Theta \left( \tfrac{n^{2}}{n - j} \right) 
\, .
\end{align*}
Note that we can take $k$ to be any arbitrarily large constant.
Thus, it can be shown that event (I) occurs for every 
$0 \leq j \leq \frac{1}{2} n - 2$
with high probability, 
which means that the execution does not reach a stabilizing pitfall configuration.

Let
\[
\hat{t}
\, = \,
\min \{ t > 0 \mid \c^{t}(X) \leq \c^{t}(Y) \}
\]
be the stabilization step of $\eta_{r}$.
We can bound $\Ex (\hat{t})$
as
\[
\Ex (\hat{t})
\, \leq \, 
\left( \sum_{j = 0}^{\frac{1}{2} n - 2} O \left( \frac{n^{2}}{n - j} \right) \right)
\, = \, 
O \left( n^{2} \right) \, .
\]
The upper bound follows by recalling that the time span of each step is
$\Theta (1 / n)$.
Notice that in order to reach a halting configuration, the final non-void reaction must be $\gamma_{k}$, which shifts the $L_{k}$ molecule into $L_{k + 1}$.
It can be shown that the expected stochastic halting runtime is $O (n \log n)$, which is still asymptotically faster than the adversarial (stabilization and halting) runtime.

\subsection{Round Inflation}
\label{section:large-adversarial-small-stochastic:round-inflation}
In this section, we present a CRN protocol whose halting runtime is
asymptotically larger than the stochastic halting runtime due to a larger
number of rounds.
Consider the CRN protocol
$\Pi = (\Species, \Reactions)$
presented in \Fig{}~\ref{figure:crn-kill-alternating-leaders:protocol}.
Intuitively, an execution of $\Pi$ halts once a $\beta$ reaction is scheduled.
The adversary can delay halting by scheduling $\gamma$ reactions again and
again until all of the $X$ molecules are used up.
In particular, the adversary can continue to schedule $\gamma$ reactions when
few $X$ molecules remain, even though such reactions have low propensity and
are not likely to be scheduled stochastically.
By doing so, the adversary inflates the number of rounds, thus increasing the
execution’s runtime.

Consider the valid initial configuration
$\c^{0} = L_{0} + ((n / 2) - 1) X + (n / 2) Y$
for some sufficiently large even integer $n$ and fix a weakly fair execution
$\eta = \langle \c^{t}, \zeta^{t} \rangle_{t \geq 0}$
emerging from $\c^{0}$.

\begin{observation}
\label{observation:large-adversarial-small-stochastic:round-inflation:counts}
For every
$t \geq 0$,
we have
\\
(1)
$\c^{t + 1}(X) \leq \c^{t}(X)$;
and
\\
(2)
$\c^{t + 1}(Y) \geq \c^{t}(Y)$.
\end{observation}

\begin{observation}
\label{observation:large-adversarial-small-stochastic:round-inflation:halting}
There exists a step
$t^{*} > 0$
such that
$\c^{t}(\{ L_{0}, L_{1} \}) = 1$
for every
$0 \leq t < t^{*}$
and
$\c^{t}(L_{0}, L_{1}) = 0$
for every
$t \geq t^{*}$.
Moreover, $t^{*}$ is the halting step of $\eta$.
\end{observation}

We first argue that if $\eta$ is generated by the stochastic scheduler, then
the expected stochastic runtime of the execution prefix
$\langle \c^{t}, \zeta^{t} \rangle_{0 \leq t < t^{*}}$
is
$O (1)$,
where
$t^{*} > 0$
is the halting step promised in
\Obs{}~\ref{observation:large-adversarial-small-stochastic:round-inflation:halting}.
Indeed, since the molecular count of $Y$ remains
$\Omega (n)$
throughout the execution, it follows that at any step
$0 \leq t < t^{*}$,
a $\beta$ reaction is scheduled, which causes $\eta$ to halt, with
probability
$\Omega (1 / n)$.
Thus, in expectation, it takes
$O (n)$
steps until $\eta$ halts, whereas the time span of each step is
$O (1 / n)$.

On the other hand, we argue that a weakly fair adversarial scheduler can
construct the execution $\eta$ so that its (adversarial) halting runtime is
$\Omega (n)$.
To this end, denote
$x = \c^{0}(X)$,
recalling that
$x = \Omega (n)$.
The adversarial scheduler constructs the execution prefix
$\langle \c^{t}, \zeta^{t} \rangle_{0 \leq t < x}$
by setting
\[\textstyle
\zeta^{t}
\, = \,
\begin{cases}
\gamma_{0}
\, , &
t = 0 \bmod 2
\\
\gamma_{1}
\, , &
t = 1 \bmod 2
\end{cases}
\, ,
\]
and chooses the skipping policy to be the identity skipping policy
$\SkipPol_{\Identity}$, which leads to the following observation.

\begin{observation}
\label{observation:large-adversarial-small-stochastic:round-inflation:applicable}
The execution $\eta$ satisfies
\[\textstyle
\Applicable(\c^{t})
\, = \,
\begin{cases}
\{ \beta_{0}, \gamma_{0} \}
\, , &
t = 0 \bmod 2
\\
\{ \beta_{1}, \gamma_{1} \}
\, , &
t = 1 \bmod 2
\end{cases}
\]
for every
$0 \leq t < x$
and
$\Applicable(\c^{x})
=
\{ \beta_{0}, \beta_{1} \} - \Applicable(\c^{x - 1})$.
\end{observation}

Fix a runtime policy $\RunPol$ for $\Pi$ and let $\EffStep(i)$ and
$\EffConf^{i} = \c^{\EffStep(i)}$
be the effective step and effective configuration, respectively, of round
$i \geq 0$ under $\RunPol$ and $\SkipPol_{\Identity}$.
\Obs{}~\ref{observation:large-adversarial-small-stochastic:round-inflation:applicable}
implies that
$\EffStep(i) = i$
for every
$0 \leq i < x$,
hence $\eta$ includes (at least)
$x = \Omega (n)$
rounds before it halts.
The key observation now is that the temporal cost of each round
$0 \leq i < x$
is
$\Omega (1)$;
this holds since
$\c(\{ L_{0}, L_{1} \}) \leq 1$
for every configuration $\c$ reachable from $\EffConf^{i}$, whereas $L_{0}$
and $L_{1}$ are reactants of each reaction in $\NonVoid(\Reactions)$ and, in
particular, in the target reaction set $\RunPol(\EffConf^{i})$ of round $i$.
Therefore, the halting runtime of $\eta$ under $\RunPol$ and
$\SkipPol_{\Identity}$ is
$\Omega (n)$,
as promised.

In \Fig{}~\ref{figure:crn-kill-alternating-leaders:runtime-policy}, we devise
a runtime policy $\RunPol$ demonstrating that the
$\Omega (n)$
halting runtime lower bound of $\Pi$ is tight.
The reader might question why the temporal cost of the ``late rounds'' under
$\RunPol$ is
$\Theta (1)$
although the propensity of the $\gamma$ reactions (scheduled by the
adversarial scheduler) in those rounds is
$\Theta (1 / n)$.
Intuitively, this captures the fact that the temporal cost associated with a
round is independent of the reactions scheduled by the adversarial scheduler
in that round;
rather, it is determined by the reactions targeted by the protocol designer
(for that specific round).
Put differently, while the adversary has the power to determine which
reactions are actually scheduled, the adversary has no control over the
expected time for the target reactions to be scheduled.
\LongVersionEnd 

\LongVersion 
\section{Additional Related Work and Discussion}
\label{section:related-work}
The runtime of stochastically scheduled CRNs is the focus of a vast body of
literature, mainly under the umbrella of population protocols.
When it comes to stochastically scheduled executions
$\eta = \langle \c^{t}, \alpha^{t} \rangle$,
an important distinction is made between stabilizing into a desired
configuration set $Z$ in step $t^{*}$, which means that
$\c \in Z$
for all configurations $\c$ reachable from $\c^{t^{*}}$ (as defined in
\Sect{}~\ref{section:model}), and \emph{converging} into $Z$ in step $t^{*}$,
which means that
$\c^{t} \in Z$
for all
$t \geq t^{*}$.
The latter notion is weaker as it applies only to the configurations visited
(after step $t^{*}$) by the specific execution $\eta$ and does not rule out
the existence of some configurations
$\c \notin Z$
that can be reached.
Moreover, in contrast to stabilization (and halting), the notion of
convergence does not make much sense once we switch to the (weakly or strongly
fair) adversarial scheduler.

Angluin et al.~\cite{AngluinAE2008fast} prove that if a context is available
(which is equivalent to having a designated leader in the initial
configuration), then any semilinear predicate can be decided by a protocol
whose expected convergence runtime (under the stochastic scheduler) is
$\operatorname{polylog}(n)$.
However, the expected stabilization runtime of these protocols is
$\Omega (n)$
(see the discussion in \cite{BellevilleDS2017hardness}).
For leaderless protocols, Belleville et al.~\cite{BellevilleDS2017hardness}
establish an
$\Omega (n)$
lower bound on the expected stabilization runtime (under the stochastic
scheduler) of protocols that decide any predicate which is not
\emph{eventually constant} --- a restricted subclass of semilinear predicates
(that includes the detection predicates).
This leaves an open question regarding the expected runtime of leaderless
protocols that decide non-detection eventually constant predicates and another
open question regarding the expected runtime of protocols that decide
non-eventually constant predicates with a non-zero context --- see
Table~\ref{table:predicate-decidability-landscape-stochastic}.
Notice that both these questions are answered once we switch to the
adversarial scheduler and runtime definition of the current paper --- see
Table~\ref{table:predicate-decidability-landscape-adversarial}.

Our goal of measuring the runtime of protocols in scenarios that go beyond the
(perhaps less realistic) assumption of ``purely'' stochastically scheduled
executions is shared with other papers, again, mainly in the context of
population protocols.
In particular, Schwartzman and Sudo~\cite{SS21} study population protocols in
which an adversary chooses which agents interact at each step, but with
probability $p$, the adversary's choice is replaced by a randomly chosen
interaction (the population protocols analog of \emph{smoothed analysis}
\cite{SpielmanT2004smoothed}).
Angluin et al.~\cite{AngluinAE2008simple} consider population protocols in
which a small proportion of agents can assume any identity, and can change
that identity over time, thereby skewing the rates at which reactions occur.
Both of these runtime definitions seem quite specific in their modeling
choices relative to those of the current paper that considers arbitrary
(weakly fair) executions.
Yet other models of faulty interactions are studied by Di Luna et al.\
\cite{DILUNA201935,
diluna-etal-20},
but runtime analysis is still done using the stochastic model, so the emphasis
is more on protocol correctness.

The fundamental work of Chen et al.~\cite{ChenDSR2021rateindependent} on
rate-independent computations in continuous CRNs introduces and relies
critically on notions of adversarial schedulers and fairness.
In their continuous CRN model, configurations are vectors of concentrations of
chemical species (rather than vectors of integral species counts as in the
discrete model described in our paper).
Trajectories describe how configurations change over time as a result of
reaction ``fluxes''.
Roughly, in a fair continuous-time CRN trajectory, an adversary determines how
flux is pushed through the CRN's reactions, but must ensure that applicable
reactions eventually occur.
This notion of a fair adversarial scheduler corresponds quite naturally to the
weakly fair adversarial scheduler in our paper, and is used by Chen et al.\ to
characterize what functions can be stably computed by continuous CRNs.
Chen et al.\ do not provide results on the time complexity of continuous-time
CRNs, and note that this can be challenging
\cite{SeeligS2009timecomplexity,
ChenDSR2021rateindependent,
VasicCLKS2022-programming}.

The current paper leaves several open questions, starting with the ones that
stick out from Table~\ref{table:predicate-decidability-landscape-adversarial}:
Do there exist vote amplified CRDs for (non-detection) semilinear predicates
whose halting runtime (under an adversarial scheduler) is
$O (n)$?
If so, can these CRDs be leaderless?
Next, while this paper focuses on the adversarial runtime of predicate
decidability tasks, much attention has been devoted in the CRN literature also
to function computation tasks
\cite{ChenDS2014deterministic,
DotyH2015leaderless,
BellevilleDS2017hardness}
which calls for a rigorous study of their adversarial runtime.
Finally, perhaps the framework of our paper, namely partitioning executions
into rounds by means of policies chosen by the adversary and protocol
designer, could be useful in formulating a runtime notion of rate-independent
continuous CRNs.
\LongVersionEnd 

\clearpage
\bibliography{references}

\LongVersion 
\clearpage
\appendix

\begin{figure}[!t]
{\centering
\Large{APPENDIX}
\par}
\end{figure}

\section{Bounding the Temporal Cost}
\label{appendix:proof:lemma:runtime:tools:temporal-cost}
\begin{proof}[Proof of \Lem{}~\ref{lemma:runtime:tools:temporal-cost}]
Let
$\eta_{r} = \langle \c_{r}^{t}, \alpha_{r}^{t} \rangle_{t \geq 0}$
be a stochastic execution emerging from
$\c = \c_{r}^{0}$
and for
$t \geq 0$,
denote
$\eta_{r}^{< t}
=
\langle \c_{r}^{t'}, \alpha_{r}^{t'} \rangle_{t' \in [0, t)}$.
Taking
$\Stop = \Stop(\eta_{r}, 0, \RunPol(\c))$,
we develop
\begin{align*}
\TemporalCost^{\RunPol}(\c)
\, = \, &
\Ex_{\eta_{r}} \left(
\sum_{t = 0}^{\infty} \Indicator_{\Stop > t} \cdot \frac{1}{\Prp_{\c_{r}^{t}}}
\right)
\\
\, = \, &
\sum_{t = 0}^{\infty} \Ex_{\eta_{r}} \left(
\Indicator_{\Stop > t} \cdot \frac{1}{\Prp_{\c_{r}^{t}}}
\right)
\\
\, = \, &
\sum_{t = 0}^{\infty} \Ex_{\eta_{r}^{< t}} \left(
\Indicator_{\Stop > t} \cdot \frac{1}{\Prp_{\c_{r}^{t}}}
\right)
\\
\leq \, &
\frac{1}{p} \cdot
\sum_{t = 0}^{\infty} \Ex_{\eta_{r}^{< t}} \left(
\Indicator_{\Stop > t}
\cdot
\frac{\Prp_{\c_{r}^{t}}(\RunPol(\c))}{\Prp_{\c_{r}^{t}}}
\right)
\\
= \, &
\frac{1}{p} \cdot
\sum_{t = 0}^{\infty} \Ex_{\eta_{r}^{< t}} \left(
\Indicator_{\Stop > t}
\cdot
\Pr \left( \alpha_{r}^{t} \in \RunPol(\c) \mid \eta_{r}^{< t} \right)
\right)
\\
= \, &
\frac{1}{p} \cdot
\sum_{t = 0}^{\infty} \Ex_{\eta_{r}^{< t}} \left(
\Indicator_{\Stop > t}
\cdot
\Pr \left(
\alpha_{r}^{t} \in \RunPol(\c) \land \Stop > t \mid \eta_{r}^{< t}
\right)
\right)
\\
\leq \, &
\frac{1}{p} \cdot
\sum_{t = 0}^{\infty}
\Ex_{\eta_{r}^{< t}} \left(
\Indicator_{\Stop > t}
\cdot
\Pr \left( \Stop = t + 1 \mid \eta_{r}^{< t} \right)
\right)
\, = \,
\frac{1}{p} \cdot
\sum_{t = 0}^{\infty} \Pr \left( \Stop = t + 1 \right)
\, ,
\end{align*}
where
the second transition follows from the linearity of expectation for infinite
sums,
the third transition holds as both
$\Indicator_{\Stop > t}$
and $\Prp_{\c_{r}^{t}}$ are fully determined by $\eta_{r}^{< t}$;
the fourth transition holds due to the assumption on the propensity of
$\RunPol(\c)$ as
$\Stop > t$
implies that
$\c \Reaching_{\langle \RunPol \rangle} \c_{r}^{t}$;
the fifth transition follows from the definition of the stochastic scheduler;
and
the seventh transition holds by observing that
$[\alpha_{r}^{t} \in \RunPol(\c) \land \Stop > t]
\Longrightarrow
\Stop = t + 1$.
The assertion follows as
$\sum_{t = 0}^{\infty} \Pr(\Stop = t + 1)
=
\Pr(\Stop < \infty)
\leq
1$.\footnote{%
Relying on the assumption that the protocol respects finite density, it is
easy to verify (e.g., using the Borel-Cantelli lemma) that
$\Pr(\Stop < \infty) = 1$,
however this is not necessary for the proof of
\Lem{}~\ref{lemma:runtime:tools:temporal-cost}.}
\end{proof}

\section{The Curse of the Cartesian Product}
\label{appendix:curse-cartesian-product}
In this section, we present an example for the ``Cartesian product curse''
discussed in \Sect{}~\ref{section:configuration-digraph}.
For
$i \in \{ 1, 2 \}$,
consider the density preserving CRN protocol
$\Pi_{i} = (\Species_{i}, \Reactions_{i})$
defined by setting
$\Species_{i} = \{ B_{i}, D_{i}, L_{i}, K_{i}, K'_{i} \}$
and
$\NonVoid(\Reactions_{i}) = \{ \beta_{i}, \gamma_{i}, \gamma'_{i} \}$,
where
\\
$\beta_{i}$:
$L_{i} + B_{i} \rightarrow L_{i} + D_{i}$;
\\
$\gamma_{i}$:
$B_{i} + K_{i} \rightarrow B_{i} + K'_{i}$;
and
\\
$\gamma'_{i}$:
$B_{i} + K'_{i} \rightarrow B_{i} + K_{i}$.
\\
It is easy to verify that any execution of $\Pi_{i}$ that emerges from an
initial configuration
$\c^{0} \in \Naturals^{\Species_{i}}$
with
$\c^{0}(L_{i}) > 0$
is guaranteed to  halt into a configuration
$\c \in \Naturals^{\Species_{i}}$
that satisfies
\\
(1)
$\c(L_{i}) = \c^{0}(L_{i})$;
\\
(2)
$\c(B_{i}) = 0$;
and
\\
(3)
$\c(D_{i}) = \c^{0}(B_{i}) + \c^{0}(D_{i})$.

Let $\Pi_{\times}$ be the Cartesian product of $\Pi_{1}$ and $\Pi_{2}$,
namely, the CRN protocol
$\Pi_{\times} = (\Species_{\times}, \Reactions_{\times})$
defined by setting
\[\textstyle
\Species_{\times}
\, = \,
\Species_{1} \times \Species_{2}
\]
and
\[\textstyle
\Reactions_{\times}
\, = \,
\left\{
\left(
\mathbf{r}_{1} \mathbf{r}_{2}, \mathbf{p}_{1} \mathbf{p}_{2}
\right)
\mid
\left( \mathbf{r}_{1}, \mathbf{p}_{1} \right) \in \Reactions_{1}
\land
\left( \mathbf{r}_{2}, \mathbf{p}_{2} \right) \in \Reactions_{2}
\right\}
\]
It is well known (see, e.g., \cite{AngluinADFP2006computation}) that assuming
strong fairness, every execution $\eta$ of $\Pi_{\times}$
simulates parallel executions of $\Pi_{1}$ and $\Pi_{2}$.
In particular, if $\eta$ emerges from an initial configuration
$\c^{0} \in \Naturals^{\Species_{\times}}$
with
$\c^{0}((L_{1}, \cdot)) > 0$
and
$\c^{0}((\cdot, L_{2})) > 0$,
then it is guaranteed to halt into a configuration
$\c \in \Naturals^{\Species_{\times}}$
that satisfies
\\
(1)
$\c((L_{1}, \cdot)) = \c^{0}((L_{1}, \cdot))$
and
$\c((\cdot, L_{2})) = \c^{0}((\cdot, L_{2}))$;
\\
(2)
$\c((B_{1}, \cdot)) = \c((\cdot, B_{2})) = 0$;
and
\\
(3)
$\c((D_{1}, \cdot)) = \c^{0}((B_{1}, \cdot)) + \c^{0}((D_{1}, \cdot))$
and
$\c((\cdot, D_{2})) = \c^{0}((\cdot, B_{2})) + \c^{0}((\cdot, D_{2}))$.

The situation changes dramatically once we switch to the fairness notion
considered in the current paper:
the weakly fair adversarial scheduler can prevent $\Pi_{\times}$ from halting
at all!
We demonstrate that this is the case by observing that the
configuration set
$S = \{ \c_{1}, \c_{2}, \c_{3}, \c_{4} \}$,
where
\begin{align*}
&
\c_{1}
\, = \,
(L_{1}, K_{2}) + (D_{1}, B_{2}) + (K_{1}, L_{2}) + (B_{1}, D_{2})
\, ,
\\
&
\c_{2}
\, = \,
(L_{1}, K'_{2}) + (D_{1}, B_{2}) + (K_{1}, L_{2}) + (B_{1}, D_{2})
\, ,
\\
&
\c_{3}
\, = \,
(L_{1}, K'_{2}) + (D_{1}, B_{2}) + (K'_{1}, L_{2}) + (B_{1}, D_{2})
\, ,
\quad\text{and}
\\
&
\c_{4}
\, = \,
(L_{1}, K_{2}) + (D_{1}, B_{2}) + (K'_{1}, L_{2}) + (B_{1}, D_{2})
\, ,
\end{align*}
forms a component in the configuration digraph $D^{\Pi_{\times}}$ and that the
component $S$ does not admit any escaping reaction.\footnote{%
We use a configuration with $4$ molecules for simplicity;
it is straightforward to construct such bad examples over configurations with
an arbitrarily large molecular count.}
Indeed, one can construct a weakly fair infinite path $P$ in $D^{\Pi_{\times}}$,
which is trapped in $S$, by repeatedly traversing the cycle depicted in
\Fig{}~\ref{figure:cartesian-product-cycle} (strictly speaking, the path $P$
has to be augmented with applicable void reactions to become weakly fair).
The key here is that although reaction $\beta_{1}$ (resp., $\beta_{2}$) is
continuously applicable in the execution projected from $P$ on $\Pi_{1}$
(resp., $\Pi_{2}$), none of the
$(\beta_{1}, \cdot)$
reactions (resp.,
$(\cdot, \beta_{2})$
reactions) is continuously applicable in $P$ as none of the
$(L_{1}, \cdot)$
species
(resp.,
$(\cdot, L_{2})$
species) is continuously present.

\section{The Random-Walk Broadcast Protocol}
\label{appendix:random-walk}
In this section, we describe a VA protocol, developed by
Angluin et al.~\cite{AngluinAE2008fast} who call it \emph{random-walk
broadcast}, and show that it does not stabilize under a weakly fair adversarial
scheduler.
Fix some sets $\PermVoters_{0}$ and $\PermVoters_{1}$ of permanent $0$- and
$1$-voters, respectively.
The random-walk broadcast protocol
$\Pi_{\RW} = (\Species, \Reactions)$
is defined over
the fluid voter sets
$\FluidVoters_{0} = \{ F_{0} \}$
and
$\FluidVoters_{1} = \{ F_{1} \}$.\footnote{Angluin et al.~\cite{AngluinAE2008fast} present their protocol under the population protocols model.}

The reaction set $\NonVoid(\Reactions)$ of $\Pi_{\RW}$ includes the following non-void
reactions:
\\
$\beta_{v, P_{v}}^{F, 1 - v}$:
$P_{v} + F_{1 - v} \rightarrow P_{v} + F_{v}$
for every
$v \in \{ 0, 1 \}$
and
$P_{v} \in \PermVoters_{v}$;
\\
$\gamma_{0}$:
$F_{0} + F_{1} \rightarrow F_{0} + F_{0}$ w.p. $1 / 2$;
and
\\
$\gamma_{1}$:
$F_{0} + F_{1} \rightarrow F_{1} + F_{1}$ w.p. $1 / 2$.

Under the stochastic scheduler, any reaction between two fluid voters with opposite votes is equally likely to turn them into two $0$- or $1$-fluid voters.
From an initial $v$-voting configuration $\c^{0}$, the number of $v$-voters in $\c^{t}$, $t > 0$, is driven in a random-walk fashion.
However, the permanent voter does not change its vote and produces a bias in the direction of the amplifications of $\c^{0}$.

When dealing with a weakly fair adversarial scheduler, this technique is doomed to failure though; in fact, protocol $\Pi_{\RW}$ does not stabilize into an amplification of $\c^{0}$ under a weakly fair adversarial scheduler.
Assume for simplicity that $\c^{0}$ is $1$-voting and that $\c^{0}(F_{1}) \geq 2$.
The adversarial scheduler sets $\SkipPol$ to be the identity skipping policy $\SkipPol_{\Identity}$ and constructs the execution prefix
$\langle \c^{t}, \zeta^{t} \rangle_{0 \leq t }$
by setting
\[\textstyle
\zeta^{t}
\, = \,
\begin{cases}
\gamma_{1}
\, , &
\c^{t}(F_{0}) > 0
\\
\beta_{0, P_{0}}^{F, 1}
\, , &
\c^{t}(F_{0}) = 0
\end{cases}
\, .
\]
Observe that this is a weakly fair execution that does not stabilize into a
$1$-voting configuration.

\LongVersionEnd 

\clearpage

\begin{figure}[!t]
{\centering
\Large{FIGURES AND TABLES}
\par}
\end{figure}

\begin{table}
\centering
\begin{tabular}{|l|c|c|c|c|}
\hline
predicates &
leaderless &
amplified vote &
stabilization runtime &
halting runtime
\\
\hline\hline
\multirow{4}{*}{\parbox{22mm}{semilinear (non-detection)}} &
yes &
yes &
$\Theta (n)$ &
$\Omega (n)$,
$O (n \log n)$
\\
&
yes &
no &
$\Theta (n)$ &
$\Theta (n)$ \\
&
no &
yes &
$\Theta (n)$ &
$\Omega (n)$,
$O (n \log n)$
\\
&
no &
no &
$\Theta (n)$ &
$\Theta (n)$
\\
\hline
\multirow{4}{*}{detection} &
yes &
yes &
$\Theta (\log n)$ &
$\Theta (\log n)$
\\
&
yes &
no &
$\Theta (\log n)$ &
$\Theta (\log n)$
\\
&
no &
yes &
$\Theta (\log n)$ &
$\Theta (\log n)$
\\
&
no &
no &
$\Theta (\log n)$ &
$\Theta (\log n)$
\\
\hline
\end{tabular}
\caption{\label{table:predicate-decidability-landscape-adversarial}%
The (adversarial) runtime complexity landscape of predicate decidability CRN
protocols operating under the weakly fair adversarial scheduler.
The upper bounds ($O$-notation) hold with a universal quantifier over the
predicate family and an existential quantifier over the CRD family;
the lower bounds ($\Omega$-notation) hold with a universal quantifier over
both the predicate and CRD families.
(As usual,
$\Theta (f(n))$
should be interpreted as both
$O (f(n))$
and
$\Omega (f(n))$.)%
}
\end{table}

\begin{table}
\centering
\begin{tabular}{|l|c|c|c|c|}
\hline
predicates &
leaderless &
amplified vote &
stabilization runtime &
halting runtime
\\
\hline\hline
\multirow{4}{*}{\parbox{22mm}{semilinear (non-eventually constant)}} &
yes &
yes &
$\Theta (n)$ &
$\Theta (n)$
\\
&
yes &
no &
$\Theta (n)$ &
$\Theta (n)$
\\
&
no &
yes &
$\Omega (\log n)$,
$O (n)$ &
$\Omega (\log n)$,
$O (n)$
\\
&
no &
no &
$\Omega (\log n)$,
$O (n)$ &
$\Omega (\log n)$,
$O (n)$
\\
\hline
\multirow{4}{*}{\parbox{22mm}{eventually constant (non-detection)}} &
yes &
yes &
$\Omega (\log n)$,
$O (n)$ &
$\Omega (\log n)$,
$O (n)$
\\
&
yes &
no &
$\Omega (\log n)$,
$O (n)$ &
$\Omega (\log n)$,
$O (n)$
\\
&
no &
yes &
$\Omega (\log n)$,
$O (n)$ &
$\Omega (\log n)$,
$O (n)$
\\
&
no &
no &
$\Omega (\log n)$,
$O (n)$ &
$\Omega (\log n)$,
$O (n)$
\\
\hline
\multirow{4}{*}{\parbox{22mm}{detection}} &
yes &
yes &
$\Theta (\log n)$ &
$\Theta (\log n)$
\\
&
yes &
no &
$\Theta (\log n)$ &
$\Theta (\log n)$
\\
&
no &
yes &
$\Theta (\log n)$ &
$\Theta (\log n)$
\\
&
no &
no &
$\Theta (\log n)$ &
$\Theta (\log n)$
\\
\hline
\end{tabular}
\caption{\label{table:predicate-decidability-landscape-stochastic}%
The (expected stochastic) runtime complexity landscape of predicate
decidability CRN protocols operating under the stochastic scheduler%
\ShortVersion 
\ (refer to the attached full version for details)%
\ShortVersionEnd 
.
The upper bounds ($O$-notation) hold with a universal quantifier over the
predicate family and an existential quantifier over the CRD family;
the lower bounds ($\Omega$-notation) hold with a universal quantifier over
both the predicate and CRD families.
(As usual,
$\Theta (f(n))$
should be interpreted as both
$O (f(n))$
and
$\Omega (f(n))$.)%
}
\end{table}

\begin{figure}
\begin{framed}
\begin{subfigure}{\textwidth}
\centering
\begin{minipage}{0.3\textwidth}
$\Species = \{ A, B, X, X' \}$
\end{minipage}
\hskip 15mm
\begin{minipage}{0.3\textwidth}
$\beta$:
$A + X \rightarrow 2 A$
\\
$\beta'$:
$A + X' \rightarrow 2 A$
\\
$\gamma$:
$A + B \rightarrow 2 A$
\\
$\delta$:
$B + X \rightarrow B + X'$
\\
$\delta'$:
$B + X' \rightarrow B + X$
\end{minipage}
\caption{\label{figure:crn-kill-b:protocol}
The species and non-void reactions of $\Pi$.
A configuration
$\c^{0} \in \Naturals^{\Species}$
is valid as an initial configuration of $\Pi$ if
$\c^{0}(A) = \c^{0}(B) = 1$.
The protocol is designed so that the molecular count of $A$ is non-decreasing,
whereas the molecular counts of $B$ and of
$\{ X, X' \}$
are non-increasing.
Moreover, all weakly fair valid executions of $\Pi$ halt in a configuration
that includes only $A$ molecules.
}
\end{subfigure}%
\vskip 5mm
\begin{subfigure}{\textwidth}
\centering
\includegraphics[scale=0.7]{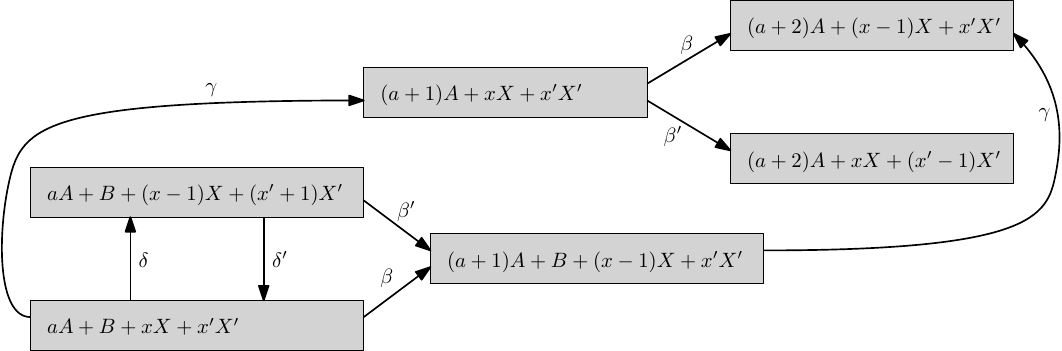}
\caption{\label{figure:crn-kill-b:digraph}%
Part of the configuration digraph $D^{\Pi}$ (excluding void reactions).
The configuration set
$S = \{ a A + B + (x + z) X + (x' - z) X' \mid -x \leq z \leq x' \}$
forms a component of $D^{\Pi}$.
Since reactions $\beta$ and $\beta'$ are inapplicable in configurations
$a A + B + (x + x') X'$
and
$a A + B + (x + x') X$,
respectively, it follows that these two reactions do not escape from $S$.
Reaction $\gamma$ on the other hand does escape from $S$, ensuring that a
weakly fair execution cannot remain in $S$ indefinitely.
After the $B$ molecule is consumed (by a $\gamma$ reaction), each
configuration constitutes a (singleton) component of $D^{\Pi}$ and every
applicable non-void reaction is escaping.
}
\end{subfigure}%
\vskip 5mm
\begin{subfigure}{\textwidth}
\centering
$\displaystyle
\RunPol(\c)
=
\begin{cases}
\{ \beta, \beta', \gamma \}
, &
\c(B) > 0
\\
\{ \beta, \beta' \}
, &
\c(B) = 0
\end{cases}$
\caption{\label{figure:crn-kill-b:runtime-policy}%
A  runtime policy $\RunPol$ for $\Pi$.
Under $\RunPol$, a round with effective configuration
$\EffConf = a A + B + x X + x' X'$
is target-accomplished and ends upon scheduling one of the reactions
$\beta, \beta', \gamma$;
in particular, it is guaranteed that the next effective configuration
$\EffConf'$ satisfies
$\EffConf'(A) > a$
(no matter what the adversarial skipping policy is), i.e., progress is made.
The adversary may opt to schedule reactions $\delta$ and $\delta'$ many times
before the round ends, however the crux of our runtime definition is that this
does not affect the round's temporal cost.
Specifically, since
$\Prp_{\c}(\{ \beta, \beta', \gamma \})
=
\frac{a (x + x' + 1)}{\Vol}$
for every
$\c \in S$
(recall the definition of component $S$ from
\Fig{}~\ref{figure:crn-kill-b:digraph}), it follows, by a simple probabilistic
argument%
\LongVersion 
\ (see \Lem{}~\ref{lemma:runtime:tools:temporal-cost})%
\LongVersionEnd 
, that
$\TemporalCost^{\RunPol}(\EffConf)
=
\frac{\Vol}{a (x + x' + 1)}$.
A similar (simpler in fact) argument leads to the conclusion that if
$\EffConf = a A + x X + x' X'$,
then
$\TemporalCost^{\RunPol}(\EffConf)
=
\frac{\Vol}{a (x + x')}$.
This allows us to show that
$\RunTime_{\Halt}^{\Pi}(n) = O (\log n)$.
}
\end{subfigure}%
\caption{\label{figure:crn-kill-b}%
A CRN protocol
$\Pi = (\Species, \Reactions)$
demonstrating how a carefully chosen runtime policy guarantees significant
progress in each round while up-bounding the round's temporal cost.
}
\end{framed}
\end{figure}

\begin{figure}
\begin{framed}
\begin{subfigure}{\textwidth}
\centering
\begin{minipage}{0.4\textwidth}
$\Species = \{ L_{0}, L_{1}, \dots, L_{k + 1}, C, X, Y \}$
\end{minipage}
\hskip 10mm
\begin{minipage}{0.5\textwidth}
$\beta_{i}$: $L_{i} + X \rightarrow L_{0} + Y$ for $0 \leq i \leq k + 1$, $i \neq k$
\\
$\gamma_{i}$: $L_{i} + C \rightarrow L_{i + 1} + C$ for $0 \leq i \leq k$
\end{minipage}
\caption{\label{figure:crn-multiple-pitfalls:protocol}%
The species and non-void reactions of $\Pi$, where $k$ is an arbitrarily large
constant.
A configuration
$\c^{0} \in \Naturals^{\Species}$
is valid as an initial configuration of $\Pi$ if
$\c^{0}(L_{0}) = 1$,
$\c^{0}(C) = 1$,
and
$\c^{0}(\{ L_{1}, \dots, L_{k + 1} \}) = 0$.
The protocol is designed so that
$\mathbf{c}(\{ L_{0}, L_{1}, \dots,  L_{k + 1} \}) = 1$
for any configuration $\mathbf{c}$ reachable from a valid initial
configuration (i.e.,
$L_{0}, L_{1}, \dots,  L_{k + 1}$
are ``leader species'').
Species $C$ is a catalyst for any reaction it participates in and
$\mathbf{c}(C) = 1$
for any configuration $\mathbf{c}$ reachable from a valid initial
configuration.
The execution progresses by shifting all $X$ molecules into $Y$ molecules.
We are interested in the stabilization of $\Pi$'s executions into the (set of)
configurations
$\c \in \Naturals^{\Species}$
satisfying
$\c(X) < \c(Y)$,
although the executions actually halt once
$\c(X) = 0$
(and
$\c(L_{k + 1}) = 1$).
}
\end{subfigure}%
\vskip 5mm
\begin{subfigure}{\textwidth}
\centering
\includegraphics[scale=0.7]{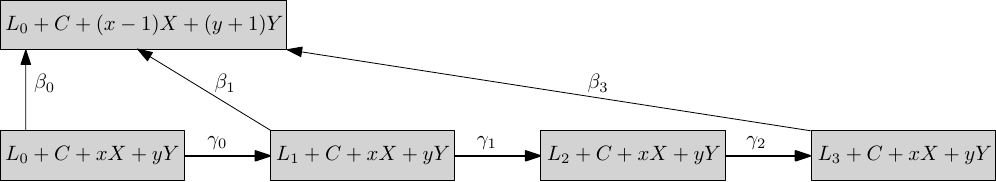}
\caption{\label{figure:crn-multiple-pitfalls:digraph}%
Part of the configuration digraph $D^{\Pi}$ (excluding void reactions) for
$k = 2$.
Notice that each configuration constitutes a (singleton) component of
$D^{\Pi}$ and every applicable non-void reaction is escaping.
}
\end{subfigure}
\vskip 5mm
\begin{subfigure}{\textwidth}
\centering
$\displaystyle
\RunPol(\c)
=
\Applicable(\c) \cap \NonVoid(\Reactions)
$
\caption{\label{figure:crn-multiple-pitfalls:runtime-policy}%
A runtime policy $\RunPol$ for $\Pi$.
Under $\RunPol$, every round ends once any non-void reaction is applied to the
round's effective configuration.
In %
\LongVersion 
\Sect{}~\ref{section:large-adversarial-small-stochastic:multiple-pitfalls}%
\LongVersionEnd 
\ShortVersion 
the attached full version%
\ShortVersionEnd 
, we show that
$\RunTime_{\Stab}^{\RunPol, \SkipPol}(\eta) \leq O (n^{2})$
for any skipping policy $\SkipPol$ and weakly fair valid execution $\eta$ of
initial molecular count $n$.
It turns out that this bound is tight:
The (weakly fair) adversarial scheduler can generate the execution
$\eta = \langle \c^{t}, \alpha^{t} \rangle$
by scheduling
$\alpha^{t} = \gamma_{i}$
if
$\c^{t}(L_{i}) = 1$
for some
$0 \leq i \leq k$;
and
$\alpha^{t} = \beta_{k + 1}$
if
$\c^{t}(L_{k + 1}) = 1$.
Using the identity skipping policy $\SkipPol_{\Identity}$ and assuming that
$\c^{0}(X) = x_{0}$
and
$\c^{0}(Y) = 0$,
it is easy to show
(see
\LongVersion 
\Sect{}~\ref{section:large-adversarial-small-stochastic:multiple-pitfalls}%
\LongVersionEnd 
\ShortVersion 
the attached full version%
\ShortVersionEnd 
) that $\eta$ visits the configuration
$\c_{y} = L_{k} + C + (x_{0} - y) X + y Y$
for every
$0 \leq y \leq x_{0} / 2$
before it stabilizes and that each such visit constitutes the effective
configuration of the corresponding round, regardless of the runtime policy
$\RunPol'$.
Since each such configuration $\c_{y}$ is a $2$-pitfall (recall the definition
from \Sect{}~\ref{section:runtime:speed-faults}), we deduce that
$\TemporalCost^{\RunPol'}(\c_{y}) = \Omega (n)$,
which sums up to
$\RunTime_{\Stab}^{\RunPol', \SkipPol_{\Identity}}(\eta) = \Omega (n^{2})$.
The interesting aspect of protocol $\Pi$ is that with high probability, a
stochastic execution stabilizes without visiting the pitfall configurations
$\c_{y}$ even once, which allows us to conclude that the expected stochastic
stabilization runtime of $\Pi$ is
$O (n)$
---
see
\LongVersion 
\Sect{}~\ref{section:large-adversarial-small-stochastic:multiple-pitfalls}
\LongVersionEnd 
\ShortVersion 
the attached full version
\ShortVersionEnd 
for details.
}
\end{subfigure}
\caption{\label{figure:crn-multiple-pitfalls}%
A CRN protocol
$\Pi = (\Species, \Reactions)$
demonstrating that the adversarial stabilization runtime may be significantly
larger than the expected stochastic runtime due to (asymptotically many)
pitfall configurations.}
\end{framed}
\end{figure}

\begin{figure}
\begin{framed}
\begin{subfigure}{\textwidth}
\centering
\begin{minipage}{0.3\textwidth}
$\Species = \{ L_{0}, L_{1}, X, Y \}$
\end{minipage}
\hskip 15mm
\begin{minipage}{0.3\textwidth}
$\beta_{0}$:
$L_{0} + Y \rightarrow 2 Y$
\\
$\beta_{1}$:
$L_{1} + Y \rightarrow 2 Y$
\\
$\gamma_{0}$:
$L_{0} + X \rightarrow L_{1} + Y$
\\
$\gamma_{1}$:
$L_{1} + X \rightarrow L_{0} + Y$
\end{minipage}
\caption{\label{figure:crn-kill-alternating-leaders:protocol}%
The species and non-void reactions of $\Pi$.
A configuration
$\c^{0} \in \Naturals^{\Species}$
is valid as an initial configuration of $\Pi$ if
$\c^{0}(\{ L_{0}, L_{1} \}) = 1$
and
$\c^{0}(Y) \geq \c^{0}(X)$.
The protocol is designed so that the molecular counts of
$\{ L_{0}, L_{1} \}$
and of $X$ are non-increasing, whereas the molecular count of $Y$ is
non-decreasing.
Moreover, all weakly fair valid executions of $\Pi$ halt in a configuration
that includes zero $L_{0}$ and $L_{1}$ molecules.
}
\end{subfigure}%
\vskip 5mm
\begin{subfigure}{\textwidth}
\centering
\includegraphics[scale=0.7]{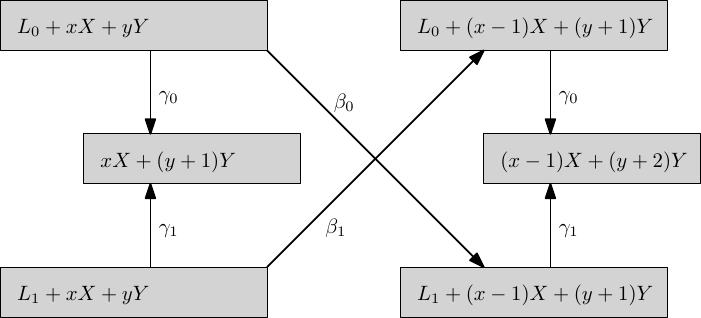}
\caption{\label{figure:crn-kill-alternating-leaders:digraph}%
Part of the configuration digraph $D^{\Pi}$ (excluding void reactions).
Notice that each configuration constitutes a (singleton) component of
$D^{\Pi}$ and every applicable non-void reaction is escaping.
Furthermore, if
$\c \LeadsTo \c'$,
then
$\Applicable(\c) \cap \Applicable(\c') \cap \NonVoid(\Reactions) = \emptyset$.
}
\end{subfigure}
\vskip 5mm
\begin{subfigure}{\textwidth}
\centering
$\displaystyle
\RunPol(\c)
=
\begin{cases}
\beta_{0}
, &
\c(L_{0}) = 1
\\
\beta_{1}
, &
\c(L_{1}) = 1
\end{cases}$
\caption{\label{figure:crn-kill-alternating-leaders:runtime-policy}%
A runtime policy $\RunPol$ for $\Pi$.
Under $\RunPol$, a round whose (non-halting) effective configuration is
$\EffConf$ ends once the execution reaches any configuration
$\c \neq \EffConf$.
This may result from scheduling the target reaction $\RunPol(\EffConf)$, in
which case the round is target-accomplished;
otherwise, the round ends because $\RunPol(\EffConf)$ is inapplicable in
$\c$, in which case the round is target-deprived.
Since
$\Prp_{L_{i} + x X + y Y}(\beta_{i})
=
y / \Vol$
for
$i \in \{ 0, 1 \}$,
it follows, by a simple probabilistic argument%
\LongVersion 
\ (see \Lem{}~\ref{lemma:runtime:tools:temporal-cost})%
\LongVersionEnd 
, that
$\TemporalCost^{\RunPol}(L_{i} + x X + y Y)
=
\Vol / y$.
This allows us to show that
$\RunTime_{\Halt}^{\Pi}(n) = O (n)$.
In
\LongVersion 
\Sect{}~\ref{section:large-adversarial-small-stochastic}%
\LongVersionEnd 
\ShortVersion 
the attached full version%
\ShortVersionEnd 
, it is shown that the
$O (n)$
halting runtime bound is (asymptotically) tight under the (weakly fair)
adversarial scheduler;
in contrast, the expected stochastic runtime of $\Pi$ under the stochastic
scheduler is only
$O (1)$.
This gap holds despite the fact that the executions that realize the
$\Omega (n)$
runtime lower bound do not reach a pitfall configuration.
}
\end{subfigure}
\caption{\label{figure:crn-kill-alternating-leaders}
A CRN protocol
$\Pi = (\Species, \Reactions)$
demonstrating that the adversarial runtime may be significantly larger than
the expected stochastic runtime even though the protocol does not admit a speed
fault.}
\end{framed}
\end{figure}

\begin{figure}
\begin{framed}
\begin{subfigure}{\textwidth}
\centering
\begin{minipage}{0.3\textwidth}
$\Species = \{ A, B, E, X \}$
\end{minipage}
\hskip 10mm
\begin{minipage}{0.5\textwidth}
$\beta$:
$A + A \rightarrow 2 E$
\\
$\gamma$:
$X + X \rightarrow 2 E$
\\
$\delta$:
$B + X \rightarrow 2 B$
\\
$\chi_{Z}$:
$E + Z \rightarrow 2 E$
$\quad$ for every
$Z \in \{ A, B, X \}$
\end{minipage}
\caption{\label{figure:crn-skipping-policy:protocol}%
The species and non-void reactions of $\Pi$.
A configuration
$\c^{0} \in \Naturals^{\Species}$
is valid as an initial configuration of $\Pi$ if
$\c^{0}(A) = 2$,
$\c^{0}(B) = 1$,
and
$\c^{0}(E) = 0$.
The protocol is designed so that once the first (two) $E$ molecules are
produced, species $E$ spreads by means of the $\chi$ reactions until it takes
over the entire test tube.
The first $E$ molecules can be produced either through reaction $\beta$, that
remains applicable as long as $E$ was not produced, or through reaction
$\gamma$, that may become inapplicable if reaction $\delta$ is scheduled many
times.
In any case, all weakly fair valid executions of $\Pi$ halt in a configuration
that includes only $E$ molecules.
}
\end{subfigure}%
\vskip 5mm
\begin{subfigure}{\textwidth}
\centering
\includegraphics[scale=0.7]{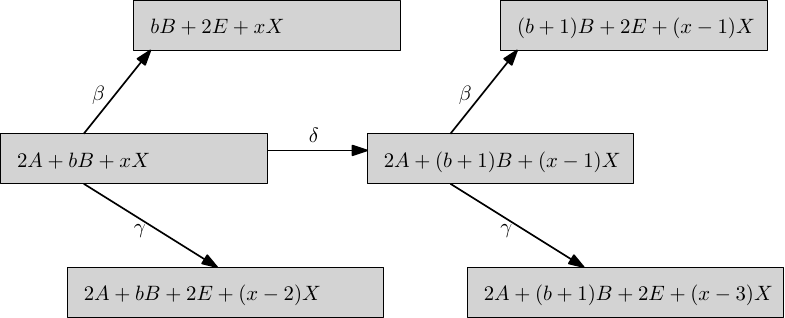}
\caption{\label{figure:crn-skipping-policy:digraph}%
Part of the configuration digraph $D^{\Pi}$ (excluding void reactions).
Notice that each configuration constitutes a (singleton) component of
$D^{\Pi}$ and every applicable non-void reaction is escaping.
}
\end{subfigure}
\vskip 5mm
\begin{subfigure}{\textwidth}
\centering
$\displaystyle
\RunPol(\c)
=
\begin{cases}
\{ \beta, \gamma \}
, &
\c(E) = 0
\\
\{ \chi_{A}, \chi_{B}, \chi_{X} \}
, &
\c(E) > 0
\end{cases}$
\caption{\label{figure:crn-skipping-policy:runtime-policy}%
A runtime policy $\RunPol$ for $\Pi$.
Under $\RunPol$, a round whose (non-halting) effective configuration
$\EffConf$ satisfies
$\EffConf(E) = 0$
ends once (two) $E$ molecules are produced, either through reaction $\beta$ or
through reaction $\gamma$ (either way, the round is target accomplished);
in particular, there can be at most one such round, that is, the first round
of the execution.
To maximize the temporal cost charged for this round, the adversarial
scheduler devises (the execution and) the skipping policy $\SkipPol$ so that
$\EffConf(X) < 2$
which means that reaction $\gamma$ is inapplicable in $\EffConf$
(this requires that $\SkipPol$ generates a ``large skip'').
Such a configuration $\EffConf$ is a halting $2$-pitfall as reaction $\beta$
must be scheduled in order to advance the execution and the propensity of
$\beta$ is
$2 / \Vol$.
We conclude, by a simple probabilistic argument%
\LongVersion 
\ (see \Lem{}~\ref{lemma:runtime:tools:temporal-cost})%
\LongVersionEnd 
, that
$\TemporalCost^{\RunPol}(\EffConf) = \Vol / 2$.
A round whose (non-halting) effective configuration $\EffConf$ satisfies
$\EffConf(E) > 0$
ends once a $\chi$ configuration is scheduled, thus ensuring that the
execution's next effective configuration $\EffConf'$ satisfies
$\EffConf'(E) > \EffConf(E)$.
Using standard arguments, one can prove that the total contribution of (all)
these rounds to the halting runtime of an execution $\eta$ with initial
molecular count $n$ is up-bounded by
$O (\log n)$.
Therefore, together with the contribution of the first ``slow'' round, we get
$\RunTime_{\Halt}^{\RunPol, \SkipPol}(\eta) = \Theta (n)$.
Note that the same runtime policy $\RunPol$ leads to a much better halting
runtime of
$\RunTime_{\Halt}^{\RunPol, \SkipPol_{\Identity}}(\eta) = \Theta (\log n)$
if the adversarial scheduler opts to use the identity skipping policy
$\SkipPol_{\Identity}$ instead of the aforementioned skipping policy
$\SkipPol$.
As a direct consequence of \Lem{}~\ref{lemma:speed-fault-lower-bound}, we
conclude that $\RunPol$ is asymptotically optimal for $\Pi$, hence
$\RunTime_{\Halt}^{\Pi}(n) = \Theta (n)$.
}
\end{subfigure}
\caption{\label{figure:crn-skipping-policy}%
A CRN protocol
$\Pi = (\Species, \Reactions)$
demonstrating that a non-trivial skipping policy results in a significantly
larger runtime, compared to the identity skipping policy.}
\end{framed}
\end{figure}

\LongVersion 
\begin{figure}
{\centering
\includegraphics[scale=0.8]{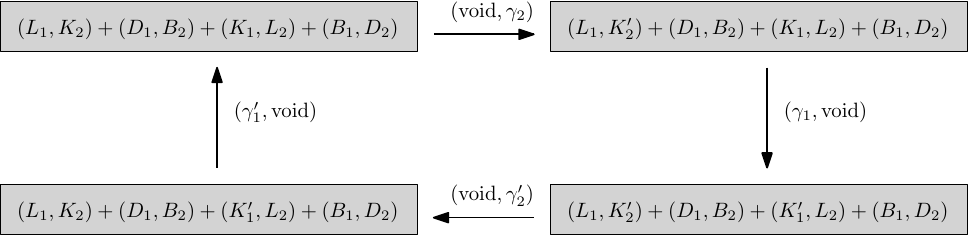}
\par}
\caption{\label{figure:cartesian-product-cycle}%
A cycle through the configurations
$\c_{1}, \c_{2}, \c_{3}, \c_{4}$
in the configuration digraph $D^{\Pi_{\times}}$.}
\end{figure}
\LongVersionEnd 
\end{document}